\documentclass{LMCS}

\usepackage{enumerate}

\usepackage{hyperref}
\usepackage{proof}

\usepackage{graphicx}

\newcommand{\idx}[2][i]{\ensuremath {#1 \in \{1,\ldots,\mbox{$#2$}\}}}

\newcommand{\ua}{\ensuremath{\underline a}}

\newcommand{\ud}{\ensuremath{\underline d}}
\newcommand{\ue}{\ensuremath{\underline e}}
\newcommand{\ug}{\ensuremath{\underline g}}

\newcommand{\ui}{\ensuremath{\underline i}}
\newcommand{\uj}{\ensuremath{\underline j}}
\newcommand{\uk}{\ensuremath{\underline k}}
\newcommand{\ut}{\ensuremath{\underline t}}
\newcommand{\ux}{\ensuremath{\underline x}}

\newcommand{\ul}{\ensuremath{\underline l}}

\newcommand{\cA}{\ensuremath \mathcal A}
\newcommand{\cB}{\ensuremath \mathcal B}
\newcommand{\cM}{\ensuremath \mathcal M}
\newcommand{\cN}{\ensuremath \mathcal N}

\newcommand{\cSi}{\ensuremath \mathcal S}
\newcommand{\Cc}{\ensuremath \mathcal C}
\newcommand{\lra}{\longrightarrow}

\renewcommand{\int}{\ensuremath {\mathcal I}}

\newcommand{\mywidehat}[1]{\ensuremath{\lbrack\!\lbrack #1 \rbrack\!\rbrack}}

\theoremstyle{plain}\newtheorem{assumption}[thm]{Assumption}
\theoremstyle{plain}\newtheorem{proposition}[thm]{Proposition}
\theoremstyle{plain}\newtheorem{property}[thm]{Property}
\theoremstyle{plain}\newtheorem{example}[thm]{Example}
\theoremstyle{plain}
\theoremstyle{plain}\newtheorem{lemma}[thm]{Lemma}

\long\def\ig#1{\relax}
\ig{Thanks to Roberto Minio for this def'n.  Compare the def'n of
\comment in AMSTeX.}
 
\newcount \coefa
\newcount \coefb
\newcount \coefc
\newcount\tempcounta
\newcount\tempcountb
\newcount\tempcountc
\newcount\tempcountd
\newcount\xext
\newcount\yext
\newcount\xoff
\newcount\yoff
\newcount\gap%
\newcount\arrowtypea
\newcount\arrowtypeb
\newcount\arrowtypec
\newcount\arrowtyped
\newcount\arrowtypee
\newcount\height
\newcount\width
\newcount\xpos
\newcount\ypos
\newcount\run
\newcount\rise
\newcount\arrowlength
\newcount\halflength
\newcount\arrowtype
\newdimen\tempdimen
\newdimen\xlen
\newdimen\ylen
\newsavebox{\tempboxa}%
\newsavebox{\tempboxb}%
\newsavebox{\tempboxc}%
 
\makeatletter
\def\settypes(#1,#2,#3){\arrowtypea#1 \arrowtypeb#2 \arrowtypec#3}
\def\settoheight#1#2{\setbox\@tempboxa\hbox{#2}#1\ht\@tempboxa\relax}%
\def\settodepth#1#2{\setbox\@tempboxa\hbox{#2}#1\dp\@tempboxa\relax}%
\def\settokens[#1`#2`#3`#4]{%
     \def\tokena{#1}\def\tokenb{#2}\def\tokenc{#3}\def\tokend{#4}}
\def\setsqparms[#1`#2`#3`#4;#5`#6]{%
\arrowtypea #1
\arrowtypeb #2
\arrowtypec #3
\arrowtyped #4
\width #5
\height #6
}
\def\setpos(#1,#2){\xpos=#1 \ypos#2}
 
\def\bfig{\begin{picture}(\xext,\yext)(\xoff,\yoff)}
\def\efig{\end{picture}}
 
\def\putbox(#1,#2)#3{\put(#1,#2){\makebox(0,0){$#3$}}}
 
\def\settriparms[#1`#2`#3;#4]{\settripairparms[#1`#2`#3`1`1;#4]}%
 
\def\settripairparms[#1`#2`#3`#4`#5;#6]{%
\arrowtypea #1
\arrowtypeb #2
\arrowtypec #3
\arrowtyped #4
\arrowtypee #5
\width #6
\height #6
}
 
\def\resetparms{\settripairparms[1`1`1`1`1;500]\width 500}
 
\resetparms
 
\def\mvector(#1,#2)#3{
\put(0,0){\vector(#1,#2){#3}}%
\put(0,0){\vector(#1,#2){30}}%
}
\def\evector(#1,#2)#3{{
\arrowlength #3
\put(0,0){\vector(#1,#2){\arrowlength}}%
\advance \arrowlength by-30
\put(0,0){\vector(#1,#2){\arrowlength}}%
}}
 
\def\horsize#1#2{%
\settowidth{\tempdimen}{$#2$}%
#1=\tempdimen
\divide #1 by\unitlength
}
 
\def\vertsize#1#2{%
\settoheight{\tempdimen}{$#2$}%
#1=\tempdimen
\settodepth{\tempdimen}{$#2$}%
\advance #1 by\tempdimen
\divide #1 by\unitlength
}
 
\def\vertadjust[#1`#2`#3]{%
\vertsize{\tempcounta}{#1}%
\vertsize{\tempcountb}{#2}%
\ifnum \tempcounta<\tempcountb \tempcounta=\tempcountb \fi
\divide\tempcounta by2
\vertsize{\tempcountb}{#3}%
\ifnum \tempcountb>0 \advance \tempcountb by20 \fi
\ifnum \tempcounta<\tempcountb \tempcounta=\tempcountb \fi
}
 
\def\horadjust[#1`#2`#3]{%
\horsize{\tempcounta}{#1}%
\horsize{\tempcountb}{#2}%
\ifnum \tempcounta<\tempcountb \tempcounta=\tempcountb \fi
\divide\tempcounta by20
\horsize{\tempcountb}{#3}%
\ifnum \tempcountb>0 \advance \tempcountb by60 \fi
\ifnum \tempcounta<\tempcountb \tempcounta=\tempcountb \fi
}
 
\ig{ In this procedure, #1 is the paramater that sticks out all the way,
#2 sticks out the least and #3 is a label sticking out half way.  #4 is
the amount of the offset.}
 
\def\sladjust[#1`#2`#3]#4{%
\tempcountc=#4
\horsize{\tempcounta}{#1}%
\divide \tempcounta by2
\horsize{\tempcountb}{#2}%
\divide \tempcountb by2
\advance \tempcountb by-\tempcountc
\ifnum \tempcounta<\tempcountb \tempcounta=\tempcountb\fi
\divide \tempcountc by2
\horsize{\tempcountb}{#3}%
\advance \tempcountb by-\tempcountc
\ifnum \tempcountb>0 \advance \tempcountb by80\fi
\ifnum \tempcounta<\tempcountb \tempcounta=\tempcountb\fi
\advance\tempcounta by20
}
 
\def\putvector(#1,#2)(#3,#4)#5#6{{%
\xpos=#1
\ypos=#2
\run=#3
\rise=#4
\arrowlength=#5
\arrowtype=#6
\ifnum \arrowtype<0
    \ifnum \run=0
        \advance \ypos by-\arrowlength
    \else
        \tempcounta \arrowlength
        \multiply \tempcounta by\rise
        \divide \tempcounta by\run
        \ifnum\run>0
            \advance \xpos by\arrowlength
            \advance \ypos by\tempcounta
        \else
            \advance \xpos by-\arrowlength
            \advance \ypos by-\tempcounta
        \fi
    \fi
    \multiply \arrowtype by-1
    \multiply \rise by-1
    \multiply \run by-1
\fi
\ifnum \arrowtype=1
    \put(\xpos,\ypos){\vector(\run,\rise){\arrowlength}}%
\else\ifnum \arrowtype=2
    \put(\xpos,\ypos){\mvector(\run,\rise)\arrowlength}%
\else\ifnum\arrowtype=3
    \put(\xpos,\ypos){\evector(\run,\rise){\arrowlength}}%
\fi\fi\fi
}}
 
\def\putsplitvector(#1,#2)#3#4{
\xpos #1
\ypos #2
\arrowtype #4
\halflength #3
\arrowlength #3
\gap 140
\advance \halflength by-\gap
\divide \halflength by2
\ifnum \arrowtype=1
    \put(\xpos,\ypos){\line(0,-1){\halflength}}%
    \advance\ypos by-\halflength
    \advance\ypos by-\gap
    \put(\xpos,\ypos){\vector(0,-1){\halflength}}%
\else\ifnum \arrowtype=2
    \put(\xpos,\ypos){\line(0,-1)\halflength}%
    \put(\xpos,\ypos){\vector(0,-1)3}%
    \advance\ypos by-\halflength
    \advance\ypos by-\gap
    \put(\xpos,\ypos){\vector(0,-1){\halflength}}%
\else\ifnum\arrowtype=3
    \put(\xpos,\ypos){\line(0,-1)\halflength}%
    \advance\ypos by-\halflength
    \advance\ypos by-\gap
    \put(\xpos,\ypos){\evector(0,-1){\halflength}}%
\else\ifnum \arrowtype=-1
    \advance \ypos by-\arrowlength
    \put(\xpos,\ypos){\line(0,1){\halflength}}%
    \advance\ypos by\halflength
    \advance\ypos by\gap
    \put(\xpos,\ypos){\vector(0,1){\halflength}}%
\else\ifnum \arrowtype=-2
    \advance \ypos by-\arrowlength
    \put(\xpos,\ypos){\line(0,1)\halflength}%
    \put(\xpos,\ypos){\vector(0,1)3}%
    \advance\ypos by\halflength
    \advance\ypos by\gap
    \put(\xpos,\ypos){\vector(0,1){\halflength}}%
\else\ifnum\arrowtype=-3
    \advance \ypos by-\arrowlength
    \put(\xpos,\ypos){\line(0,1)\halflength}%
    \advance\ypos by\halflength
    \advance\ypos by\gap
    \put(\xpos,\ypos){\evector(0,1){\halflength}}%
\fi\fi\fi\fi\fi\fi
}
 
\def\putmorphism(#1)(#2,#3)[#4`#5`#6]#7#8#9{{%
\run #2
\rise #3
\ifnum\rise=0
  \puthmorphism(#1)[#4`#5`#6]{#7}{#8}{#9}%
\else\ifnum\run=0
  \putvmorphism(#1)[#4`#5`#6]{#7}{#8}{#9}%
\else
\setpos(#1)%
\arrowlength #7
\arrowtype #8
\ifnum\run=0
\else\ifnum\rise=0
\else
\ifnum\run>0
    \coefa=1
\else
   \coefa=-1
\fi
\ifnum\arrowtype>0
   \coefb=0
   \coefc=-1
\else
   \coefb=\coefa
   \coefc=1
   \arrowtype=-\arrowtype
\fi
\width=2
\multiply \width by\run
\divide \width by\rise
\ifnum \width<0  \width=-\width\fi
\advance\width by60
\if l#9 \width=-\width\fi
\putbox(\xpos,\ypos){#4}
{\multiply \coefa by\arrowlength
\advance\xpos by\coefa
\multiply \coefa by\rise
\divide \coefa by\run
\advance \ypos by\coefa
\putbox(\xpos,\ypos){#5} }%
{\multiply \coefa by\arrowlength
\divide \coefa by2
\advance \xpos by\coefa
\advance \xpos by\width
\multiply \coefa by\rise
\divide \coefa by\run
\advance \ypos by\coefa
\if l#9%
   \put(\xpos,\ypos){\makebox(0,0)[r]{$#6$}}%
\else\if r#9%
   \put(\xpos,\ypos){\makebox(0,0)[l]{$#6$}}%
\fi\fi }%
{\multiply \rise by-\coefc
\multiply \run by-\coefc
\multiply \coefb by\arrowlength
\advance \xpos by\coefb
\multiply \coefb by\rise
\divide \coefb by\run
\advance \ypos by\coefb
\multiply \coefc by70
\advance \ypos by\coefc
\multiply \coefc by\run
\divide \coefc by\rise
\advance \xpos by\coefc
\multiply \coefa by140
\multiply \coefa by\run
\divide \coefa by\rise
\advance \arrowlength by\coefa
\ifnum \arrowtype=1
   \put(\xpos,\ypos){\vector(\run,\rise){\arrowlength}}%
\else\ifnum\arrowtype=2
   \put(\xpos,\ypos){\mvector(\run,\rise){\arrowlength}}%
\else\ifnum\arrowtype=3
   \put(\xpos,\ypos){\evector(\run,\rise){\arrowlength}}%
\fi\fi\fi}\fi\fi\fi\fi}}
 
\def\puthmorphism(#1,#2)[#3`#4`#5]#6#7#8{{%
\xpos #1
\ypos #2
\width #6
\arrowlength #6
\putbox(\xpos,\ypos){#3\vphantom{#4}}%
{\advance \xpos by\arrowlength
\putbox(\xpos,\ypos){\vphantom{#3}#4}}%
\horsize{\tempcounta}{#3}%
\horsize{\tempcountb}{#4}%
\divide \tempcounta by2
\divide \tempcountb by2
\advance \tempcounta by30
\advance \tempcountb by30
\advance \xpos by\tempcounta
\advance \arrowlength by-\tempcounta
\advance \arrowlength by-\tempcountb
\putvector(\xpos,\ypos)(1,0){\arrowlength}{#7}%
\divide \arrowlength by2
\advance \xpos by\arrowlength
\vertsize{\tempcounta}{#5}%
\divide\tempcounta by2
\advance \tempcounta by20
\if a#8 %
   \advance \ypos by\tempcounta
   \putbox(\xpos,\ypos){#5}%
\else
   \advance \ypos by-\tempcounta
   \putbox(\xpos,\ypos){#5}%
\fi}}
 
\def\putvmorphism(#1,#2)[#3`#4`#5]#6#7#8{{%
\xpos #1
\ypos #2
\arrowlength #6
\arrowtype #7
\settowidth{\xlen}{$#5$}%
\putbox(\xpos,\ypos){#3}%
{\advance \ypos by-\arrowlength
\putbox(\xpos,\ypos){#4}}%
{\advance\arrowlength by-140
\advance \ypos by-70
\ifdim\xlen>0pt
   \if m#8%
      \putsplitvector(\xpos,\ypos){\arrowlength}{\arrowtype}%
   \else
      \putvector(\xpos,\ypos)(0,-1){\arrowlength}{\arrowtype}%
   \fi
\else
   \putvector(\xpos,\ypos)(0,-1){\arrowlength}{\arrowtype}%
\fi}%
\ifdim\xlen>0pt
   \divide \arrowlength by2
   \advance\ypos by-\arrowlength
   \if l#8%
      \advance \xpos by-40
      \put(\xpos,\ypos){\makebox(0,0)[r]{$#5$}}%
   \else\if r#8%
      \advance \xpos by40
      \put(\xpos,\ypos){\makebox(0,0)[l]{$#5$}}%
   \else
      \putbox(\xpos,\ypos){#5}%
   \fi\fi
\fi
}}
 
\def\topadjust[#1`#2`#3]{%
\yoff=10
\vertadjust[#1`#2`{#3}]%
\advance \yext by\tempcounta
\advance \yext by 10
}
\def\botadjust[#1`#2`#3]{%
\vertadjust[#1`#2`{#3}]%
\advance \yext by\tempcounta
\advance \yoff by-\tempcounta
}
\def\leftadjust[#1`#2`#3]{%
\xoff=0
\horadjust[#1`#2`{#3}]%
\advance \xext by\tempcounta
\advance \xoff by-\tempcounta
}
\def\rightadjust[#1`#2`#3]{%
\horadjust[#1`#2`{#3}]%
\advance \xext by\tempcounta
}
\def\rightsladjust[#1`#2`#3]{%
\sladjust[#1`#2`{#3}]{\width}%
\advance \xext by\tempcounta
}
\def\leftsladjust[#1`#2`#3]{%
\xoff=0
\sladjust[#1`#2`{#3}]{\width}%
\advance \xext by\tempcounta
\advance \xoff by-\tempcounta
}
\def\adjust[#1`#2;#3`#4;#5`#6;#7`#8]{%
\topadjust[#1``{#2}]
\leftadjust[#3``{#4}]
\rightadjust[#5``{#6}]
\botadjust[#7``{#8}]}
 
\def\putsquarep<#1>(#2)[#3;#4`#5`#6`#7]{{%
\setsqparms[#1]%
\setpos(#2)%
\settokens[#3]%
\puthmorphism(\xpos,\ypos)[\tokenc`\tokend`{#7}]{\width}{\arrowtyped}b%
\advance\ypos by \height
\puthmorphism(\xpos,\ypos)[\tokena`\tokenb`{#4}]{\width}{\arrowtypea}a%
\putvmorphism(\xpos,\ypos)[``{#5}]{\height}{\arrowtypeb}l%
\advance\xpos by \width
\putvmorphism(\xpos,\ypos)[``{#6}]{\height}{\arrowtypec}r%
}}
 
\def\putsquare{\@ifnextchar <{\putsquarep}{\putsquarep%
   <\arrowtypea`\arrowtypeb`\arrowtypec`\arrowtyped;\width`\height>}}
\def\square{\@ifnextchar< {\squarep}{\squarep
   <\arrowtypea`\arrowtypeb`\arrowtypec`\arrowtyped;\width`\height>}}
\def\squarep<#1>[#2`#3`#4`#5;#6`#7`#8`#9]{{
\setsqparms[#1]
\xext=\width                                          
\yext=\height                                         
\topadjust[#2`#3`{#6}]
\botadjust[#4`#5`{#9}]
\leftadjust[#2`#4`{#7}]
\rightadjust[#3`#5`{#8}]
\begin{picture}(\xext,\yext)(\xoff,\yoff)
\putsquarep<\arrowtypea`\arrowtypeb`\arrowtypec`\arrowtyped;\width`\height>%
(0,0)[#2`#3`#4`#5;#6`#7`#8`{#9}]%
\end{picture}%
}}
 
\def\putptrianglep<#1>(#2,#3)[#4`#5`#6;#7`#8`#9]{{%
\settriparms[#1]%
\xpos=#2 \ypos=#3
\advance\ypos by \height
\puthmorphism(\xpos,\ypos)[#4`#5`{#7}]{\height}{\arrowtypea}a%
\putvmorphism(\xpos,\ypos)[`#6`{#8}]{\height}{\arrowtypeb}l%
\advance\xpos by\height
\putmorphism(\xpos,\ypos)(-1,-1)[``{#9}]{\height}{\arrowtypec}r%
}}
 
\def\putptriangle{\@ifnextchar <{\putptrianglep}{\putptrianglep
   <\arrowtypea`\arrowtypeb`\arrowtypec;\height>}}
\def\ptriangle{\@ifnextchar <{\ptrianglep}{\ptrianglep
   <\arrowtypea`\arrowtypeb`\arrowtypec;\height>}}
 
\def\ptrianglep<#1>[#2`#3`#4;#5`#6`#7]{{
\settriparms[#1]%
\width=\height                         
\xext=\width                           
\yext=\width                           
\topadjust[#2`#3`{#5}]
\botadjust[#3``]
\leftadjust[#2`#4`{#6}]
\rightsladjust[#3`#4`{#7}]
\begin{picture}(\xext,\yext)(\xoff,\yoff)
\putptrianglep<\arrowtypea`\arrowtypeb`\arrowtypec;\height>%
(0,0)[#2`#3`#4;#5`#6`{#7}]%
\end{picture}%
}}
 
\def\putqtrianglep<#1>(#2,#3)[#4`#5`#6;#7`#8`#9]{{%
\settriparms[#1]%
\xpos=#2 \ypos=#3
\advance\ypos by\height
\puthmorphism(\xpos,\ypos)[#4`#5`{#7}]{\height}{\arrowtypea}a%
\putmorphism(\xpos,\ypos)(1,-1)[``{#8}]{\height}{\arrowtypeb}l%
\advance\xpos by\height
\putvmorphism(\xpos,\ypos)[`#6`{#9}]{\height}{\arrowtypec}r%
}}
 
\def\putqtriangle{\@ifnextchar <{\putqtrianglep}{\putqtrianglep
   <\arrowtypea`\arrowtypeb`\arrowtypec;\height>}}
\def\qtriangle{\@ifnextchar <{\qtrianglep}{\qtrianglep
   <\arrowtypea`\arrowtypeb`\arrowtypec;\height>}}
 
\def\qtrianglep<#1>[#2`#3`#4;#5`#6`#7]{{
\settriparms[#1]
\width=\height                         
\xext=\width                           
\yext=\height                          
\topadjust[#2`#3`{#5}]
\botadjust[#4``]
\leftsladjust[#2`#4`{#6}]
\rightadjust[#3`#4`{#7}]
\begin{picture}(\xext,\yext)(\xoff,\yoff)
\putqtrianglep<\arrowtypea`\arrowtypeb`\arrowtypec;\height>%
(0,0)[#2`#3`#4;#5`#6`{#7}]%
\end{picture}%
}}
 
\def\putdtrianglep<#1>(#2,#3)[#4`#5`#6;#7`#8`#9]{{%
\settriparms[#1]%
\xpos=#2 \ypos=#3
\puthmorphism(\xpos,\ypos)[#5`#6`{#9}]{\height}{\arrowtypec}b%
\advance\xpos by \height \advance\ypos by\height
\putmorphism(\xpos,\ypos)(-1,-1)[``{#7}]{\height}{\arrowtypea}l%
\putvmorphism(\xpos,\ypos)[#4``{#8}]{\height}{\arrowtypeb}r%
}}
 
\def\putdtriangle{\@ifnextchar <{\putdtrianglep}{\putdtrianglep
   <\arrowtypea`\arrowtypeb`\arrowtypec;\height>}}
\def\dtriangle{\@ifnextchar <{\dtrianglep}{\dtrianglep
   <\arrowtypea`\arrowtypeb`\arrowtypec;\height>}}
 
\def\dtrianglep<#1>[#2`#3`#4;#5`#6`#7]{{
\settriparms[#1]
\width=\height                         
\xext=\width                           
\yext=\height                          
\topadjust[#2``]
\botadjust[#3`#4`{#7}]
\leftsladjust[#3`#2`{#5}]
\rightadjust[#2`#4`{#6}]
\begin{picture}(\xext,\yext)(\xoff,\yoff)
\putdtrianglep<\arrowtypea`\arrowtypeb`\arrowtypec;\height>%
(0,0)[#2`#3`#4;#5`#6`{#7}]%
\end{picture}%
}}
 
\def\putbtrianglep<#1>(#2,#3)[#4`#5`#6;#7`#8`#9]{{%
\settriparms[#1]%
\xpos=#2 \ypos=#3
\puthmorphism(\xpos,\ypos)[#5`#6`{#9}]{\height}{\arrowtypec}b%
\advance\ypos by\height
\putmorphism(\xpos,\ypos)(1,-1)[``{#8}]{\height}{\arrowtypeb}r%
\putvmorphism(\xpos,\ypos)[#4``{#7}]{\height}{\arrowtypea}l%
}}
 
\def\putbtriangle{\@ifnextchar <{\putbtrianglep}{\putbtrianglep
   <\arrowtypea`\arrowtypeb`\arrowtypec;\height>}}
\def\btriangle{\@ifnextchar <{\btrianglep}{\btrianglep
   <\arrowtypea`\arrowtypeb`\arrowtypec;\height>}}
 
\def\btrianglep<#1>[#2`#3`#4;#5`#6`#7]{{
\settriparms[#1]
\width=\height                         
\xext=\width                           
\yext=\height                          
\topadjust[#2``]
\botadjust[#3`#4`{#7}]
\leftadjust[#2`#3`{#5}]
\rightsladjust[#4`#2`{#6}]
\begin{picture}(\xext,\yext)(\xoff,\yoff)
\putbtrianglep<\arrowtypea`\arrowtypeb`\arrowtypec;\height>%
(0,0)[#2`#3`#4;#5`#6`{#7}]%
\end{picture}%
}}
 
\def\putAtrianglep<#1>(#2,#3)[#4`#5`#6;#7`#8`#9]{{%
\settriparms[#1]%
\xpos=#2 \ypos=#3
{\multiply \height by2
\puthmorphism(\xpos,\ypos)[#5`#6`{#9}]{\height}{\arrowtypec}b}%
\advance\xpos by\height \advance\ypos by\height
\putmorphism(\xpos,\ypos)(-1,-1)[#4``{#7}]{\height}{\arrowtypea}l%
\putmorphism(\xpos,\ypos)(1,-1)[``{#8}]{\height}{\arrowtypeb}r%
}}
 
\def\putAtriangle{\@ifnextchar <{\putAtrianglep}{\putAtrianglep
   <\arrowtypea`\arrowtypeb`\arrowtypec;\height>}}
\def\Atriangle{\@ifnextchar <{\Atrianglep}{\Atrianglep
   <\arrowtypea`\arrowtypeb`\arrowtypec;\height>}}
 
\def\Atrianglep<#1>[#2`#3`#4;#5`#6`#7]{{
\settriparms[#1]
\width=\height                         
\xext=\width                           
\yext=\height                          
\topadjust[#2``]
\botadjust[#3`#4`{#7}]
\multiply \xext by2 
\leftsladjust[#3`#2`{#5}]
\rightsladjust[#4`#2`{#6}]
\begin{picture}(\xext,\yext)(\xoff,\yoff)%
\putAtrianglep<\arrowtypea`\arrowtypeb`\arrowtypec;\height>%
(0,0)[#2`#3`#4;#5`#6`{#7}]%
\end{picture}%
}}
 
\def\putAtrianglepairp<#1>(#2)[#3;#4`#5`#6`#7`#8]{{
\settripairparms[#1]%
\setpos(#2)%
\settokens[#3]%
\puthmorphism(\xpos,\ypos)[\tokenb`\tokenc`{#7}]{\height}{\arrowtyped}b%
\advance\xpos by\height
\advance\ypos by\height
\putmorphism(\xpos,\ypos)(-1,-1)[\tokena``{#4}]{\height}{\arrowtypea}l%
\putvmorphism(\xpos,\ypos)[``{#5}]{\height}{\arrowtypeb}m%
\putmorphism(\xpos,\ypos)(1,-1)[``{#6}]{\height}{\arrowtypec}r%
}}
 
\def\putAtrianglepair{\@ifnextchar <{\putAtrianglepairp}{\putAtrianglepairp%
   <\arrowtypea`\arrowtypeb`\arrowtypec`\arrowtyped`\arrowtypee;\height>}}
\def\Atrianglepair{\@ifnextchar <{\Atrianglepairp}{\Atrianglepairp%
   <\arrowtypea`\arrowtypeb`\arrowtypec`\arrowtyped`\arrowtypee;\height>}}
 
\def\Atrianglepairp<#1>[#2;#3`#4`#5`#6`#7]{{%
\settripairparms[#1]%
\settokens[#2]%
\width=\height
\xext=\width
\yext=\height
\topadjust[\tokena``]%
\vertadjust[\tokenb`\tokenc`{#6}]
\tempcountd=\tempcounta                       
\vertadjust[\tokenc`\tokend`{#7}]
\ifnum\tempcounta<\tempcountd                 
\tempcounta=\tempcountd\fi                    
\advance \yext by\tempcounta                  
\advance \yoff by-\tempcounta                 %
\multiply \xext by2 
\leftsladjust[\tokenb`\tokena`{#3}]
\rightsladjust[\tokend`\tokena`{#5}]%
\begin{picture}(\xext,\yext)(\xoff,\yoff)%
\putAtrianglepairp
<\arrowtypea`\arrowtypeb`\arrowtypec`\arrowtyped`\arrowtypee;\height>%
(0,0)[#2;#3`#4`#5`#6`{#7}]%
\end{picture}%
}}
 
\def\putVtrianglep<#1>(#2,#3)[#4`#5`#6;#7`#8`#9]{{%
\settriparms[#1]%
\xpos=#2 \ypos=#3
\advance\ypos by\height
{\multiply\height by2
\puthmorphism(\xpos,\ypos)[#4`#5`{#7}]{\height}{\arrowtypea}a}%
\putmorphism(\xpos,\ypos)(1,-1)[`#6`{#8}]{\height}{\arrowtypeb}l%
\advance\xpos by\height
\advance\xpos by\height
\putmorphism(\xpos,\ypos)(-1,-1)[``{#9}]{\height}{\arrowtypec}r%
}}

\def\putVtriangle{\@ifnextchar <{\putVtrianglep}{\putVtrianglep
   <\arrowtypea`\arrowtypeb`\arrowtypec;\height>}}
\def\Vtriangle{\@ifnextchar <{\Vtrianglep}{\Vtrianglep
   <\arrowtypea`\arrowtypeb`\arrowtypec;\height>}}
 
\def\Vtrianglep<#1>[#2`#3`#4;#5`#6`#7]{{
\settriparms[#1]
\width=\height                         
\xext=\width                           
\yext=\height                          
\topadjust[#2`#3`{#5}]
\botadjust[#4``]
\multiply \xext by2 
\leftsladjust[#2`#3`{#6}]
\rightsladjust[#3`#4`{#7}]
\begin{picture}(\xext,\yext)(\xoff,\yoff)%
\putVtrianglep<\arrowtypea`\arrowtypeb`\arrowtypec;\height>%
(0,0)[#2`#3`#4;#5`#6`{#7}]%
\end{picture}%
}}
 
\def\putVtrianglepairp<#1>(#2)[#3;#4`#5`#6`#7`#8]{{
\settripairparms[#1]%
\setpos(#2)%
\settokens[#3]%
\advance\ypos by\height
\putmorphism(\xpos,\ypos)(1,-1)[`\tokend`{#6}]{\height}{\arrowtypec}l%
\puthmorphism(\xpos,\ypos)[\tokena`\tokenb`{#4}]{\height}{\arrowtypea}a%
\advance\xpos by\height
\putvmorphism(\xpos,\ypos)[``{#7}]{\height}{\arrowtyped}m%
\advance\xpos by\height
\putmorphism(\xpos,\ypos)(-1,-1)[``{#8}]{\height}{\arrowtypee}r%
}}
 
\def\putVtrianglepair{\@ifnextchar <{\putVtrianglepairp}{\putVtrianglepairp%
    <\arrowtypea`\arrowtypeb`\arrowtypec`\arrowtyped`\arrowtypee;\height>}}
\def\Vtrianglepair{\@ifnextchar <{\Vtrianglepairp}{\Vtrianglepairp%
    <\arrowtypea`\arrowtypeb`\arrowtypec`\arrowtyped`\arrowtypee;\height>}}
 
\def\Vtrianglepairp<#1>[#2;#3`#4`#5`#6`#7]{{%
\settripairparms[#1]%
\settokens[#2]
\xext=\height                  
\width=\height                 
\yext=\height                  
\vertadjust[\tokena`\tokenb`{#4}]
\tempcountd=\tempcounta        
\vertadjust[\tokenb`\tokenc`{#5}]
\ifnum\tempcounta<\tempcountd%
\tempcounta=\tempcountd\fi
\advance \yext by\tempcounta
\botadjust[\tokend``]%
\multiply \xext by2
\leftsladjust[\tokena`\tokend`{#6}]%
\rightsladjust[\tokenc`\tokend`{#7}]%
\begin{picture}(\xext,\yext)(\xoff,\yoff)%
\putVtrianglepairp
<\arrowtypea`\arrowtypeb`\arrowtypec`\arrowtyped`\arrowtypee;\height>%
(0,0)[#2;#3`#4`#5`#6`{#7}]%
\end{picture}%
}}

\def\putCtrianglep<#1>(#2,#3)[#4`#5`#6;#7`#8`#9]{{%
\settriparms[#1]%
\xpos=#2 \ypos=#3
\advance\ypos by\height
\putmorphism(\xpos,\ypos)(1,-1)[``{#9}]{\height}{\arrowtypec}l%
\advance\xpos by\height
\advance\ypos by\height
\putmorphism(\xpos,\ypos)(-1,-1)[#4`#5`{#7}]{\height}{\arrowtypea}l%
{\multiply\height by 2
\putvmorphism(\xpos,\ypos)[`#6`{#8}]{\height}{\arrowtypeb}r}%
}}
 
\def\putCtriangle{\@ifnextchar <{\putCtrianglep}{\putCtrianglep
    <\arrowtypea`\arrowtypeb`\arrowtypec;\height>}}
\def\Ctriangle{\@ifnextchar <{\Ctrianglep}{\Ctrianglep
    <\arrowtypea`\arrowtypeb`\arrowtypec;\height>}}
 
\def\Ctrianglep<#1>[#2`#3`#4;#5`#6`#7]{{
\settriparms[#1]
\width=\height                          
\xext=\width                            
\yext=\height                           
\multiply \yext by2 
\topadjust[#2``]
\botadjust[#4``]
\sladjust[#3`#2`{#5}]{\width}
\tempcountd=\tempcounta                 
\sladjust[#3`#4`{#7}]{\width}
\ifnum \tempcounta<\tempcountd          
\tempcounta=\tempcountd\fi              
\advance \xext by\tempcounta            
\advance \xoff by-\tempcounta           %
\rightadjust[#2`#4`{#6}]
\begin{picture}(\xext,\yext)(\xoff,\yoff)%
\putCtrianglep<\arrowtypea`\arrowtypeb`\arrowtypec;\height>%
(0,0)[#2`#3`#4;#5`#6`{#7}]%
\end{picture}%
}}
 
\def\putDtrianglep<#1>(#2,#3)[#4`#5`#6;#7`#8`#9]{{%
\settriparms[#1]%
\xpos=#2 \ypos=#3
\advance\xpos by\height \advance\ypos by\height
\putmorphism(\xpos,\ypos)(-1,-1)[``{#9}]{\height}{\arrowtypec}r%
\advance\xpos by-\height \advance\ypos by\height
\putmorphism(\xpos,\ypos)(1,-1)[`#5`{#8}]{\height}{\arrowtypeb}r%
{\multiply\height by 2
\putvmorphism(\xpos,\ypos)[#4`#6`{#7}]{\height}{\arrowtypea}l}%
}}
 
\def\putDtriangle{\@ifnextchar <{\putDtrianglep}{\putDtrianglep
    <\arrowtypea`\arrowtypeb`\arrowtypec;\height>}}
\def\Dtriangle{\@ifnextchar <{\Dtrianglep}{\Dtrianglep
   <\arrowtypea`\arrowtypeb`\arrowtypec;\height>}}
 
\def\Dtrianglep<#1>[#2`#3`#4;#5`#6`#7]{{
\settriparms[#1]
\width=\height                         
\xext=\height                          
\yext=\height                          
\multiply \yext by2 
\topadjust[#2``]
\botadjust[#4``]
\leftadjust[#2`#4`{#5}]
\sladjust[#3`#2`{#5}]{\height}
\tempcountd=\tempcountd                
\sladjust[#3`#4`{#7}]{\height}
\ifnum \tempcounta<\tempcountd         
\tempcounta=\tempcountd\fi             
\advance \xext by\tempcounta           %
\begin{picture}(\xext,\yext)(\xoff,\yoff)
\putDtrianglep<\arrowtypea`\arrowtypeb`\arrowtypec;\height>%
(0,0)[#2`#3`#4;#5`#6`{#7}]%
\end{picture}%
}}
 
\def\setrecparms[#1`#2]{\width=#1 \height=#2}%
\def\recursep<#1`#2>[#3;#4`#5`#6`#7`#8]{{%
\width=#1 \height=#2
\settokens[#3]
\settowidth{\tempdimen}{$\tokena$}
\ifdim\tempdimen=0pt
  \savebox{\tempboxa}{\hbox{$\tokenb$}}%
  \savebox{\tempboxb}{\hbox{$\tokend$}}%
  \savebox{\tempboxc}{\hbox{$#6$}}%
\else
  \savebox{\tempboxa}{\hbox{$\hbox{$\tokena$}\times\hbox{$\tokenb$}$}}%
  \savebox{\tempboxb}{\hbox{$\hbox{$\tokena$}\times\hbox{$\tokend$}$}}%
  \savebox{\tempboxc}{\hbox{$\hbox{$\tokena$}\times\hbox{$#6$}$}}%
\fi
\ypos=\height
\divide\ypos by 2
\xpos=\ypos
\advance\xpos by \width
\xext=\xpos \yext=\height
\topadjust[#3`\usebox{\tempboxa}`{#4}]%
\botadjust[#5`\usebox{\tempboxb}`{#8}]%
\sladjust[\tokenc`\tokenb`{#5}]{\ypos}%
\tempcountd=\tempcounta
\sladjust[\tokenc`\tokend`{#5}]{\ypos}%
\ifnum \tempcounta<\tempcountd
\tempcounta=\tempcountd\fi
\advance \xext by\tempcounta
\advance \xoff by-\tempcounta
\rightadjust[\usebox{\tempboxa}`\usebox{\tempboxb}`\usebox{\tempboxc}]%
\bfig
\putCtrianglep<-1`1`1;\ypos>(0,0)[`\tokenc`;#5`#6`{#7}]%
\puthmorphism(\ypos,0)[\tokend`\usebox{\tempboxb}`{#8}]{\width}{-1}b%
\puthmorphism(\ypos,\height)[\tokenb`\usebox{\tempboxa}`{#4}]{\width}{-1}a%
\advance\ypos by \width
\putvmorphism(\ypos,\height)[``\usebox{\tempboxc}]{\height}1r%
\efig
}}
 
\def\recurse{\@ifnextchar <{\recursep}{\recursep<\width`\height>}}
 
\def\puttwohmorphisms(#1,#2)[#3`#4;#5`#6]#7#8#9{{%
\puthmorphism(#1,#2)[#3`#4`]{#7}0a
\ypos=#2
\advance\ypos by 20
\puthmorphism(#1,\ypos)[\phantom{#3}`\phantom{#4}`#5]{#7}{#8}a
\advance\ypos by -40
\puthmorphism(#1,\ypos)[\phantom{#3}`\phantom{#4}`#6]{#7}{#9}b
}}
 
\def\puttwovmorphisms(#1,#2)[#3`#4;#5`#6]#7#8#9{{%
\putvmorphism(#1,#2)[#3`#4`]{#7}0a
\xpos=#1
\advance\xpos by -20
\putvmorphism(\xpos,#2)[\phantom{#3}`\phantom{#4}`#5]{#7}{#8}l
\advance\xpos by 40
\putvmorphism(\xpos,#2)[\phantom{#3}`\phantom{#4}`#6]{#7}{#9}r
}}
 
\def\puthcoequalizer(#1)[#2`#3`#4;#5`#6`#7]#8#9{{%
\setpos(#1)%
\puttwohmorphisms(\xpos,\ypos)[#2`#3;#5`#6]{#8}11%
\advance\xpos by #8
\puthmorphism(\xpos,\ypos)[\phantom{#3}`#4`#7]{#8}1{#9}
}}
 
\def\putvcoequalizer(#1)[#2`#3`#4;#5`#6`#7]#8#9{{%
\setpos(#1)%
\puttwovmorphisms(\xpos,\ypos)[#2`#3;#5`#6]{#8}11%
\advance\ypos by -#8
\putvmorphism(\xpos,\ypos)[\phantom{#3}`#4`#7]{#8}1{#9}
}}
 
\def\putthreehmorphisms(#1)[#2`#3;#4`#5`#6]#7(#8)#9{{%
\setpos(#1) \settypes(#8)
\if a#9 %
     \vertsize{\tempcounta}{#5}%
     \vertsize{\tempcountb}{#6}%
     \ifnum \tempcounta<\tempcountb \tempcounta=\tempcountb \fi
\else
     \vertsize{\tempcounta}{#4}%
     \vertsize{\tempcountb}{#5}%
     \ifnum \tempcounta<\tempcountb \tempcounta=\tempcountb \fi
\fi
\advance \tempcounta by 60
\puthmorphism(\xpos,\ypos)[#2`#3`#5]{#7}{\arrowtypeb}{#9}
\advance\ypos by \tempcounta
\puthmorphism(\xpos,\ypos)[\phantom{#2}`\phantom{#3}`#4]{#7}{\arrowtypea}{#9}
\advance\ypos by -\tempcounta \advance\ypos by -\tempcounta
\puthmorphism(\xpos,\ypos)[\phantom{#2}`\phantom{#3}`#6]{#7}{\arrowtypec}{#9}
}}
 
\def\putarc(#1,#2)[#3`#4`#5]#6#7#8{{%
\xpos #1
\ypos #2
\width #6
\arrowlength #6
\putbox(\xpos,\ypos){#3\vphantom{#4}}%
{\advance \xpos by\arrowlength
\putbox(\xpos,\ypos){\vphantom{#3}#4}}%
\horsize{\tempcounta}{#3}%
\horsize{\tempcountb}{#4}%
\divide \tempcounta by2
\divide \tempcountb by2
\advance \tempcounta by30
\advance \tempcountb by30
\advance \xpos by\tempcounta
\advance \arrowlength by-\tempcounta
\advance \arrowlength by-\tempcountb
\halflength=\arrowlength \divide\halflength by 2
\divide\arrowlength by 5
\put(\xpos,\ypos){\bezier{\arrowlength}(0,0)(50,50)(\halflength,50)}
\ifnum #7=-1 \put(\xpos,\ypos){\vector(-3,-2)0} \fi
\advance\xpos by \halflength
\put(\xpos,\ypos){\xpos=\halflength \advance\xpos by -50
   \bezier{\arrowlength}(0,50)(\xpos,50)(\halflength,0)}
\ifnum #7=1 {\advance \xpos by
   \halflength \put(\xpos,\ypos){\vector(3,-2)0}} \fi
\advance\ypos by 50
\vertsize{\tempcounta}{#5}%
\divide\tempcounta by2
\advance \tempcounta by20
\if a#8 %
   \advance \ypos by\tempcounta
   \putbox(\xpos,\ypos){#5}%
\else
   \advance \ypos by-\tempcounta
   \putbox(\xpos,\ypos){#5}%
\fi
}}
 
\makeatother

\def\doi{6 (4:10) 2010}
\lmcsheading%
{\doi}
{1--48}
{}
{}
{May\phantom{.}~17, 2010}
{Dec.~21, 2010}
{}

\begin{document}


\title[Backward Reachability of Array-based Systems by SMT solving]{%
  Backward Reachability of Array-based Systems by SMT solving:
  Termination and Invariant Synthesis}

\author[S.\ Ghilardi]{Silvio Ghilardi$^1$}	
\address{$^1$Dipartimento di Scienze dell'Informazione, Universit\`a degli Studi di Milano (Italy)}	
\email{ghilardi@dsi.unimi.it}  

\author[S.\ Ranise]{Silvio Ranise$^2$}	
\address{$^2$FBK-Irst, Trento (Italy)}	
\email{ranise@fbk.eu}  


\titlecomment{This paper extends~\cite{tableaux09} with all the proofs
  and adapts materials in~\cite{ijcar08,ijcar10} to make it
  self-contained.}


\begin{abstract}
  The safety of infinite state systems can be checked by a backward
  reachability procedure.  For certain classes of systems, it is
  possible to prove the termination of the procedure and hence
  conclude the decidability of the safety problem.  Although backward
  reachability is property-directed, it can unnecessarily explore
  (large) portions of the state space of a system which are not
  required to verify the safety property under consideration.  To
  avoid this, invariants can be used to dramatically prune the search
  space.  Indeed, the problem is to guess such appropriate invariants.

  In this paper, we present a fully declarative and symbolic approach
  to the mechanization of backward reachability of infinite state
  systems manipulating arrays by Satisfiability Modulo Theories
  solving.  Theories are used to specify the topology and the data
  manipulated by the system.  We identify sufficient conditions on the
  theories to ensure the termination of backward reachability and we
  show the completeness of a method for invariant synthesis (obtained
  as the dual of backward reachability), again, under suitable
  hypotheses on the theories.  We also present a pragmatic approach to
  interleave invariant synthesis and backward reachability so that a
  fix-point for the set of backward reachable states is more easily
  obtained.  Finally, we discuss heuristics that allow us to derive an
  implementation of the techniques in the model checker \textsc{mcmt},
  showing remarkable speed-ups on a significant set of safety problems
  extracted from a variety of sources.
\end{abstract}

\maketitle

\newpage

\tableofcontents
\begin{figure}[b]
  \centering
  \textsc{How to read the paper: main tracks}

  \ \\

  \includegraphics[scale=.9]{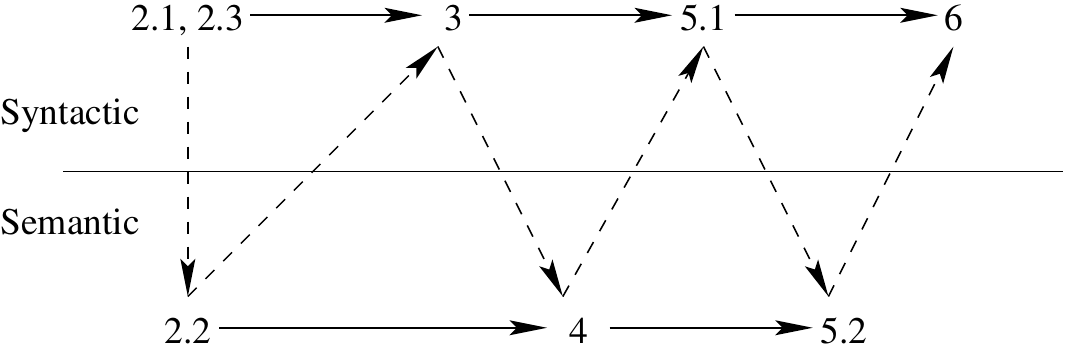}
\end{figure}


\section{Introduction}

Backward reachability analysis has been widely adopted in model
checking of safety properties for infinite state systems (see,
e.g.,~\cite{lics}).  This verification procedure repeatedly computes
pre-images of
a set of unsafe states, usually obtained by complementing a safety
property that a system should satisfy.  Potentially infinite sets of
states are represented by constraints so that pre-image computation
can be done symbolically.  The procedure halts in two cases, either
when the current set of (backward) reachable states has a non-empty
intersection with the set of initial states---called the \emph{safety
  check}---and the system is unsafe, or when such a set has reached a
fix-point (i.e.\ further application of the transition does not
enlarge the set of reachable states)---called the \emph{fix-point
  check}---and the system is safe.  One of the most important key
insights of backward reachability is the possibility to show the
decidability of checking safety properties for some classes of
infinite state systems, such as broadcast protocols~\cite{bro1,bro2},
lossy channel systems~\cite{lossy-channels}, timed
networks~\cite{AbdullaTCS}, and parametric and distributed systems
with global conditions~\cite{cav06,vmcai08}.  The main ingredient of
the technique for proving decidability of safety is the existence of a
well-quasi-ordering
over the infinite set of states entailing the termination of backward
reachability~\cite{lics}.

\subsection{Array-based systems and symbolic backward reachability}

An \emph{array-based system} (first introduced in~\cite{ijcar08}) is a
generalization of all the classes of infinite state systems mentioned
above.  Even more, it supports also the specification and verification
of algorithms manipulating arrays and fault tolerant systems that are
well beyond the paradigms underlying the verification method mentioned
above.  Roughly, an array-based system is a transition system which
updates one (or more) array variable $\mathtt{a}$.  Being parametric
in the structures associated to the indexes and the elements in
$\mathtt{a}$, the notion of array-based system is quite flexible and
allows one the declarative specification of several classes of
infinite state systems.  For example, consider parametrised systems
and the task of specifying their topology: by using no structure at
all, indexes are simply identifiers of processes that can only be
compared for equality; by using a linear order, indexes are
identifiers of processes so that it is possible to distinguish between
those on the left or on the right of a process with a particular
identifier; by using richer and richer structures (such as trees and
graphs), it is possible to specify more and more complex
topologies. Similar observations hold also for elements, where it is
well-known how to use algebraic structures to specify data structures.
Formally, the structure on both indexes and elements is declaratively
and uniformly specified by \emph{theories}, i.e.\ pairs formed by a
(first-order) language and a class of (first-order) structures.

On top of the notion of array-based system, it is possible to design a
fully symbolic and declarative version of backward reachability for
the verification of safety properties where sets of backward reachable
states are represented by certain classes of first-order formulae over
the signature induced by the theories over the indexes and the
elements of the array-based system under consideration.  To mechanize
this approach, the following three requirements are mandatory:
\begin{enumerate}[{\rm (i)}]
\item the class $\mathcal{F}$ of (possibly quantified) first-order
  formulae used to represent sets of states is expressive enough to
  represent interesting classes of systems and safety properties,
\item $\mathcal{F}$ is closed under pre-image computation, and
\item the checks for safety and fix-point can be reduced to decidable
  logical problems (e.g., satisfiability) of formulae in
  $\mathcal{F}$.
\end{enumerate}
Once requirements (i)---(iii) are satisfied, this technique can be
seen as a symbolic version of the model checking techniques
of~\cite{lics} revisited in the declarative framework of first-order
logic augmented with theories (as first discussed in~\cite{ijcar08}).
Using a declarative framework has several \emph{potential} advantages;
two of the most important ones are the following.  First, the
computation of the pre-image (requirement (ii) above) becomes
computationally cheap: we only need to build the formula $\phi$
representing the (iterated) pre-images of the set of unsafe states and
then put the burden of using suitable data structures to represent
$\phi$ on the available (efficient) solver for logical problems
encoding safety and fix-point checks.  This is in sharp contrast to
what is usually done in almost all other approaches to symbolic model
checking of infinite state systems, where the computation of the
pre-image is computationally very expensive as it requires a
substantial process of normalization on the data structure
representing the (infinite) sets of states so as to simplify safety
and fix-point checks.  The second advantage is the possibility to use
state-of-the-art Satisfiability Modulo Theories (SMT) solvers, a
technology that is showing very good success in scaling up various
verification techniques, to support both safety and fix-point checks
(requirement (iii) above).  Unfortunately, the kind of satisfiability
problems obtained in the context of the backward search algorithm
requires to cope with (universal) quantifiers and this makes the
off-the-shelf use of SMT solvers problematic.  In fact, even when
using classes of formulae with decidable satisfiability problem,
currently available SMT solvers are not yet mature enough to
efficiently discharge formulae containing (universal) quantifiers,
despite the fact that this problem has recently attracted a lot of
efforts (see,
e.g.,~\cite{miopapersttt,barret-tinelli,demoura-bijoerner}).  To
alleviate this problem, we have designed a general decision procedure
for a class of formulae satisfying requirement (i) above, based on
quantifier instantiation (see~\cite{ijcar08} and
Theorem~\ref{th:decidability} below); this allows for an easier way to
integrate currently available SMT-solvers in the backward reachability
procedure.  
Interestingly, it is possible to describe the symbolic backward
reachability procedure by means of a Tableaux-like calculus which
offers a good starting point for implementation.  In fact, the main
loop of \textsc{mcmt}~\cite{ijcar10},\footnote{The latest available
  release of the tool with all the benchmarks discussed in this paper
  (and more) can be downloaded at
  \texttt{\url{http://homes.dsi.unimi.it/~ghilardi/mcmt}}. } the model
checker for array-based systems that we are currently developing, can
be easily understood in terms of the rules of the calculus.
The current version of the tool uses Yices~\cite{Dutertre06theyices}
as the back-end SMT solver.  We have chosen Yices among the many
available state-of-the-art solvers because it has scored well in many
editions of the SMT-COMP competition and because its lightweight API
allowed us to easily embed it in \textsc{mcmt}.  An interesting line of
future work would be to make the tool parametric with respect to the
back-end SMT solver so as to permit the user to select the most
appropriate for the problem under consideration.

In our declarative framework, it is also possible to identify
sufficient conditions on the theories about indexes and elements of
the array-based systems so as to ensure the termination of the
symbolic backward reachability procedure.  This allows us to derive
all the decidability results for the safety problems of the classes of
systems mentioned above.  Interestingly, the well-quasi-ordering used
for the proof of termination can be obtained by using standard model
theoretic notions (namely, sub-structures and embeddings) and in
conjunction with well-known mathematical results for showing that a
binary relation is a well-quasi-order (e.g., Dickson's Lemma or
Kruskal's Theorem).
Contrary to the approach proposed in~\cite{lics}---where some ingenuity
is required, in our framework the definition of well-quasi-order is
derived from the class of structures formalizing indexes and elements
in a uniform way by using the model-theoretic notions of sub-structure
and embedding.

\subsection{Symbolic backward reachability and Invariant synthesis}

One of the key advantages of backward reachability over other
verification methods is to be \emph{goal-directed}; the goal being the
set of unsafe states from which pre-images are computed.
Despite this, it can unnecessarily explore (large) portions of the
symbolic state space of a system which are not required to verify the
safety property under consideration.  Even worse, in some cases the
analysis may not detect a fix-point, thereby causing non-termination.
In order to avoid visiting irrelevant parts of the symbolic state
space during backward reachability, techniques for analyzing
pre-images, over-approximating the set of backward reachable states,
and guessing invariants have been devised (see,
e.g.,~\cite{dill-cav95,henzinger-hytech,park-dill-pvs,bjorner-browne-manna-cp95,bradleymanna,pnueli,qaaderflanagan,rybalchenko,indexedabs}
to name a few).  The success of these techniques depend crucially on
the heuristics used to guess the invariants or compute
over-approximations.  Our approach is similar in spirit
to~\cite{bradleymanna}, but employs techniques which are specific for
our different intended application domains.

Along this line of research, we discuss a technique for interleaving
pre-image computation and invariant synthesis which tries to eagerly
prune irrelevant parts of the search space.  Formally, in the context
of the declarative framework described above, our main result about
invariant synthesis ensures that the technique \emph{will find an
  invariant---provided one exists---under suitable hypotheses}, which
are satisfied for important classes of array-based systems (e.g.,
mutual exclusion algorithms or cache coherence protocols).  The key
ingredient in the proof of the result is again the model-theoretic
notion of well-quasi-ordering obtained by applying standard model
theoretic notions that already played a key role in showing the
termination of the backward reachability procedure.  In this case, it
allows us to finitely characterize the search space of candidate
invariants.
Although the technique is developed for array-based systems, we
believe that the underlying idea can be adapted to other symbolic
approaches to model checking (e.g.,~\cite{tacas06,cav06}).

Although the correctness of our invariant synthesis method is
theoretically interesting, its implementation seems to be impractical
because of the huge (finite) search space that must be traversed in
order to find the desired invariant.  In order to make our findings
more practically relevant, we study how to integrate invariant
synthesis with backward reachability so as to prune the search space
of the latter efficiently.  To this end, we develop techniques that
allow us to analyze a set of backward reachable states and then guess
candidate invariants.  Such candidate invariants are then proved to be
``real'' invariants by using a resource bounded variant of the
backward reachability procedure and afterwards are used during
fix-point checking with the hope that they help pruning the search
space by augmenting the chances to detect a fix-point.  Two
observations are important.  First, the bound on the resources of the
backward reachability procedure is because we want to obtain
invariants in a computationally cheap way.  Second, we have complete
freedom in the design of the invariant generation techniques as all
the candidate invariants are checked to be real invariants before
being used by the main backward reachability procedure.
As a consequence, (even coarse) abstraction techniques can be used to
compute candidate invariants without putting at risk the accuracy of
the (un-)safety result returned by the main verification procedure.
For concreteness, we discuss two techniques for invariant guessing:
both compute over-approximation of the set of backward reachable
states.  The former, called \emph{index abstraction} (which resembles
the technique of~\cite{indexedabs}), projects away the indexes in the
formula used to describe a set of backward reachable states while the
latter, called \emph{signature abstraction} (which can be seen as a
form of predicate abstraction~\cite{seminal}), projects away the
elements of a sub-set of the array variables by quantifier elimination
(if possible).  The effectiveness of the proposed invariant synthesis
techniques and their integration in the backward reachability
procedure must be judged experimentally.  Hence, we have implemented
them in \textsc{mcmt} and we have performed an experimental analysis
on several safety problems translated from available model checkers
for parametrised systems (e.g., \textsc{pfs}, Undip, the version of
UCLID extended with predicate abstraction) or obtained by the
formalizing programs manipulating arrays (e.g., sorting algorithms).
The results confirm the viability and the effectiveness of the
proposed invariant synthesis techniques either by more quickly finding
a fix-point (when the backward reachability procedure alone was
already able to find it) or by allowing to find a fix-point (when the
backward reachability procedure alone was not terminating).

\subsection*{How to read the paper}

Given the size of the paper, we identify two tracks for the reader.
The former is the `symbolic' track and allows one to focus on the
declarative framework, the mechanization of the backward reachability
procedure, its combination with invariant synthesis techniques, and
its experimental evaluation.  The latter is the `semantic' track which
goes into the details of the connection between the syntactic
characterization of sets of states and the well-quasi-ordering
permitting one to prove the termination of the backward reachability
procedure and the completeness of invariant synthesis.  To some
extent, the two tracks can be read independently.
\begin{enumerate}[$\bullet$]
\item \emph{Symbolic track}.  In Sections~\ref{subsec:case-def}
  and~\ref{subsec:many-sort}, some preliminary notions underlying the
  concept of array-based system (Section~\ref{subsec:array-based}) are
  given.  In Section~\ref{sec:array}, the symbolic version of backward
  reachability is described, requirements for its mechanization are
  considered, namely closure under pre-image computation and
  decidability of safety and fix-point checks
  (Section~\ref{subsec:brs}), and its formalization using a
  Tableaux-like calculus is presented (Section~\ref{subsec:tab}).  In
  Section~\ref{subsec:dec+inv}, the notion of safety invariants is
  introduced, their synthesis and use to prune the search space of the
  backward reachability procedure is described,
  and their implementation is
  considered in Section~\ref{sec:implementation}.  Particular care has
  been put in the experimental evaluation of the proposed techniques
  for invariant synthesis as illustrated in Section~\ref{sec:exp}.
\item \emph{Semantic track}.  In Section~\ref{subsec:embed}, some
  notions related to the model theoretic concept of embedding are
  briefly summarized.  In Section~\ref{sec:term}, it is explained how
  a pre-order can be defined on sets of states by using the notion of
  embedding and how this allows us (in case the pre-order is a
  well-quasi-order) to prove the termination of the backward
  reachability procedure designed in Section~\ref{sec:array}.  For the
  sake of completeness, it is also stated that the safety problem for
  array-based system is undecidable (Section~\ref{subsec:undec}) and
  its proof can be found in the Appendix.  In
  Section~\ref{subsec:dual}, the completeness of an algorithm for
  invariant synthesis (obtained as the dual of backward reachability)
  is proved under suitable hypotheses.
\end{enumerate}
In Section~\ref{sec:discussion}, we conclude the paper by positioning
our work with respect to the state-of-the-art in verification of the
safety of infinite state systems and we sketch some lines of future
work.  For ease of reference, at the end of the paper, we include the
table of contents and a figure depicting the two tracks for reading
mentioned above.

\section{Formal Preliminaries}
\label{sec:prelim}

We assume the usual syntactic (e.g., signature, variable, term, atom,
literal, and formula) and semantic (e.g., structure, 
truth, satisfiability, and validity) notions of first-order logic
(see, e.g.,~\cite{enderton}).  The equality symbol $=$ is included in
all signatures considered below.  A signature is \emph{relational} if
it does not contain function symbols and it is \emph{quasi-relational}
if its function symbols are all 
constants.  An \emph{expression} is a term, an atom, a literal, or a
formula. Let $\ux$ be a finite tuple of variables and $\Sigma$ a
signature; a $\Sigma(\ux)$-expression is an expression built out of
the symbols in $\Sigma$ where at most the variables in $\ux$ may occur
free (we will write $E(\ux)$ to emphasize that $E$ is a
$\Sigma(\ux)$-expression).  Let $\underline{e}$ be a finite sequence
of expressions and $\sigma$ a substitution; $\underline{e}\sigma$ is
the result of applying the substitution $\sigma$ to each element of
the sequence $\underline{e}$.

According to the current practice in the SMT
literature~\cite{smt-lib}, a \emph{theory} $T$ is a pair $({\Sigma},
\Cc)$, where $\Sigma$ is a signature and $\Cc$ is a class of
$\Sigma$-structures; the structures in $\Cc$ are the \emph{models} of
$T$.  Below, we let $T=({\Sigma}, \Cc)$.
A $\Sigma$-formula $\phi$ is \emph{$T$-satisfiable} if there exists a
$\Sigma$-structure $\cM$ in $\Cc$ such that $\phi$ is true in $\cM$
under a suitable assignment to the free variables of $\phi$ (in
symbols, $\cM \models \phi$); it is \emph{$T$-valid} (in symbols,
$T\models \varphi$) if its negation is $T$-unsatisfiable.  Two
formulae $\varphi_1$ and $\varphi_2$ are \emph{$T$-equivalent} if
$\varphi_1 \leftrightarrow \varphi_2$ is \emph{$T$-valid}.  The
\emph{quantifier-free satisfiability modulo the theory $T$} ($SMT(T)$)
\emph{problem} amounts to establishing the $T$-satisfiability of
quantifier-free $\Sigma$-formulae.

$T$ admits \emph{quantifier elimination} iff
given an arbitrary formula $\varphi(\ux)$, it is always possible to
compute a quantifier-free formula $\varphi'(\ux)$ such that $T\models
\forall \ux (\varphi(\ux)\leftrightarrow \varphi'(\ux))$.  Linear
Arithmetics, Real Arithmetics, acyclic lists, and enumerated data-type
theories (see below) are examples of theories that admit elimination
of quantifiers.

A theory $T=({\Sigma}, \Cc)$ is said to be \emph{locally finite} iff
$\Sigma$ is finite and, for every finite set of variables \ux, there
are finitely many $\Sigma(\ux)$-terms $t_1, \dots, t_{k_{\ux}}$ such
that for every further $\Sigma({\ux})$-term $u$, we have that $T
\models u = t_i$ (for some \idx{k_{\ux}}). The terms $t_1, \dots,
t_{k_{\ux}}$ are called \emph{$\Sigma(\ux)$-representative terms}; if
they are effectively computable from \ux\ (and $t_i$ is computable
from $u$), then $T$ is said to be \emph{effectively locally finite}
(in the following, when we say `locally finite', we in fact always
mean `effectively locally finite').  If $\Sigma$ is relational or
quasi-relational, then any $\Sigma$-theory $T$ is locally finite.

An important class of theories, ubiquitously used in verification,
formalizes enumerated data-types.  An \emph{enumerated data-type
  theory} $T$ is a theory in a quasi-relational signature whose class
of models contains only a single finite $\Sigma$-structure $\cM = (M,
\mathcal I )$ such that for every $m \in M$ there exists a constant $c
\in \Sigma$ such that $c^{\mathcal I} = m$.  For example, enumerated
data-type theories can be used to model control locations of processes
in parametrised systems (see Example~\ref{ex:one} below).

\subsection{Case Defined Functions}
\label{subsec:case-def}
In the SMT-LIB format~\cite{smtlib}, it is possible to use
if-then-else constructors when building terms. This may seem to be
beyond the realm of first order logic, but in fact these constructors
can be easily eliminated in SMT problems. Since case-defined functions
(introduced via nested if-then-else constructors) are quite useful for
us too, we briefly explain the underlying formal aspects here. Given a
theory $T$,
a \emph{$T$-partition} is a finite set $C_1(\ux), \dots, C_n(\ux)$ of
quantifier-free formulae such that $T\models \forall \ux
\bigvee_{i=1}^n C_i(\ux)$ and $T\models \bigwedge_{i\not=j}\forall \ux
\neg (C_i(\ux)\wedge C_j(\ux))$.  A \emph{case-definable extension}
$T'=(\Sigma', \Cc')$ of a theory $T=(\Sigma, \Cc)$ is obtained from
$T$ by applying (finitely many times) the following procedure: (i)
take a $T$-partition $C_1(\ux), \dots, C_n(\ux)$ together with
$\Sigma$-terms $t_1(\ux), \dots, t_n(\ux)$; (ii) let $\Sigma'$ be
$\Sigma\cup\{F\}$, where $F$ is a ``fresh'' function symbol
(i.e. $F\not\in\Sigma$) whose arity is equal to the length of $\ux$;
(iii) take as $\Cc'$ the class of $\Sigma'$-structures $\cM$ whose
$\Sigma$-reduct is a model of $T$ and such that $\cM \models
\bigwedge_{i=1}^n \forall\ux\; (C_i(\ux) \to F(\ux) = t_i(\ux)).$ Thus
a case-definable extension $T'$ of a theory $T$ contains finitely many
additional function symbols, called case-defined functions.

\begin{lemma}
  \label{lem:extensions}
  Let $T'$ be a case-definable extension of $T$; for every formula
  $\phi'$ in the signature of $T'$ it is possible to compute a formula
  $\phi$ in the signature of $T$ such that $\phi$ and $\phi'$ are
  $T'$-equivalent.
\end{lemma}
\begin{proof}
  It is sufficient to show the claim for an atomic $\phi'$ containing
  a single occurrence of a case defined function: if this holds, one
  can get the general statement by the replacement theorem for
  equivalent formulae (the procedure must be iterated until all case
  defined additional function symbols are eliminated). Let $\phi$ be
  atomic and let it contain a sub-term of the kind $F\sigma$ in
  position $p$.  Then $\phi$ is $T'$-equivalent to $\bigvee_i
  (C_i\sigma\wedge \phi'[t_i\sigma]_p)$. Here the $C_i$'s are the
  partition formulae for the case definition of $F$ and the $t_i$'s
  are the related `value' terms; the notation $\phi'[t_i\sigma]_p$
  means the formula obtained from $\phi'$ by putting $t_i\sigma$ in
  position $p$.
\end{proof}
Notice that a case-definable extension $T'$ of $T$ is a conservative
extension of $T$, i.e.\ formulae in the signature of $T$ are
$T$-satisfiable iff they are $T'$-satisfiable (this is because, as far
as the signature of $T$ is concerned, the two theories have `the same
models').
Thus, by Lemma~\ref{lem:extensions}, $T$ and $T'$ are basically the
same theory and, by abuse of notation, we shall write $T$ instead of
$T'$.

\subsection{Embeddings}
\label{subsec:embed}
We summarize some basic model-theoretic notions that will be used in
Sections~\ref{sec:term} and~\ref{sec:inv} below (for more details, the
interested reader is pointed to standard textbooks in model theory,
such as~\cite{CK}).

A {\it $\Sigma$-embedding} (or, simply, an embedding) between two
$\Sigma$-structu\-res $\cM=(M, \int)$ and $\cN=(N,\mathcal J)$ is any
mapping $\mu: M \lra N$ among the corresponding support sets
satisfying the following three conditions: (a) $\mu$ is an injective
function; (b) $\mu$ is an algebraic homomorphism, that is for every
$n$-ary function symbol $f$ and for every $a_1, \dots, a_n\in M$, we
have $f^{\cN}(\mu(a_1), \dots, \mu(a_n))= \mu(f^{\cM}(a_1, \dots,
a_n))$; (c) $\mu$ preserves and reflects predicates, i.e.\ for every
$n$-ary predicate symbol $P$, we have $(a_1, \dots, a_n)\in P^{\cM}$
iff $(\mu(a_1), \dots, \mu(a_n))\in P^{\cN}$.

If $M\subseteq N$ and the embedding $\mu: \cM \lra \cN$ is just the
identity inclusion $M\subseteq N$, we say that $\cM$ is a {\it
  substructure} of $\cN$ or that $\cN$ is an {\it superstructure} of
$\cM$. Notice that a substructure of $\cN$ is nothing but a subset of
the support set of $\cN$ which is closed under the $\Sigma$-operations
and whose $\Sigma$-structure is inherited from $\cN$ by
restriction. In fact, given $\cN=(N, \mathcal J)$ and $G\subseteq N$,
there exists the smallest substructure of $\cN$ containing $G$ in its
support set. This is called the substructure \emph{generated by $G$}
and its support set can be characterized as the set of the elements
$b\in N$ such that $t^{\mathcal N}(\ua)=b$ for some $\Sigma$-term $t$
and some finite tuple $\ua$ from $G$ (when we write $t^{\mathcal
  N}(\ua)=b$, we mean that $(\cN, \mathtt{a}) \models t(\ux)=y$ for an
assignment $\mathtt{a}$ mapping the $\ua$ to the $\ux$ and $b$ to
$y$).

Below, we will make frequent use of the easy---but fundamental---fact
that the truth of a universal (resp. existential) sentence is
preserved through substructures (resp. through superstructures).  A
\emph{universal} (resp. \emph{existential}) sentence is obtained by
prefixing a string of universal (resp. existential) quantifiers to a
quantifier-free formula.

\subsection{A many-sorted framework}
\label{subsec:many-sort}
From now on, we use many-sorted first-order logic.  All notions
introduced above can be easily adapted to a many-sorted framework.
\textbf{In the rest of the paper, we fix} (i) a theory $T_I=(\Sigma_I,
\Cc_I)$ whose only sort symbol is \texttt{INDEX}; (ii) a
theory $T_E=(\Sigma_E, \Cc_E)$ for data whose only sort symbol is
\texttt{ELEM} (the class $\Cc_E$ of models of this theory is usually a
singleton).
The \textbf{theory $A_I^E=(\Sigma, \Cc)$ of arrays with indexes in $T_I$
  and elements in $T_E$} is obtained as the combination of $T_I$ and $T_E$
as follows: \texttt{INDEX}, \texttt{ELEM}, and \texttt{ARRAY} are the
only sort symbols of $A_I^E$, the signature is $\Sigma := \Sigma_I\cup
\Sigma_E\cup \{ \_[\_] \}$ where $\_[\_]$ has type $\mathtt{ARRAY},
\mathtt{INDEX} \longrightarrow \mathtt{ELEM}$ (intuitively, $a[i]$
denotes the element stored in the array $a$ at index $i$); a
three-sorted structure $\cM = ({\tt INDEX}^{\cM}, {\tt ELEM}^{\cM},
{\tt ARRAY}^{\cM}, ~\int)$ is in $\Cc$ iff ${\tt ARRAY}^{\cM}$ is the
set of (total) functions from ${\tt INDEX}^{\cM}$ to ${\tt
  ELEM}^{\cM}$, the function symbol $\_[\_]$ is interpreted as
function application, and $\cM_I=(\mathtt{INDEX}^{\cM},\int_{\vert
  \Sigma_I})$, $\cM_E=(\mathtt{ELEM}^{\cM},\int_{\vert \Sigma_E})$ are
models of $T_I$ and $T_E$, respectively (here $\int_{\vert \Sigma_X}$
is the restriction of the interpretation $\int$ to the symbols in
$\Sigma_X$ for $X\in\{I,E\}$).

\paragraph{ {\bf Notational conventions.}}  For
the sake of brevity, we
introduce the following notational conventions: $d,e$ range over
variables of sort {\tt ELEM}, $a$ over variables of sort {\tt ARRAY},
$i,j,k,$ and $z$ over variables of sort {\tt INDEX}.  An underlined
variable name abbreviates a tuple of variables of unspecified (but
finite) length and, if $\ui:=i_1, \dots, i_n$, the notation $a[\ui]$
abbreviates the tuple of terms $a[i_1], \dots, a[i_n]$.  Possibly
sub/super-scripted expressions of the form
$\phi(\ui,\ue),\psi(\ui,\ue)$ denote \emph{\textbf{ quantifier-free $(\Sigma_I
  \cup \Sigma_E)$-formulae} in which at most the variables $\ui\cup\ue$
  occur}.
Also, $\phi(\ui, \ut/\ue)$ (or simply $\phi(\ui, \ut)$)
abbreviates the substitution of the $\Sigma$-terms $\ut$ for the variables
$\ue$.
 Thus, for
instance, \emph{$\phi(\ui, a[\ui])$ denotes the formula obtained by
replacing $\ue$ with $a[\ui]$ in a quantifier-free formula
$\phi(\ui, \ue)$}.

\section{Backward Reachability}
\label{sec:array}
Following~\cite{avocs08}, we focus on a particular yet large class of
array-based systems corresponding to guarded assignments.

\subsection{Array-based Systems}
\label{subsec:array-based}
A \emph{(guarded assignment) array-based (transition) system (for
  $(T_I,T_E)$)} is a triple $\cSi=(a,I,\tau)$ where (i) $a$ is the
\emph{state} variable of sort \texttt{ARRAY};\footnote{For the sake of
  simplicity, we limit ourselves to array-based systems having just
  one variable $a$ of sort $\mathtt{ARRAY}$. All the definitions and
  results can be easily generalized to the case of several variables of
  sort $\mathtt{ARRAY}$.  In the examples, we will consider cases
  where more than one variable is required and, in addition, the
  theory $T_E$ is many-sorted.}  (ii) $I(a)$ is the \emph{initial}
$\Sigma(a)$-formula; and (iii) $\tau(a, a')$ is the \emph{transition}
$(\Sigma\cup \Sigma_D)(a, a')$-formula, where $a'$ is a renamed copy
of $a$ and $\Sigma_D$ is a finite set of case-defined function symbols
not in $\Sigma_I\cup \Sigma_E$.  Below, we also \textbf{assume $I(a)$
  to be a \emph{$\forall^I$-formula}}, i.e.\ a formula of the form
$\forall \ui.\phi(\ui, a[\ui])$, and \textbf{$\tau(a,a')$ to be in
  \emph{functional form}}, i.e.\ a \emph{disjunction of} formulae of
the form
\begin{equation}
  \label{eq:transition2}
  \exists \ui\,(\phi_L(\ui, a[\ui]) \wedge
    \forall j\,a'[j]=F_G(\ui, a[\ui], j, a[j]))
\end{equation}
where $\phi_L$ is the \emph{guard} (also called the local component
in~\cite{ijcar08}), and $F_G$ is a case-defined function (called the
\emph{global} component in~\cite{ijcar08}).  To understand why we say
that formulae \eqref{eq:transition2} are `in functional form',
consider $\lambda$-abstraction; then, the sub-formula $\forall
j\,a'[j]=F_G(\ui, a[\ui], j, a[j]))$ can be re-written as $a'=\lambda
j. F_G(\ui, a[\ui], j, a[j])$.
In~\cite{ijcar08}, we adopted a more liberal format for transitions;
the format of this paper, however, is sufficient to formalize all
relevant examples we met so far. Results in this paper extend in a
straightforward way to the case in which $T_E$ is assumed to have
quantifier elimination and~\eqref{eq:transition2} is allowed to have
existentially quantified variables ranging over data.
This extension is crucial to formalize, e.g., non-deterministic
updates
or timed networks~\cite{verify}.

Given an array-based system $\cSi=(a,I,\tau)$ and a formula $U(a)$,
(an instance of) the \emph{safety problem} is to establish whether
there exists a natural number $n$ such that the formula
\begin{equation}
  \label{eq:unsafe}
  I(a_0)\wedge \tau(a_0, a_1)\wedge \cdots \wedge \tau(a_{n-1}, a_n)\wedge U(a_n)
\end{equation}
is $A^E_I$-satisfiable. If there is no such $n$, then $\cSi$ is
\emph{safe} (w.r.t. $U$); otherwise, it is \emph{unsafe} since the
$A^E_I$-satisfiability of \eqref{eq:unsafe} implies the existence of a
run (of length $n$) leading the system from a state in $I$ to a state
in $U$.  From now on, we \textbf{assume $U(a)$ to be a
  \emph{$\exists^I$-formula}}, i.e.\ a formula of the form $\exists
\ui.\phi(\ui, a[\ui])$.

We illustrate the above notions by considering the Mesi cache
coherence protocol, taken from the extended version of~\cite{tacas06}.
\begin{example}\em
  \label{ex:one}
  Let $T_I$ be the pure theory of equality and $T_E$ be the enumerated
  data-types theory with four constants denoted by the numerals from
  $1$ to $4$.  Each numeral corresponds to a control location of a
  cache: $1$ to $\mathtt{modified}$, $2$ to $\mathtt{exclusive}$, $3$
  to $\mathtt{shared}$, and $4$ to $\mathtt{invalid}$.

  Initially, all caches are $\mathtt{invalid}$ and the formula
  characterizing the set of initial states is $\forall i.\; a[i]=4$.
  There are four transitions.  In the first (resp. second) transition,
  a cache in state $\mathtt{invalid}$ (resp. $\mathtt{shared}$) goes
  to the state $\mathtt{exclusive}$ and invalidates all the other
  caches. Formally, these can be encoded with formulae as follows:
  \begin{eqnarray*}
    \exists i.~ (a[i]= 4 \wedge
        a'=\lambda j.\; (\mathtt{if}~(j=i) ~\mathtt{then}~2~\mathtt{else}~4)) & \mbox{and } \\
    \exists i.~ (a[i]= 3 \wedge
        a'=\lambda j.\; (\mathtt{if}~(j=i) ~\mathtt{then}~2~\mathtt{else}~4)) & .
  \end{eqnarray*}
  In the third transition, a cache in state $\mathtt{invalid}$ goes to
  the state $\mathtt{shared}$ and so do all other caches:
  \begin{eqnarray*}
    \exists i.~ (a[i]= 4 \wedge
        a'=\lambda j.\; 3) .
  \end{eqnarray*}
  In the fourth and last transition, a cache in state
  $\mathtt{exclusive}$ can move to the state $\mathtt{modified}$ (the
  other caches maintain their current state):
  \begin{eqnarray*}
    \exists i.~ (a[i]= 2 \wedge
        a'=\lambda j.\; (\mathtt{if}~(j=i) ~\mathtt{then}~1~\mathtt{else}~a[j])) .
  \end{eqnarray*}
  To be safe, the protocol should not reach a state in which there is
  a cache in state $\mathtt{modified}$ and another cache in state
  $\mathtt{modified}$ or in state $\mathtt{shared}$.  Thus, one can
  take
  \begin{eqnarray*}
    \exists i_1\, \exists i_2.~(i_1\neq i_2 \wedge a[i_1]= 1 \wedge ( a[i_2] = 1 \vee a[i_2]= 3))
  \end{eqnarray*}
  as the unsafety formula. \qed
\end{example}
  The reader with some experience in infinite state model checking may
  wonder how it is possible to encode in our framework transitions
  with `global conditions,' i.e.\ guards requiring a universal
  quantification over indexes.  Indeed, the
  format~\eqref{eq:transition2} for transitions is clearly too
  restrictive for this purpose.  However, it is possible to overcome
  this limitation by using the \emph{stopping failures model}
  introduced in the literature about distributed algorithms (see,
  e.g.,~\cite{dalg}): according to this model, processes may crash at
  any time and do not play any role in the rest of the execution of
  the protocol (they ``disappear'').  In this model, there is no need
  to check the universal conditions of a transition, rather the
  transition is taken and any process not satisfying the global
  condition is assumed to crash.  In this way, we obtain an
  over-approximation of the original system admitting more runs and
  any safety certification obtained for this over-approximation is
  also a safety certification for the original model.  Indeed, the
  converse is not always true and spurious error traces may be
  obtained.  Interestingly, the approximated model can be obtained
  from the original system by simple syntactical transformations of
  the formulae encoding the transitions requiring the universal
  conditions. For more details concerning the implementation of the
  approximated model in \textsc{mcmt}, the reader is referred
  to~\cite{stop}.  A more exhaustive discussion of the use of a
  similar approximated model can be found in ~\cite{tacas06,cav06,approx}.

\subsection{Backward Reachable States}
\label{subsec:brs}
A general approach to solve instances of the safety problem is based
on computing the set of backward reachable states.  For $n\geq 0$, the
$n$-\emph{pre-image} of a formula $K(a)$ is $Pre^0(\tau,K) := K$ and
$Pre^{n+1}(\tau, K) := Pre(\tau, Pre^n(\tau, K))$, where
\begin{eqnarray}
  \label{eq:def-pre}
  Pre(\tau, K) & := & \exists a'.(\tau(a,a') \wedge K(a')).
\end{eqnarray}
Given $\cSi=(a,I,\tau)$ and $U(a)$, the formula $Pre^n(\tau, U)$
describes the set of backward reachable states in $n$ steps (for
$n\geq 0$).
\begin{figure}[tb]
  \begin{center}
  \begin{tabular}{ccc}
    \begin{minipage}{.45\textwidth}
      \begin{tabbing}
        foo \= foo \= \kill
        \textbf{function} $\mathsf{BReach}(U ~:~ \exists^I\mbox{-formula})$ \\
        1 \> $P\longleftarrow U$;  $B\longleftarrow \bot$; \\
        2\> \textbf{while} ($P\wedge \neg B$ is $A^E_I$-sat.) \textbf{do}\\
        3\>\> \textbf{if} ($I\wedge P$ is $A^E_I$-sat.) \\
         \>\> \hspace{.75cm} \textbf{then return}  $\mathsf{unsafe}$;\\
        4\> \> $B\longleftarrow P\vee B$; \\
        5\>\> $P\longleftarrow Pre(\tau, P);$ \\
        6\> \textbf{end} \\
        7\> \textbf{return} $(\mathsf{safe}, B);$
      \end{tabbing}
    \end{minipage}
    & \hspace{.35cm} &
    \begin{minipage}{.45\textwidth}
      \begin{tabbing}
        foo \= foo \= \kill
        \textbf{function} $\mathsf{SInv}(U ~:~ \exists^I\mbox{-formula})$ \\
        1 \> $P\longleftarrow \mathsf{ChooseCover}(U)$;
           $B\longleftarrow \bot$; \\
        2\> \textbf{while} ($P\wedge \neg B$ is $A^E_I$-sat.) \textbf{do}\\
        3\>\> \textbf{if} ($I\wedge P$ is $A^E_I$-sat.) \\
         \>\>  \hspace{.75cm} \textbf{then return} $\mathsf{failure}$;  \\
        4\> \> $B\longleftarrow P\vee B$; \\
        5\>\> $P\longleftarrow \mathsf{ChooseCover}(Pre(\tau, P));$ \\
        6\> \textbf{end} \\
        7\> \textbf{return} $(\mathsf{success}, \neg B)$;
      \end{tabbing}
    \end{minipage} \\
    (a) & & (b)
  \end{tabular}
  \end{center}
  \caption{\label{fig:reach-algo}%
    The basic backward reachability (a) and the invariant synthesis
    (b) algorithms}
\end{figure}
At the (end of) $n$-th iteration of the loop, the \emph{basic backward
  reachability algorithm}, depicted in Figure~\ref{fig:reach-algo}
(a), stores in the variable $B$ the formula $BR^n(\tau,
U):=\bigvee^n_{i=0} Pre^i(\tau, U)$ representing the set of states
which are backward reachable from the states in $U$ in at most $n$
steps (whereas the variable $P$ stores the formula $Pre^{n+1}(\tau, U)$).
While computing $BR^n(\tau, U)$, $\mathsf{BReach}$ also checks whether
the system is unsafe (cf.\ line 3, which can be read as `$I\wedge
Pre^n(\tau, U)$ is $A_I^E$-satisfiable') or a fix-point has been
reached (cf.\ line 2, which can be read as `$\neg(BR^{n+1}(\tau,
U)\rightarrow BR^{n}(\tau, U))$ is $A_I^E$-satisfiable' or,
equivalently, that `$(BR^{n+1}(\tau, U)\rightarrow BR^{n}(\tau, U))$ is
not $A_I^E$-valid').
%
When $\mathsf{BReach}$ returns the safety of the
system (cf.\ line 7), the variable $B$ stores the formula describing
the set of states which are backward reachable from $U$ which is also
a fix-point.

Indeed, for $\mathsf{BReach}$ (Figure~\ref{fig:reach-algo} (a)) to be
a true (possibly non-terminating) procedure, it is mandatory that (i)
$\exists^I$-formulae are closed under pre-image computation and (ii)
both the $A^E_I$-satisfiability test for safety (line 3) and that for
fix-point (line 2) are effective.

Concerning (i), it is sufficient to use the following result
from~\cite{avocs08}.\footnote{The proposition may be read as the
  characterization of a weakest liberal pre-condition
  transformer~\cite{wlp} for array-based systems.}
\begin{proposition}
  \label{prop:trans}
  Let $K(a):=\exists \uk \, \phi(\uk, a[\uk])$
  and $\tau(a,a'):=\bigvee_{h=1}^m\exists \ui \;(\phi^h_L(\ui,
  a[\ui])\wedge a'=\lambda j. F^h_G(\ui, a[\ui], j, a[j]) )$.
  Then, $Pre(\tau, K)$ is $A_I^E$-equivalent to an (effectively
  computable) $\exists^I$-formula.
\end{proposition}
\begin{proof}
  Let $\tau_h$ be one of the $m$ disjuncts of $\tau$. Using the
  $\lambda$-abstraction formulation and a single $\beta$-reduction
  step, it is clear that $Pre(\tau_h, K)$ is $A_I^E$-equivalent to the
  following $\exists^I$-formula
  \begin{eqnarray}
    \label{eq:pre}
    \begin{array}{r}
      \exists \ui\, \exists \uk.(
      \phi^h_L(\ui, a[\ui]) \wedge \phi(\uk, F_G^h(\ui,a[\ui],\uk, a[\uk])))
    \end{array}
  \end{eqnarray}
  where $\uk$ is the tuple $k_1, \dots, k_l$ and $\phi(\uk,
  F_G^h(\ui,a[\ui],\uk, a[\uk]))$ is the formula obtained from
  $\phi(\uk, a'[\uk])$ by replacing $a'[k_s]$ with
  $F_G^h(\ui,a[\ui],k_s, a[k_s])$, for $s=1, ..., l$.  Now it is
  sufficient to eliminate the $F^h_G$ as shown in
  Lemma~\ref{lem:extensions}. As a final step, the existential
  quantifiers can be moved in front of the disjunction arising from
  the $m$ disjuncts $\tau_1, ..., \tau_m$.
\end{proof}
The proof and the algorithm underlying Proposition~\ref{prop:trans}
are quite simple.  This is in sharp contrast to most approaches to
infinite state model checking available in the literature
(e.g.,~\cite{tacas06,cav06}) that use special data structures (such as
strings with constraints) to represent sets of states. These special
data structures can be considered as normal forms when compared to our
formulae.  In this respect, our framework is more flexible
since---although it can use normal forms (when these can be cheaply
computed)---it is not obliged to do so.
The drawback is that safety and fix-point checks may become
computationally much more expensive.  In particular, the bottle-neck
is the handling of the quantified variables in the prefix of
$\exists^I$-formulae which may become quite large at each pre-image
computation: notice that the prefix $\exists\;\uk$ is augmented with
$\exists\;\ui$ in~\eqref{eq:pre} with respect to $K$.  This and other
issues which are relevant for the implementation of our framework are
discussed in~\cite{avocs08,afm09,ijcar10}.

Concerning the mechanization of the safety and fix-point checks (point
(ii) above), observe that the formulae involved in the satisfiability
checks are $I \wedge BR^n(\tau,K)$ and
$\neg(BR^n(\tau, U)\rightarrow BR^{n-1}(\tau, U))$.  Since we have
closure under pre-image computation, both formulae are of the form $
\exists \ua\; \exists \ui\; \forall \uj\; \psi(\ui, \uj, \ua[\ui],
\ua[\uj])$, where $\psi$ is quantifier free: we call these sentences
\emph{$\exists^{A,I}\forall^I$-sentences}~\cite{ijcar08}.
 \begin{thm}
   \label{th:decidability}
   The $A_I^E$-satisfiability of $\exists^{A,I}\forall^I$-sentences is
   decidable if (I) $T_I$ is locally finite and is closed under
   substructures\footnote{
     By this we mean that if $\cM$ is a model of $T_I$ and $\cN$ is a
     substructure of $\cM$, then $\cN$ is a model of $T_I$ as well.  }
   and (II) the $SMT(T_I)$ and $SMT(T_E)$ problems are
   decidable. Under the same hypotheses, it holds that an
   $\exists^{A,I}\forall^I$-sentence is $A_I^E$-satisfiable iff it is
   satisfiable in a finite index model (a \emph{finite index model} is
   a model $\cM$ in which the set $\texttt{INDEX}^{\cM}$ has finite
   cardinality).
\end{thm}
A generalization of Theorem~\ref{th:decidability} can be found in the
extended version of~\cite{ijcar08} and is reported in Appendix A (with
a proof) to make this paper self-contained.  The proof of
Theorem~\ref{th:decidability} is the starting point to develop a
satisfiability procedure for formulae of the form $ \exists \ua\;
\exists \ui\; \forall \uj\; \psi(\ui, \uj, \ua[\ui], \ua[\uj])$
consisting of the following steps: (a) the variables $\ua, \ui$ are
Skolemized away: (b) the variables $\uj$ are instantiated in all
possible ways by using the representative $\ui$-terms; (c) the
resulting combined problem is purified and an arrangement (i.e.\ an
equivalence class) over the shared index variables is guessed; (d) the
positive literals from this arrangement are propagated to the
$T_E$-literals (this is a variant of the Nelson-Oppen schema adopted
in `theory connections,' see~\cite{BaGh}); (e) finally, the purified
constraints are passed to the theory solvers for $T_I$ and $T_E$,
respectively.  From the implementation viewpoint, powerful heuristics
are needed~\cite{afm09} to keep the potential combinatorial explosion
in step (b) under control.  Fortunately, the adoption of a certain
format for formulae (called, `primitive differentiated,' see below for
details) makes steps (c) and (d) redundant (see~\cite{afm09} for more
on this point).

Hypothesis (I) from Theorem~\ref{th:decidability} concerns the
\emph{topology} of the system (not the data manipulated by the
components of the system)
and its intuitive meaning can be easily explained when the signature
$\Sigma_I$ is relational: in that case, local finiteness is guaranteed
and closure under substructures says that if some elements are deleted
from a model of $T_I$, we still get a model of $T_I$ (i.e.\ the
topology does not change under elimination of elements).  For example,
Hypothesis (I) is true for (finite) sets, linear orders, graphs,
forests, while it does not hold for 'rings,' because, after deleting
one of their elements, they are no more rings.
We emphasize that it is \emph{not possible to weaken} Hypothesis (I)
on the theory $T_I$.  Indeed, it is possible to show that any
weakening yields undecidable fragments of the theory of arrays over
integers~\cite{arrays} (as it is shown in Appendix A).  Furthermore,
we observe that Hypothesis (I) is not too restrictive because, as said
above, it concerns only the topology of the system.
%
So, for example, the topology of virtually any cache coherence
protocol (see Example~\ref{ex:one}) can be formalized by finite sets
while that of
standard mutual exclusion protocols by linear orders.

We summarize our working hypotheses in the following.
\begin{assumption}
  We fix an array-based system $\cSi=(a,I,\tau)$ such that the initial
  formula $I$ is a $\forall^I$-formula, and the transition formula
  $\tau(a,a')$ is $\bigvee_{h=1}^m \tau_h(a,a')$, where $\tau_h$ is a
  formula
  of the form~\eqref{eq:transition2} for $h=1, ..., m$.
  We suppose that  $\exists$-formulae are used to describe the set of
  unsafe states.
  Finally, we assume that hypotheses (I) and (II) of
  Theorem~\ref{th:decidability} are satisfied.%
\end{assumption}

\subsection{Tableaux-like Implementation of Backward Reachability}
\label{subsec:tab}

A naive implementation of the algorithm in Figure~\ref{fig:reach-algo}
(a) does not scale up.  The main problem is the size of the formula
$BR^n(\tau,U)$ which contains many redundant or unsatisfiable
sub-formulae.  We now discuss how Tableaux-like techniques can be used
to circumvent these difficulties.  We need one more definition: an
$\exists^I$-formula $\exists i_1\cdots \exists i_n \phi$ is said to be
\emph{primitive} iff $\phi$ is a conjunction of literals and is said
to be \emph{differentiated} iff $\phi$ contains as a conjunct the
negative literal $i_k\not= i_l$ for all $1\leq k<l\leq n$. By applying
various distributive laws together with the rewriting rules
\begin{eqnarray}
  \label{eq:rew}
  \exists j (i=j\wedge \theta) \leadsto \theta(i/j)
  & \mbox{ and } &
  \theta \leadsto (\theta \wedge i=j)\vee (\theta \wedge i\not= j)
\end{eqnarray}
it is always possible to transform every $\exists^I$-formula into a
disjunction of primitive differentiated ones.

We initialize our tableau with the $\exists^I$-formula $U(a)$
representing the set of unsafe states.  The key observation is to
revisit the computation of the pre-image as the following inference
rule (we use square brackets to indicate the applicability  condition of the rule):
\begin{displaymath}
  \infer[\mathsf{PreImg}]
   {Pre(\tau_1, K) ~|~ \cdots ~|~ Pre(\tau_m, K)}
   {K~~[\mbox{$K$ is primitive differentiated}]}
\end{displaymath}
where $Pre(\tau_h, K)$ computes the $\exists^I$-formula which is
$A^E_I$-equivalent to the pre-image of $K$ w.r.t.\ $\tau_h$ (this is
possible according to the proof of Proposition~\ref{prop:trans}).

Since the $\exists^I$-formulae labeling the consequents of the rule
$\mathsf{PreImg}$ may not be primitive and differentiated, we need the
following $\mathsf{Beta}$ rule
\begin{displaymath}
  \infer[\mathsf{Beta}]
  { K_1 ~|~ \cdots ~|~ K_n}
   {K}
\end{displaymath}
where $K$ is  transformed  by applying rewriting rules like
\eqref{eq:rew} together with standard distributive laws, in order to
get $K_1, \dots, K_n$ which are primitive, differentiated and whose
disjunction is $A^E_I$-equivalent to $K$.

By repeatedly applying the above rules, it is possible to build a tree
whose nodes are labelled by $\exists^I$-formulae describing the set of
backward reachable states.  Indeed, it is not difficult to see that
the disjunction of the $\exists^I$-formulae labelling all the nodes in
the (potentially infinite) tree is $A_I^E$-equivalent to
the (infinite) disjunction of the formulae $
BR^n(\tau,U)$, where $\tau:=\bigvee_{h=1}^m \tau_h$.  Indeed, there is
no need to fully expand our tree.  For example, it is useless to apply
the rule $\mathsf{PreImg}$ to a node $\nu$ labelled by an
$\exists^I$-formula which is $A_I^E$-unsatisfiable as all the formulae
labelling nodes in the sub-tree rooted at $\nu$ will also be
$A_I^E$-unsatisfiable.  This observation can be formalized by the
following rule closing a branch in the tree (we mark the terminal node
of a closed branch by $\times$):
\begin{displaymath}
  \infer[\mathsf{NotAppl}]{ \times
}{K\mbox{ ~~[$K$ is $A_I^E$-unsatisfiable}]}
\end{displaymath}
This rule is effective since $\exists^I$-formulae are a subset of
$\exists^{A,I}\forall^I$-sentences and the $A_I^E$-sa\-tis\-fia\-bi\-li\-ty of
these formulae is decidable by Theorem~\ref{th:decidability}.

According to procedure $\mathsf{BReach}$, there are two more
situations in which we can stop expanding a branch in the tree.  One
terminates the branch because of the safety test (cf.\ line 3 of
Figure~\ref{fig:reach-algo} (a)):
\begin{displaymath}
  \infer[\mathsf{Safety}]{\mathsf{UnSafe} }
        {K ~~[I\wedge K\mbox{ is $A_I^E$-satisfiable}]}
\end{displaymath}
Interestingly, if we label with $\tau_h$ the edge connecting a node
labeled with $K$ with that labeled with $Pre(\tau_h, K)$ when applying
rule $\mathsf{PreImg}$, then the transitions $\tau_{h_1}, ...,
\tau_{h_e}$ labelling the edges in the branch terminated by
$\mathsf{UnSafe}$ (from the leaf node to the root node) give a
\emph{error trace}, i.e.\ a sequence of transitions leading the
array-based system from a state satisfying $I$ to one satisfying $U$.
Again, rule $\mathsf{UnSafe}$ is effective since $I\wedge K$ is
equivalent to an $\exists^{A,I}\forall^I$-sentence and its
$A_I^E$-satisfiability is decidable by Theorem~\ref{th:decidability}.
The other situation in which one can close a branch corresponds to the
fix-point test (cf.\ line 2 of Figure~\ref{fig:reach-algo} (a))
\begin{displaymath}
  \infer[\mathsf{FixPoint}]{\times }%
  {K~~[K\wedge \bigwedge \{\neg K'\vert K'\preceq K\}\mbox{ is $A_I^E$-unsatisfiable}]}
\end{displaymath}
where $K'\preceq K$ means that $K'$ is a primitive differentiated
$\exists^I$-formula labeling a node preceding the node labeling $K$
(nodes can be ordered according to the strategy for expanding the
tree).
Once more, this rule is effective since $K\wedge \bigwedge \{\neg
K'\vert K'\preceq K\}$ can be straightforwardly transformed into an
$\exists^{A,I}\forall^I$-sentence
whose $A_I^E$-satisfiability is decidable by
Theorem~\ref{th:decidability}.

As mentioned above, from the implementation point of view, clever
heuristics are nee\-ded to reduce the instances that have to be
generated for the satisfiability test of Theorem~\ref{th:decidability}
and to trivialize the recognition of the unsatisfiable premise of the
rule $\mathsf{NotAppl}$.  In addition, the satisfiability checks
required by Rule $\mathsf{FixPoint}$ should be performed
\emph{incrementally} by considering formulae in reverse chronological
order (i.e.\ the pre-images generated later are added first and those
generated early are possibly added later).  The interested reader is
pointed to~\cite{afm09} for a more exhaustive discussion about these
issues.

A final remark is in order.  One may think that the main difference
between our framework to model checking infinite state systems and
other approaches lies just in the technology used for constraint
solving; our system, \textsc{mcmt}, uses an SMT solver while other
tools (such as \textsc{pfs}~\cite{tacas06}) use efficient dedicated
algorithms.
This is only part of the story.  In fact, \textsc{mcmt} usually
produces many fewer nodes while visiting the tree whose nodes are
labelled with the formulae representing sets of backward reachable
states, compared to other systems.  This is so because our approach is
fully declarative and \textsc{mcmt} \emph{symbolically represents also
  the topology of the system}, not only the data.  The other model
checkers use constraints only to represent the data manipulated by the
system while the topology is encoded by using an \emph{ad hoc} data
structure, which usually requires more effort to represent sets of
states.  To illustrate this fundamental aspect, we consider a simple
(but tricky) example.
\begin{example}
  \em Let $T_I$ be the theory of linear orders and $T_E$ be an
  enumerated data-type with 15 constants denoted by the numerals from
  1 to 15.  Consider the following parametrized system having 7
  transitions and 15 control locations:
  \begin{enumerate}[$\bullet$]
  \item the first transition allows process $i$ to move from location
    1 to location 2 provided there is a process $j$ to the right of
    $i$ (i.e.\ $i<j$ holds) which is on location 9;
  \item similarly, the second transition allows process $i$ to move
    from location 2 to location 3 provided there is a process $j$ to
    the right of $i$ which is on location 10, and so on (the last
    transition allows process $i$ to move from location 7 to location
    8 provided there is a process $j$ to the right of $i$ which is on
    location 15).
  \end{enumerate}
  Initially, all processes are in location 1.  We consider the
  following safety problem: is it possible for a process to reach
  location 8? The answer is obviously no.

  \textsc{mcmt} solves the problem by generating 7 nodes in about 0.02
  seconds on a standard laptop.
  On the contrary, \textsc{pfs} takes about 4 minutes on the same
  computer and generates thousands of constraints.  Why is this so?
  The point is that tools like \textsc{pfs}
  do not symbolically represent the system topology and need to
  specify the relative positions of all the involved processes.  In
  contrast, \textsc{mcmt} can handle partial information like ``there
  exist 7 processes to the right of $i$ whose locations are from 9 to
  15, respectively'' just because it is based on a deductive engine,
  i.e.\ the SMT solver.

  Thus, \textsc{mcmt} represents a fully declarative approach to
  infinite state model checking that, when coupled with appropriate
  heuristics, should pave the way to the verification of systems with
  more and more complex topologies that other tools cannot
  handle. \qed
\end{example}

\section{Termination: a semantic analysis}
\label{sec:term}
Termination of our tableaux calculus (and of the algorithm of
Figure~\ref{fig:reach-algo} (a)) is not guaranteed in general as
safety problems are undecidable even when the data structures
manipulated by the system are simple (Sec.~\ref{subsec:undec}).
However, it is possible to identify sufficient conditions to obtain
termination (Sec.~\ref{subsec:dec}) which are useful in some
applications.  We begin by introducing an important definition to be
used in this and the following section.

\subsection{Configurations}
\label{subsec:configurations}
A \emph{state} of our array-based system $\cSi=(a,I,\tau)$ is a pair
$(s, \cM)$, where $\cM$ is a model of $A_I^E$ and
$s\in\texttt{ARRAY}^{\cM}$.  By recalling the last part of the
statement of Theorem~\ref{th:decidability}, we can focus on a
sub-class of the states (often called configurations) restricting
$\cM$ to be a finite index model.  Formally, an
\emph{$A_I^E$-configuration} (or, simply, a \emph{configuration}) is a
pair $(s, \cM)$ such that $s$ is an array of a finite index model
$\cM$ of $A_I^E$ ($\cM$ is omitted whenever it is clear from the
context).  We associate a $\Sigma_I$-structure $s_I$ and a
$\Sigma_E$-structure $s_E$ with an $A_I^E$-configuration $(s, \cM)$ as
follows: the $\Sigma_I$-structure $s_I$ is simply the finite structure
$\cM_I$, whereas $s_E$ is the smallest $\Sigma_E$-substructure of
$\cM_E$ containing the image of $s$ (in other words, if ${\tt
  INDEX}^{\cM}=\{c_1, \dots, c_k\}$, then $s_E$ is the smallest
$\Sigma_E$-substructure containing $\{s(c_1), \dots, s(c_k)\}$).

\subsection{Undecidability of the safety problem}
\label{subsec:undec}
In the general case, safety problems are undecidable. The result is
not surprising and we report it in the following for the sake of
completeness.
\begin{thm}
  \label{th:undecidability}
  The problem: ``given an $\exists^I$-formula $U$, deciding whether
  the array-based system $\cSi$ is safe w.r.t. $U$'' is undecidable
  (even if $T_E$ is locally finite).
\end{thm}
The proof consists in a rather straightforward reduction from the
reachability problem of Minsky machines.  See Appendix A for details.

\subsection{Decidability of the safety problem: sufficient conditions}
\label{subsec:dec}

A specific feature of array-based systems is that a \emph{partial
  ordering among configurations} can be defined.  This is the key
ingredient in establishing the termination of the backward
reachability procedure (and thus the decidability of the related
safety problem) and characterizing the completeness of invariant
synthesis strategies (as it will be shown in Section~\ref{sec:inv}
below).

A \emph{pre-order} $(P, \leq)$ is a set endowed with a reflexive and
transitive relation; an \emph{upset}, also called an \emph{upward
  closed set}, of such a pre-order is a subset $U\subseteq P$ such
that ($p\in U$ and $p\leq q$ imply $q\in U$). An upset $U$ is
\emph{finitely generated} iff it is a finite union of cones, where a
\emph{cone} is an upset of the form $\uparrow\!p=\{ q\in P \mid p\leq
q\}$ for some $p\in P$.  Two elements $p,q\in P$ are
\emph{incomparable} (\emph{equivalent}) if neither (both) $p\leq q$
nor (and) $q\leq p$.

We are ready to define a \emph{pre-order over configurations}.  Let $s,
s'$ be configurations: we say that $s' \leq s$ holds iff there are a
$\Sigma_I$-embedding $\mu:s'_I\longrightarrow s_I$ and a
$\Sigma_E$-embedding $\nu:s'_E\longrightarrow s_E$ such that the
set-theoretical compositions of $\mu$ with $s$ and of $s'$ with $\nu$
are equal.  This is depicted in the following diagram:
\begin{center}
 \resetparms
 \setsqparms[+2`+1`+1`+2;500`500]
 \square[s'_I` s_I`s'_E`s_E;\mu`s'`s`\nu]
\end{center}
In case $\mu$ and $\nu$ are both inclusions, we say that $s'$ is a
\emph{sub-configuration} of $s$.

Finitely generated upsets of configurations and $\exists^I$-formulae
can be used interchangeably under suitable assumptions.  Let $K(a)$ be
an $\exists^I$-formula; we let $\mywidehat{K} :=\{ (s,\cM)\mid
\cM\models K(s)\}$.
\begin{proposition}
  \label{prop:conf}
  For every $\exists^I$-formula $K(a)$, the set $\mywidehat{K}$ is
  upward closed.  For every $\exists^I$-formulae $K_1, K_2$, we have
  that $\mywidehat{K_1}\subseteq \mywidehat{K_2}$ iff $A_I^E\models
  K_1\to K_2$.
\end{proposition}
\begin{proof}
  Let us first show that the set $\mywidehat{K}$ is upward closed. By
  using disjunctive normal forms and distributing existential
  quantifiers over disjunctions, we can suppose---without loss of
  generality--that $K(a)$ is of the form $\exists \ui \phi(\ui,
  a[\ui])$, where $\phi$ is a conjunction of $\Sigma_I\cup
  \Sigma_E$-literals (the general case follows from this one because a
  union of upsets is an upset). If we also separate $\Sigma_I$- and
  $\Sigma_E$-literals, we can suppose that $\phi(\ui, a[\ui])$ is of
  the kind $\phi_I(\ui)\wedge \phi_E(a[\ui])$, where $\phi_I$ is a
  conjunction of $\Sigma_I$-literals and $\phi_E$ is a conjunction of
  $\Sigma_E$-literals. Suppose now that $(s, \cM)$ and $(t, \cN)$ are
  configurations such that $s\leq t$ and $\cM\models K(s)$: we wish to
  prove that $\cN\models K(t)$. From $\cM\models K(s)$, it follows
  that there are elements $\ui$ from ${\tt INDEX}^{\cM}$ such that
  $\cM\models \phi_I(\ui)\wedge\phi_E(s[\ui])$, i.e. such that
  $s_I\models \phi_I(\ui)$ and $s_E\models \phi_E(s(\ui))$ (to infer
  the latter, recall that the operations $a[\ui]$ are interpreted as
  functional applications in our models and also that truth of
  quantifier free formulae is preserved when considering
  substructures). Now $s\leq t$ says that there are embeddings
  $\mu:s_I\longrightarrow t_I$ and $\nu:s_E\longrightarrow t_E$ such
  that $\nu\circ s= t\circ \mu$. Since truth of quantifier free
  formulae is preserved when considering superstructures, we get
  $t_I\models\phi_I(\mu(\ui))$ and $t_E\models\phi_E(\nu(s(\ui))$
  (that is, $t_E\models\phi_E(t(\mu(\ui)))$) and also $\cN\models
  \phi_I(\mu(\ui))\wedge \phi_E(t[\mu(\ui)])$, which implies
  $\cN\models K(t)$, as desired.

  Let us now prove the second claim of the Proposition. That
  $A^E_I\models K_1\to K_2$ implies $\mywidehat{K_1}\subseteq
  \mywidehat{K_2}$ is trivial. Suppose conversely that
  $A^E_I\not\models K_1\to K_2$, which means that $K_1(a)\wedge \neg
  K_2(a)$ is $A^E_I$-satisfiable: since this implies that
  $K_1(a)\wedge \neg K_2(a)$ is satisfiable in a finite index model of
  $A^E_I$ (see Theorem~\ref{th:decidability}), we immediately get that
  $\mywidehat{K_1}\not\subseteq \mywidehat{K_2}$.
\end{proof}
Before continuing, we recall the standard model-theoretic notion of
Robinson diagrams and some related results~ (see, e.g.,~\cite{CK} for
more details).  Let $\cM=(M, \int)$ be a $\Sigma$-structure which is
generated by $G\subseteq M$. Let us take a free variable $x_g$ for
every $g\in G$ and call $G_x$ the set $\{x_g\mid g\in
G\}$.\footnote{
  One may wonder if assuming ``countably many variables'' is too
  restrictive since $G$ may be uncountable.  There are two ways to
  avoid this problem.  First, we can use free constants instead of
  variables (this is the standard solution).  Second, we realize that
  we do not need to consider---in this paper---the case when $G$ is
  uncountable since in all our applications, $G$ is finite.} The
\emph{$\Sigma_G$-diagram $\delta_{\cM}(G)$ of $\cM$} is the set of all
$\Sigma(G_x)$-literals $L$ such $\cM, {\tt a}\models L$, where ${\tt
  a}$ is the assignment mapping $x_g$ to $g$.

The following celebrated result~\cite{CK} is simple, but nevertheless
very powerful and it will be used in the rest of the paper.
\begin{lemma}[Robinson Diagram Lemma]
  \label{lem:robinson}
  Let $\cM=(M, \int)$ be a $\Sigma$-structu\-re which is generated by
  $G\subseteq M$ and $\cN=(N, \mathcal{J})$ be another
  $\Sigma$-structure.  Then, there is a bijective correspondence given
  by
  \begin{equation}\label{eq:diag}
    \mu(g) = {\tt a}(x_g)
  \end{equation}
  (for all $g\in G$) between assignments ${\tt a}$ on $N$ such that
  $\cN, {\tt a}\models \delta_{\cM}(G)$ and $\Sigma$-embeddings
  $\mu:\cM\longrightarrow \cN$.
\end{lemma}
In other words, \eqref{eq:diag} can be used to define $\mu$ from ${\tt
  a}$ and conversely. Notice that an embedding $\mu:\cM
\longrightarrow \cN$ is uniquely determined, in case it exists, by the
image of the set of generators $G$: this is because the fact that $G$
generates $\cM$ implies (and is equivalent to) the fact that every
$c\in M$ is of the kind $t^{\int}(\ug)$, for some term $t$ and some
$\ug$ from $G$.

The diagram $\delta_{\cM}(G)$ usually contains infinitely many
literals, however there are important cases where we can keep it
finite.
\begin{lemma}
  Suppose that $\cM$ is a $\Sigma$-structure (where $\Sigma$ is a
  finite signature), whose support $M$ is finite; then for every set
  $G\subseteq M$ of generators, there are finitely many literals from
  $\delta_{\cM}(G)$ having all remaining literals of $\delta_{\cM}(G)$
  as a logical consequence.
\end{lemma}
\begin{proof} %
  Choose $\Sigma(G_x)$-terms $t_1, \dots, t_n$ such that (under the
  assignment ${\tt a}:x_g\mapsto g$), $M$ is equal to the set of the
  elements assigned by ${\tt a}$ to $t_1, \dots, t_n$ (this is
  possible because the elements of $G$ are generators and $M$ is
  finite); we also include the $x_g$ varying $g\in G$ among the $t_1,
  \dots, t_n$. We can get the desired finite set $S$ of literals by
  taking the set of \emph{atoms} of the form
  \begin{equation*}
    R(t_{i_1}, \dots, t_{i_k}), \quad f(t_{i_1}, \dots, t_{i_k})= t_{i_{k+1}}
  \end{equation*}
  (as well as their negations), which are true in $\cM$ under the
  assignment ${\tt a}$. In fact, modulo $S$, it is easy to see by
  induction on
  the structure of the term $u$ that every $\Sigma(G_x)$-term $u$ is
  equal to some $t_i$; it follows that every literal from
  $\delta_{\cM}(G)$ is a logical consequence of $S$.
\end{proof}
Whenever the conditions of the above Lemma are true, we can take a
finite conjunction and treat $\delta_{\cM}(G)$ as a single formula:
notice that we are allowed to do so whenever $G$ is finite and $\cM$
is a model of a locally finite theory.
\begin{proposition}%
  \label{prop:conf1}
  Let $T_E$ be locally finite.  It is possible to effectively associate
  \begin{enumerate}[{\rm (i)}]
  \item\label{ite:conf1_1} an $\exists^I$-formula $K_{s}$ with every
    $A^E_I$-configuration $(s, \cM)$ such that $\mywidehat{K_{s}} =
    \uparrow\! s$;
  \item\label{ite:conf1_2} a finite set $\{s_1, \dots, s_n\}$ of
    $A^E_I$-configurations with every $\exists^I$-formula $K$ such
    that $K$ is $A^E_I$-equivalent to $K_{s_1}\vee \cdots\vee
    K_{s_n}$.
  \end{enumerate}
  As a consequence of~\eqref{ite:conf1_1} and~\eqref{ite:conf1_2},
  finitely generated upsets of $A^E_I$-configurations coincide with
  sets of $A^E_I$-con\-fi\-gu\-ra\-tions of the kind $\mywidehat{K}$,
  for some $\exists^I$-formula $K$.
\end{proposition}
\begin{proof}
  Ad \eqref{ite:conf1_1}: we take $G, G'$ to be the support of $s_I$
  and the image of the support of $s_I$ under the function $s$,
  respectively; clearly $G$ is a set of generators for $s_I$ and $G'$
  is a set of generators for $s_E$. Let us call the set of variables
  $G_x, G'_x$ as $\ui:=\{i_1, \dots, i_n\}$ and $\ue:=\{e_1, \dots,
  e_n\}$, respectively. We take $K_s$ to be
  \begin{equation}\label{eq:diagformula}
    \exists \ui \, (\delta_{s_I}(\ui) \wedge \delta_{s_E}(a_0[\ui]))
  \end{equation}
  where $a_0$ is a fresh array variable (in other words, we take the
  diagrams $\delta_{s_I}(G), \delta_{s_E}(G')$, make in the latter the
  replacement $\ue \mapsto a_0[\ui]$, take conjunction, and quantify
  existentially over the $\ui$). For every configuration $(t, \cN)$,
  we have that $t\in \mywidehat{K_s}$ iff $\delta_{s_I}(\ui) \wedge
  \delta_{s_E}(a_0[\ui])$ is true in $\cN$ under some assignment ${\tt
    a}$ mapping the array variable $a_0$ to $t$, that is iff there are
  embeddings $\mu: s_I\longrightarrow t_I$ and $\nu:s_E\longrightarrow
  t_E$ as prescribed by Lemma~\ref{lem:robinson} (i.e.\ Robinson
  Diagram Lemma). These embeddings map the generators $G$ onto the
  indexes assigned to the $\ui$ by ${\tt a}$ and the generators $G'$
  to the elements assigned by ${\tt a}$ to the terms $a_0[\ui]$, which
  means precisely that $t\circ \mu= \nu\circ s$. Thus $t\in
  \mywidehat{K_s}$ is equivalent to $s\leq t$, as desired.

  Ad \eqref{ite:conf1_2}: modulo taking disjunctive normal forms, we
  can suppose that $K(a_0)$ is equal to $\exists \ui\,\bigvee_k
  (\phi_k(\ui)\wedge \psi_k(a_0[\ui]))$, where the $\phi_k$'s are
  $\Sigma_I$-formulae, the $\psi_k$'s are $\Sigma_E$-formulae, and
  $\ui:=i_1, \dots, i_m$.  Since $T_I$ is locally finite, we can
  assume that for every representative $\ui$-term $t$ there is an
  $i_s\in \ui$ such that $t=i_s$ is an $A^E_I$-logical consequence of
  $\phi_k$, for all $k$: this is achieved by conjoining (just once)
  equations like $i_s=t$ with $\phi_k$ - here the $i_s$ are new
  existentially quantified variables and $t$ is a representative
  $\Sigma_I$-term in which only the original existentially quantified
  variables occur. In this way, all elements in a substructure
  generated by $\ui$ are named explicitly and so are their
  $a_0$-images $a_0[\ui]$ (otherwise said, modulo $\phi_k(\ui)$, for
  every $\Sigma_I(\ui)$-term $t$, we have that $a_0[t]$ is equal to
  some of the $a_0[\ui]$).

  Now, in a locally finite theory, every quantifier free formula
  $\theta$ having at most $m$ free variables, is equivalent to a
  disjunction of diagram formulae $\delta_{\cM}(G)$, where $\cM$ is a
  substructure of a model of the theory and $G$ is a set of generators
  for $\cM$ of cardinality at most $m$.\footnote{Since the theory is
    locally finite, there are finitely many atoms whose free variables
    are included in a given set of cardinality $m$. Maximal
    conjunctions of literals built on these atoms are either
    inconsistent (modulo the theory) or satisfiable in an
    $m$-generated substructure of a model of the theory. Because of
    maximality, these (maximal) conjunctions must be diagrams.} If we
  apply this to both $T_I$ and $T_E$, we get that our $K(a_0)$ can be
  rewritten as
  \begin{equation*}
    \bigvee_{\cA, \cB} \exists \ui\, (\delta_{\cA}(\ui) \wedge
    \delta_{\cB}(a_0[\ui]))
  \end{equation*}
  where $\cA$ ranges over the $m$-generated models of $T_I$ and $\cB$
  over the $m$-generated sub-models of $T_E$ (recall that $T_I$ is
  closed under substructures). Every such pair $(\cA, \cB)$ is either
  $A^E_I$-inconsistent (in case some equality among the generators of
  $\cA$ is not satisfied by the corresponding generators of $\cB$) or
  it gives rise to a configuration $a$ such that $\exists \ui\,
  (\delta_{\cA}(\ui) \wedge \delta_{\cB}(a_0[\ui]))$ is precisely
  $K_a$.
\end{proof}
The formula $K_s$ from Proposition~\ref{prop:conf1}(i) will be called
\emph{the diagram formula} for the configuration $s$.

The set $\cB(\tau, K)$ of configurations which are backward reachable
from the configurations satisfying a given $\exists^I$-formula $K$ is
thus an upset, being the union of infinitely many upsets; however,
even when the latter are finitely generated, $\cB(\tau, K)$ needs not
be so.  Under the hypothesis of local finiteness of $T_E$, this is
precisely what characterizes the termination of the backward
reachability procedure.
\begin{thm}[\cite{ijcar08}]
  \label{thm:term}
  Assume that $T_E$ is locally finite; let $ K$ be an
  $\exists^I$-formula. If $K$ is safe, then $\mathsf{BReach}$ in
  Figure \ref{fig:reach-algo} terminates iff $\cB(\tau, K)$ is a
  finitely generated upset.\footnote{If $K$ is unsafe, we already know
    that $\mathsf{BReach}$ terminates because it detects unsafety.}
\end{thm}
\begin{proof}
  Suppose that $\cB(\tau, K)$ is a finitely generated upset. Notice
  that
  \begin{eqnarray*}
    \cB(\tau, K)=\bigcup_n\mywidehat{BR^n(\tau, K)},
  \end{eqnarray*}
  consequently (since we have $\mywidehat{BR^0(\tau, K)}\subseteq
  \mywidehat{BR^1(\tau, K)} \subseteq \mywidehat{BR^2(\tau,
    K)}\subseteq \cdots$) we have $\cB(\tau, K)=\mywidehat{BR^n(\tau,
    K)}=\mywidehat{BR^{n+1}(\tau, K)}$ for some $n$, which means by
  the second claim of Proposition \ref{prop:conf} that $A^E_I\models
  BR^n(\tau, K)\leftrightarrow BR^{n+1}(\tau, K)$, i.e. that the
  Algorithm halts. Vice versa, if the Algorithm halts, we have
  $A^E_I\models BR^n(\tau, K)\leftrightarrow BR^{n+1}(\tau, K)$, hence
  $\mywidehat{BR^n(\tau, K)}=\mywidehat{BR^{n+1}(\tau, K)}=\cB(\tau,
  K)$ and the upset $\cB(\tau, K)$ is finitely generated by
  Proposition~\ref{prop:conf1}.
\end{proof}
To derive a sufficient condition for termination from the Theorem
above, we use the notion of a wqo as in~\cite{lics}.  A pre-order $(P,
\leq)$ is a \emph{well-quasi-ordering} (wqo) iff for every sequence
\begin{equation}
  \label{eq:succession}
  p_0, p_1, \dots, p_i, \dots
\end{equation}
of elements from $P$, there are $i<j$ with $p_i\leq p_j$.
%
\begin{cor}
  \label{coro:termination}
  $\mathsf{BReach}$ always terminates whenever the pre-order on
  $A_I^E$-configurations is a wqo.
\end{cor}
\begin{proof}
  It is sufficient to show that in a wqo all upsets are finitely
  generated.  This is a well-known fact that can be proved for
  instance as follows.  Let $U$ be an upset.  If $U$ is empty, then it
  is finitely generated.  Otherwise pick $p_0\in U$, if $\uparrow\!
  p_0=U$, clearly $U$ is finitely generated; otherwise, let $p_1\in U
  \setminus\uparrow p_0$. At the $(i+1)$-th step, either $U=\uparrow\!
  p_0\cup\cdots \cup \uparrow\! p_i$ and $U$ is finitely generated, or
  we can pick $p_{i+1}\in U$ with $p_{i+1}\not \in \uparrow\!
  p_0\cup\cdots \cup \uparrow\! p_i$. Since the last alternative
  sooner or later becomes impossible (because in an infinite sequence
  like~\eqref{eq:succession}, we must have $p_j\in\bigcup_{i<j}
  \uparrow\! p_i$ for some $j$), we conclude that $U$ is finitely
  generated.
\end{proof}
Termination of backward reachability for some classes of systems
(already considered in the literature) can be obtained from
Corollary~\ref{coro:termination}; some of these are briefly considered
in the example below.  Although decidable, many of these cases have
very bad computational behavior as only a non-primitive recursive
lower bound is known to exist.  For the detailed formalization of the
classes of systems mentioned in the example below, the interested
reader is pointed to the extended version of~\cite{ijcar08}.
\begin{example}\em
  We consider three classes of systems for which decidability of the
  safety problem can be shown by using
  Corollary~\ref{coro:termination} and well-known results (such as
  Dickson's Lemma, Highman's Lemma, or Kruskal's theorem; see,
  e.g.,~\cite{gallier} for a survey) for proving that the ordering on
  configurations is a wqo.
  \begin{enumerate}[$\bullet$]
  \item Take $T_E$ to be an enumerated data-type theory and $T_I$ to
    be the pure theory of equality over the signature
    $\Sigma_I=\{=\}$: the pre-order on $A_I^E$-configurations is a wqo
    by Dickson's Lemma.  In fact, if $T_E$ is the theory of a finite
    structure with support $\{e_1, \dots, e_k\}$, a configuration is
    uniquely determined by a $k$-tuple of integers (counting the
    number of the $i$ for which $a[i]=e_j$ holds) and the
    configuration ordering is obtained by component-wise comparison.
    In this setting, one can formalize both
    cache-coherence~\cite{cav-delzanno} (see also
    Example~\ref{ex:one}) and broadcast protocols~\cite{bro1,bro2}.
  \item Take $T_E$ to be an enumerated data-type theory and $T_I$ to be
    the theory of total order: the pre-order on $A_I^E$-configurations
    is a wqo by Higman's Lemma.  In fact, if $T_E$ is the theory of a
    finite structure with support $\{e_1, \dots, e_k\}$, a
    configuration is uniquely determined by a word on $\{e_1, \dots,
    e_k\}$ and the configuration ordering is simply the sub-word
    relation.  In this setting, one can formalize Lossy Channel
    Systems~\cite{lossy-channels,lossy2}.
  \item Take $T_E$ to be the theory of rationals (with the standard
    ordering relation $<$) and $T_I$ to be the pure theory of equality
    over the signature $\Sigma_I=\{=\}$: the pre-order on
    $A_I^E$-configurations is a wqo by Kruskal's theorem.  In fact, we
    can represent a configuration $(s, \mathcal M)$ as a list $n_1,
    \dots, n_k$ of natural numbers (of length $k$): such a list
    encodes the information that $s_E$ is a $k$-element chain and that
    $n_1$ elements from $s_I$ are mapped by $s$ into the first element
    of the chain, $n_2$ elements from $s_I$ are mapped by $s$ into the
    second element of the chain, etc.  If $w$ is the list for $s$ and
    $v$ is the list for $s'$, we have $s'\leq s$ iff $w$ is less than
    or equal component-wise to a sub-word of $v$.  Termination by
    Kruskal's theorem is obtained by representing numbers as numerals
    and by using a binary function symbol $f$ to encode the precedence
    (thus, for instance, the list 1,2,2 is represented as $f(succ(0),
    f(succ(succ(0)), succ(succ(0))))$); it is easily seen that, on these terms,
    the homeomorphic embedding~\cite{BaaNi} behaves like our configuration
    ordering.
  \end{enumerate}
\end{example}
\noindent A final remark is in order.  In the model checking literature of
infinite state systems, an important property is that of
`monotonicity'~\cite{lics}
(in an appropriate setting, this property is shown to be equivalent to
the fact that the pre-image of an upset is still an upset).  Such a
property is not used in the proofs above as we work symbolically with
definable upsets.  However, it is possible to formulate it in our
framework as follows:
\begin{itemize}
\item[--] if $(s, \cM), (s', \cM')$, and $(t, \cM)$ are configurations
  such that $s\leq s'$ and $\cM\models \tau(s,t)$, then there exists
  $(t',\cM')$ such that $t\leq t'$ and $\cM'\models \tau(s', t')$.
\end{itemize}
The proof that such a property holds for transitions in the
format~(\ref{eq:transition2}) is easy and left as an exercise to the
reader (it basically depends on the fact that truth of existential
formulae is preserved by superstructures).

\section{Invariants Search}
\label{sec:inv}

It is well-known that invariants are useful for pruning the search
space of backward reachability procedures and may help either to
obtain or to speed up termination.

\subsection{Safety Invariants}
\label{subsec:dec+inv}
First of all, we recall the basic notion of safety invariant.
\begin{defi}
  \label{def:inv}
  The $\forall^I$-formula $J(a)$ is a \emph{safety invariant} for the
  safety problem consisting of the array-based system
  $\cSi=(a,I,\tau)$ and unsafe $\exists^I$-formula $U(a)$ iff the following
  conditions hold:
  \begin{enumerate}[{\rm (i)}]
    \item $A^E_I\models \forall a (I(a)\to J(a))$,
    \item $A^E_I\models \forall a\forall a' (J(a)\wedge \tau(a, a')\to
      J(a'))$, and
    \item $\exists a.(U(a)\wedge J(a))$ is $A_I^E$-unsatisfiable.
  \end{enumerate}
  If we are not given the $\exists^I$-formula $U(a)$ and only conditions
  (i)--(ii) hold, then $J(a)$ is said to be an \emph{invariant for $\cSi$}.
\end{defi}
Checking whether conditions (i), (ii), and (iii) above hold can be
reduced, by trivial logical manipulations, to the
$A_I^E$-satisfiability of $\exists^{A,I}\forall^I$-formulae, which is
decidable by Theorem~\ref{th:decidability}.  So, establishing whether
a given $\forall^I$-formula $J(a)$ is a safety invariant can be
completely automated.
\begin{property}
  \label{prop:inv-method}
  Let $U$ be an $\exists^I$-formula.  If there exists a safety
  invariant for $U$, then the array-based system $\cSi=(a,I,\tau)$ is
  safe with respect to $U$.
\end{property}
\begin{proof}
  For reductio, suppose that there is a safety invariant for $U$ and
  the array-based system $\cSi=(a,I,\tau)$ is not safe w.r.t.\ $U$.
  This implies that the formula
  \begin{equation}
    I(a_0)\wedge \tau(a_0, a_1)\wedge \cdots \wedge \tau(a_{n-1},
    a_n)\wedge U(a_n)
  \end{equation}
  is $A^E_I$-satisfiable. By using (i) and (ii) in
  Definition~\ref{def:inv}, we derive that $J(a_n)\wedge U(a_n)$ is
  $A^E_I$-satisfiable, in contrast to (iii) in
  Definition~\ref{def:inv}.
\end{proof}
Thus, if we are given a suitable safety invariant,
Property~\ref{prop:inv-method} can be used as the basis of the safety
invariant method, which turns out to be more powerful than the basic
backward reachability procedure in Figure~\ref{fig:reach-algo} (a).
\begin{property}
  \label{ex:basic}
  Let the procedure $\mathsf{BReach}$ in Figure~\ref{fig:reach-algo}(a)
  terminate on the safety problem consisting of the array-based system
  $\cSi=(a,I,\tau)$ and unsafe formula $U(a)$.  If $\mathsf{BReach}$
  returns $(\mathsf{safe},B)$, then $\neg B$ is a safety invariant for
  $U$.
\end{property}
\begin{proof}
  Suppose that $\mathsf{BReach}$ exits the main loop at the $k$-th
  iteration by returning $B$; then $B$ is $\bigvee_{i=0}^k Pre^i(\tau,
  U)$,\footnote{Notice that the disjunction of $\exists^I$-formulae is
    (up to logical equivalence) an $\exists^I$-formula, so $B$ is
    itself an $\exists^I$-formula.}  the formula $Pre^{k+1}(\tau,
  U)\wedge \neg B$ is $A^E_I$-unsatisfiable and the formulae $I\wedge
  Pre^i(\tau, U)$ (for $i=0, \dots, k$) are also
  $A^E_I$-unsatisfiable.  The latter means that $A^E_I\models \forall
  a (I(a)\to \neg B(a))$; for $i=0$ (since $Pre^0(\tau, U)$ is $U$),
  we also get that $\exists a.(U(a)\wedge \neg B(a))$ is
  $A_I^E$-unsatisfiable.  To claim that $\neg B(a)$ is an invariant,
  we only need to check that $A^E_I\models \forall a\forall a' (\neg
  B(a)\wedge \tau(a, a')\to \neg B(a'))$, i.e.\ that $A^E_I\models
  \forall a (Pre(\tau,B(a))\to B(a))$, which trivially holds since
  $Pre(\tau, B)$ is
$\bigvee_{i=1}^{k+1}Pre^i(\tau,U)$
and hence implies
 $Pre^{k+1}(\tau,U)\vee B$ and consequently also $B$ (recall that
  $Pre^{k+1}(\tau,
  U)\wedge \neg B$ is $A^E_I$-unsatisfiable).
\end{proof}
The converse of Proposition~\ref{ex:basic} does not hold: there might
be a safety invariant even when $\mathsf{BReach}$ diverges, as
illustrated by the following example.\footnote{More significant
  examples having a similar behavior can be found in the \textsc{mcmt}
  distribution.}
\begin{example}\em
  \label{ex:only}
  We consider an algorithm to insert an element $b[0]$ into a sorted
  array $b[1], \dots, b[n]$ (this can be seen as a sub-procedure of
  the insertion sort algorithm).  To formalize this, let $\Sigma_I$
  contain one binary predicate symbol $S$ and one constant symbol $0$
  and $T_I$ be the theory whose class of models consists of the
  substructures of the structure having the naturals as domain, with
  $0$ interpreted in the obvious way, and $S$ interpreted as the graph
  of the successor function.  For the sake of simplicity, we shall use
  a two-sorted theory for data and two array variables: let $T_E$ be
  the two-sorted theory whose class of models consists of the single
  two-sorted structure given by the Booleans (with the constants
  $\top, \bot$ interpreted as true and false, respectively) and the
  rationals (with the usual ordering relation $<$); the array variable
  $a$ is a collection of Boolean flags and the array variable $b$ is
  the sorted numerical array where $b[0]$ should be inserted.  The
  initial $\forall^I$-formula is represented as follows:
  \begin{eqnarray*}
    \forall i\, (a[i]=\bot \leftrightarrow i\neq 0) \wedge
    \forall i_1, i_2 \,(S(i_1,i_2) \to i_1=0 \vee b[i_1]\leq b[i_2]) ,
  \end{eqnarray*}
  saying that the elements in the array $b$, whose corresponding
  Boolean flag is set to false (namely, all except the one at position
  $0$), are arranged in increasing order.  The procedure can be
  formalized by using just one transition formula in the
  format~\ref{eq:transition2} whose guard and global component are as
  follows:
  \begin{eqnarray*}
    \phi_L(i_1,i_2,a[i_1],a[i_2]) & := &
    S(i_1, i_2) \wedge
    a[i_1]=\top \wedge
    a[i_2]=\bot \wedge
    b[i_1] > b[i_2]  \\
    F_G(i_1, i_2, a[i_1], a[i_2],  b[i_1], b[i_2], j) & := & 
             \mathtt{if}~(j= i_1)~\mathtt{then}~\langle \top, b[i_2]\rangle \\
         &&    \mathtt{else~if}~(j=i_2)~\mathtt{then}~\langle \top, b[i_1]\rangle \\
         &&    \mathtt{else}~\langle a[j], b[j]\rangle ,
  \end{eqnarray*}
  which swaps two elements in the array $b$ if their order is
  decreasing and sets the Boolean fields appropriately (notice that
  $F_G$ updates a pair of array variables whose first component is the
  new value of $a$ and second component is the new value of $b$).  The
  obvious correctness property is that there are no two values in
  decreasing order in the array $b$ if the corresponding Boolean flags
  do not allow the transition to fire:
  \begin{equation}
    \label{ex:s}
    \exists i_1,i_2\,
    (S(i_1, i_2) \wedge
    \neg (a[i_1]= \top \wedge a[i_2]= \bot) \wedge
    b[i_1]>b[i_2]) .
  \end{equation}
  Unfortunately, $\mathsf{BReach}$ in Figure~\ref{fig:reach-algo} (a)
  diverges when applied to~\eqref{ex:s}.  Fortunately, a safety
  invariant for (\ref{ex:s}) exists.  This can be obtained as follows:
  run \textsc{mcmt} on the safety problem given by the disjunction of
  (\ref{ex:s}) and the formula
  \begin{equation}
    \label{eq:inssort}
    \exists i, j.(S(i,j) \wedge a[i]= \bot \wedge a[j]= \top)
  \end{equation}
  saying that two adjacent indexes have their Boolean flags set to
  $\bot$ and $\top$, respectively. The problem is immediately solved
  by the tool: by Property~\ref{ex:basic}, the formula describing the
  set of backward reachable states is a safety invariant for the safety
  problem given by the disjunction of~\eqref{ex:s}
  and~\eqref{eq:inssort}, hence \emph{a fortiori} also for the safety
  formula~\eqref{ex:s} alone.  In this case,
  formula~\eqref{eq:inssort} has been found manually; however,
  \textsc{mcmt} \emph{can find it without user intervention} as soon
  as its invariant synthesis capabilities are activated by suitable
  command line options.  The combination of automatic invariant search
  and backward reachability will be the main subject of
  Section~\ref{sec:back+inv} below. \qed
\end{example}
It is interesting to rephrase the conditions of
Definition~\ref{def:inv} in terms of configurations as this paves the
way to characterize the completeness of our invariant synthesis method
as will be shown below.
\begin{lemma}
  \label{lem:inv}
  Let $J$ be a $\forall^I$-formula; the conditions (i), (ii), and
  (iii) of Definition~\ref{def:inv} are equivalent to the following
  three conditions on (sets of) configurations:
  \begin{eqnarray}
    \label{eq:si1}
    \mywidehat{I} \cap \mywidehat{H} =\emptyset\\
    \label{eq:si2}
    \mywidehat{Pre(\tau, H)}\subseteq \mywidehat{H}\\
    \label{eq:si3}
    \mywidehat{U}\subseteq \mywidehat{H} ,
  \end{eqnarray}
  where $H$ is the $\exists^I$-formula which is logically equivalent
  to the negation of $J$.
\end{lemma}
\begin{proof}
  For~(\ref{eq:si1}), we have:
  \begin{eqnarray*}
    \mbox{(i) of Def.~\ref{def:inv}} ~\Leftrightarrow~
    A_I^E \models \forall a.(I(a) \to J(a)) & \Leftrightarrow~\\
    \neg \forall a.(I(a) \to J(a)) \mbox{ is
      $A_I^E$-unsat.} & \Leftrightarrow~\\
      \exists a.(I(a) \wedge \neg J(a)) \mbox{ is
      $A_I^E$-unsat.} & \Leftrightarrow~\\
    \exists a.(I(a) \wedge H(a)) \mbox{ is
      $A_I^E$-unsat.} ~\Leftrightarrow~
    \mywidehat{I} \cap \mywidehat{H} = \emptyset . &
  \end{eqnarray*}
  For~(\ref{eq:si2}), we have:
  \begin{eqnarray*}
    \mbox{(ii) of Def.~\ref{def:inv}} ~\Leftrightarrow~
    A_I^E \models \forall a,a'.(J(a) \wedge \tau(a,a')\to J(a'))
    & \Leftrightarrow \\
    \exists a,a'.\neg (J(a) \wedge \tau(a,a')\to
    J(a')) \mbox{ is $A_I^E$-unsat.} & \Leftrightarrow\\
    \exists a,a'.(J(a) \wedge \tau(a,a')\wedge \neg
    J(a')) \mbox{ is $A_I^E$-unsat.} & \Leftrightarrow \\
    \exists a.(J(a) \wedge \exists
    a'.(\tau(a,a')\wedge \neg J(a')))  \mbox{ is $A_I^E$-unsat.}
    & \Leftrightarrow \\
    \exists a.(J(a) \wedge \exists
    a'.(\tau(a,a')\wedge H(a')))  \mbox{ is $A_I^E$-unsat.}
    & \Leftrightarrow \\
    \exists a.(J(a) \wedge Pre(\tau,H)(a)) \mbox{ is
      $A_I^E$-unsat.} & \Leftrightarrow \\
    A_I^E \models \forall a.(\neg J(a) \vee \neg
    Pre(\tau,H)(a))  & \Leftrightarrow\\
    A_I^E \models \forall a.(H(a) \vee \neg
    Pre(\tau,H)(a))  & \Leftrightarrow \\
    A_I^E \models \forall a.(Pre(\tau,H)(a)\to H(a))
    ~\Leftrightarrow~  \mywidehat{Pre(\tau,H)} \subseteq
    \mywidehat{H}. &
  \end{eqnarray*}
  For~(\ref{eq:si3}), we have:
  \begin{eqnarray*}
    \mbox{(iii) of Def.~\ref{def:inv}} ~\Leftrightarrow~
    \exists a.(U(a) \wedge J(a)) \mbox{ is $A_I^E$-unsat.}
    & \Leftrightarrow \\
    \exists a.\neg (\neg U(a) \vee \neg J(a)) \mbox{ is
      $A_I^E$-unsat.} & \Leftrightarrow\\
    A_I^E\models \forall a. ( U(a) \to \neg J(a))
    & \Leftrightarrow \\
    A_I^E\models \forall a. ( U(a) \to H(a)) ~\Leftrightarrow~
    \mywidehat{U} \subseteq \mywidehat{H} . &
  \end{eqnarray*}
\end{proof}

\subsection{Invariant Synthesis}
\label{subsec:dual}
The main difficulty to exploit Property~\ref{prop:inv-method} is to
find suitable $\forall^I$-formulae satisfying conditions (i)---(iii)
of Definition~\ref{def:inv}.  Unfortunately, the set of
$\forall^I$-formulae which are candidates to become safety invariants
is infinite.  Such a search space can be dramatically restricted when
$T_E$ is locally finite, although it is still infinite because there
is no bound on the length of the universally quantified prefix.  From
a technical point of view, we need to develop some preliminary results.

First, we give a closer look to the \emph{equivalence} relation among
configurations: we recall that $s$ is equivalent to $t$
(written $s\approx t$) iff $s\leq t$ and $t\leq s$.
\begin{proposition}
  \label{prop:equiv}
  We have that $s\approx t$ holds iff there are a
  $\Sigma_I$-isomorphism $\mu$ and a $\Sigma_E$-isomorphism $\nu$ such
  that such that the set-theoretical compositions of $\mu$ with $s$
  and of $s'$ with $\nu$ are equal.\footnote{Notice that, since the
    image of $s$ is a set of generators for $s_E$, it is not difficult
    to see that $\nu$ is uniquely determined from $\mu$ (i.e., given
    $\mu$, there might be no $\nu$ such that the square commutes, but
    in case one such exists, it is unique).  Observe also that, if $s$
    comes from the finite index model $\cM$ and $t$ comes from the
    finite index model $\cN$, the fact that $s\approx t$ holds does
    not mean that $\cM$ and $\cN$ are isomorphic: their
    $\Sigma_I$-reducts are $\Sigma_I$-isomorphic, but their
    $\Sigma_E$-reducts need not be $\Sigma_E$-isomorphic (only the
    $\Sigma_E$-substructures $s_E$ and $t_E$ are
    $\Sigma_E$-isomorphic).  }  The situation is depicted in the
  following diagram:
  \begin{center}
    \resetparms
    \setsqparms[+2`+1`+1`+2;500`500]
    \square[s'_I` s_I`s'_E`s_E;\mu`s'`s`\nu]
  \end{center}
\end{proposition}
\begin{proof}
  The implication `$\Leftarrow$' is straightforward and thus we detail
  only `$\Rightarrow$' in the following.  The supports of $s_I$ and of
  $t_I$ are finite, hence the existence of embeddings
  $s_I{\buildrel{\mu_1}\over{\longrightarrow}}
  t_I\buildrel{\mu_2}\over\longrightarrow s_I$ means (for cardinality
  reasons) that $\mu_1, \mu_2$ are bijections, hence
  isomorphisms. Since the images of $s$ and $t$ are finite sets of
  generators for $s_E$ and $t_E$, respectively, we have embeddings
  $s_E\buildrel{\nu_1}\over\longrightarrow
  t_E\buildrel{\nu_2}\over\longrightarrow s_E$ mapping generators into
  generators: again,
for cardinality reasons,
$\nu_1, \nu_2$ restrict
  to bijections among generators, which means that they are
  isomorphisms.
\end{proof}
\begin{defi}
  A \emph{basis} for a finitely generated upset $S$ (resp., for an
  $\exists^I$-formula $K$) is a minimal finite  set $\{s_1, \dots, s_n\}$
  such that $S$ (resp., $\mywidehat{K}$) is equal to $\uparrow\!
  s_1\cup \cdots\cup \uparrow\! s_n$.
\end{defi}
It is easy to see that two bases for the same upset are essentially
the same, in the sense that \emph{they are formed by pairwise
  equivalent configurations}. Suppose in fact that $\{s_1, \dots,
s_n\}$ and $\{s'_1, \dots, s'_m\}$ are two bases for the same
upset. Then for every $s_i$ there exists $s'_j$ such that $s'_j\leq
s_i$; however, there is also $s_k$ with $s_k\leq s'_j$ (because
$\{s_1, \dots, s_n\}$ is a basis) and by minimality it follows that
$s_i=s_k$, which means that $s_i$ and $s'_j$ are equivalent. Thus each
member of a basis is equivalent to a member of the other (and to a
unique one by minimality again) and vice versa; in particular, we also
have that $m=n$.
\begin{lemma}\label{lem:basis}
Suppose $T_E$ is locally finite.
  A configuration $s$ belongs to a basis for an $\exists^I$-formula
  $K$ iff $s\in \mywidehat{K}$ and for every $s'$ ($s'\leq s$ and
  $s'\in \mywidehat{K}$) imply that $s\approx s'$.
\end{lemma}
\begin{proof}
  Let $B$ be a basis for $K$ and let also $s\in B$, $s'\leq s$ and
  $s'\in \mywidehat{K}$; then $s'$ is bigger than some configuration
  from $B$, which must be $s$, because elements from $B$ are
  incomparable: $s\approx s'$ follows immediately. Conversely, suppose
  that $s\in \mywidehat{K}$ and for every $s'$, $s'\leq s$ and $s'\in
  \mywidehat{K}$ imply that $s\approx s'$.  Since $T_E$ is locally
  finite, $K$ has a basis $B$ (this can be immediately deduced from
  Proposition~\ref{prop:conf1}(ii)). We have $b\leq s$ (and also
  $s\approx b$) for some $b$ from $B$: it is now clear that we can get
  another basis for $K$ by replacing in $B$ the configuration $b$ with
  $s$.
\end{proof}
Our goal is to integrate the safety invariant method into the basic
Backward Reachability algorithm of Figure~\ref{fig:reach-algo}(a).  To
this end, we introduce the
notion of `sub-reachability.'
\begin{defi}[Subreachable configurations]
  \label{def:subreach-configs}
  Suppose $T_E$ is locally finite and let $s$ be a configuration.  A
  \emph{predecessor} of $s$ is any $s'$ that belongs to a basis for
  ${Pre(\tau, K_s)}$ (see Proposition~\ref{prop:conf1} for the
  definition of $K_s$).  Let $s, s'$ be configurations: $s$ is
  \emph{sub-reachable} from $s'$ iff there exist configurations $s_0,
  \dots, s_n$ such that (i) $s_0=s$, (ii) $s_n=s'$, and (iii) either
  $s_{i-1}\leq s_i$ or $s_{i-1}$ is a predecessor of $s_i$, for each
  $i=1, \dots, n$.  If $K$ is an $\exists^I$-formula, \emph{$s$ is
    sub-reachable from $K$} iff $s$ is sub-reachable from some $s'$
  taken from a basis of ${K}$.
\end{defi}
The following is the main technical result of this section.
\begin{thm}
  \label{th:inv}
  Let $T_E$ be locally finite.  If there exists a safety invariant for
  $U,$ then there are finitely many $A_I^E$-configurations $s_1,
  \dots, s_k$ which are sub-reachable from $U$ and such that $\neg
  (K_{s_1}\vee\cdots\vee K_{s_k})$ is also a safety invariant for $U$.
\end{thm}
\begin{proof}
  Our goal is to replace an $\exists^I$-formula $H$ satisfying the
  three conditions of Lemma~\ref{lem:inv} with an $\exists^I$-formula
  $L$ whose negation is still a safety invariant for $U$ and whose
  basis is formed by configurations which are all sub-reachable from
  $U$.  To this end, we consider a function $\gamma(S)$ where $S$ is
  an $\exists^I$-formula such that $\mywidehat{S}\subseteq
  \mywidehat{H}$: the function $\gamma(S)$ returns an
  $\exists^I$-formula $K_{a_1}\vee\cdots\vee K_{a_n}$, where $\{a_1,
  \dots, a_n\}\subseteq \mywidehat{H}$ is a minimal set of
  configurations taken from a basis of $H$ such that
  $\mywidehat{S}\subseteq\,\uparrow\! a_1\cup \cdots \cup \uparrow\!
  a_n$.  (Notice that this implies that $\{a_1, \dots, a_n\}$ is a
  basis of $\gamma(S)$ and $\mywidehat{S}\subseteq
  \mywidehat{\gamma(S)}$.)\footnote{There might be many functions
    $\gamma$ satisfying the above specification, we just take one of
    them.  This can be done (by choice axiom) because, given $S$ such
    that $\mywidehat{S}\subseteq \mywidehat{H}$, there always exists a
    minimal set of configurations $\{a_1, \dots, a_n\}$ taken from a
    basis of $H$ such that $\mywidehat{S}\subseteq\,\uparrow\! a_1\cup
    \cdots \cup \uparrow\!  a_n$ (just take any basis for $H$ and
    throw out configurations from it until minimality is acquired).  }

  Now, define the following sequence of $\exists^I$-formulae $L_i$:
  (i) $L_0:=\gamma(U)$ and (ii) $L_{i+1}:= L_i\vee \gamma(Pre(\tau,
  L_i))$.  (The definition is well given because
  $\mywidehat{L_i}\subseteq \mywidehat{H}$ is a consequence of
  \eqref{eq:si3} and \eqref{eq:si2}.)  What remains to be shown is
  that the sequence becomes stable and its fix-point is the desired
  $L$, i.e.\ a safety invariant for $U$ whose basis is formed by
  configurations which are sub-reachable from $U$.

  We first show, by induction on $k$, that every configuration $b$
  that belongs to a basis of ${L_k}$ is sub-reachable from $U$:
  \begin{enumerate}[$\bullet$]
  \item if $k=0$, we have that $\{a_1, \dots, a_n\}$ is a minimal set
    of configurations taken from a basis of $H$ such that
    $\mywidehat{U}\subseteq\,\uparrow\! a_1\cup \cdots \cup \uparrow\!
    a_n$ and $b=a_j$ for some $j=1, \dots, n$. By minimality, there is
    $s$ from a basis of $U$ such that $s\not \in \uparrow \! a_1\cup
    \cdots\cup \uparrow\! a_{j-1}\cup \uparrow\! a_{j+1} \cup
    \uparrow\!  a_n$, which means that $s\in \uparrow\!a_j$, that is
    $a_j\leq s$ and $a_j=b$ is sub-reachable from $U$.
  \item Suppose now $k=i+1>0$.  A basis for $L_i\vee \gamma(Pre(\tau,
    L_i))$ is obtained by joining two bases--- one for $L_i$ and one
    for $\gamma(Pre(\tau, L_i))$---and then by discarding non-minimal
    elements.  As a consequence, if $b$ is in a basis for $L_k$, then
    $b$ is either in a basis for $L_i$ or in a basis for
    $\gamma(Pre(\tau, L_i))$ (or in both). In the former case, we just
    apply induction.  If $b$ is in a basis for $\gamma(Pre(\tau,
    L_i))$, the same argument used in the case $k=0$ shows that $b\leq
    s$ for an $s$ that belongs to a basis for $Pre(\tau, L_i)$.  Now,
    if $c_{i1},\dots, c_{ik_i}$ is a basis of $L_i$, the formula
    $Pre(\tau, L_i)$ is $A^E_I$-equivalent to the disjunction of the
    $Pre(\tau, K_{c_{ij}})$ and consequently $s$ must be in a basis of
    one of the latter (that is, $s$ is a predecessor of some
    $c_{ij}$); since the $c_{ij}$ are sub-reachable by induction
    hypothesis and $b\leq s$, the definition of sub-reachability
    guarantees that $b$ is sub-reachable from $U$.
  \end{enumerate}
  The increasing chain
  \begin{eqnarray*}
    \mywidehat{L_0}\subseteq \mywidehat{L_1}\subseteq \cdots
  \end{eqnarray*}
  becomes stationary, because at each step only configurations from a
  basis of ${H}$ can be added and bases are (unique and) finite by
  definition.  Thus, we have $\mywidehat{L_i}=\mywidehat{L_{i+1}}$ for
  some $i$: let $L$ be $L_i$ for such $i$.

  The fact that $L$ is a safety invariant is straightforward:
  condition $\mywidehat{I} \cap \mywidehat{L} =\emptyset$ follows from
  \eqref{eq:si1} and the fact that $\mywidehat{L}\subseteq
  \mywidehat{H}$, whereas conditions $\mywidehat{U}\subseteq
  \mywidehat{L}$ and $\mywidehat{Pre(\tau, L)}\subseteq \mywidehat{L}$
  follow directly from the above definitions of $L_0$ and $L_{i+1}$
  (we have $\mywidehat{U}\subseteq \mywidehat{\gamma(U)}=
  \mywidehat{L_0}\subseteq \mywidehat{L}$ and for all $i\geq 0$,
  $\mywidehat{Pre(\tau, L_i)}\subseteq \mywidehat{\gamma(Pre(\tau,
    L_i))}\subseteq \mywidehat{L_{i+1}}\subseteq \mywidehat{L}$).
\end{proof}
The intuition underlying the theorem is as follows. Let us call
`finitely representable' an upset which is of the kind $\mywidehat{K}$
for some $\exists^I$-formula $K$ and let $B$ be the set of backward
reachable states.  Usually $B$ is infinite and it is finitely
representable only in special cases (e.g., when the configuration
ordering is a wqo).  Nevertheless, it may sometimes exist a set
$B'\supseteq B$ which is finitely representable and whose complement
is an invariant of the system.  Theorem~\ref{th:inv} ensures us to
find such a $B'$, if any exists.  This is the case of
Example~\ref{ex:only} where not all configurations satisfying
(\ref{eq:inssort}) are in $B$ and $B$ must be enlarged to
encompass such configurations too (only in this way it becomes
finitely representable, witness the fact that backward reachability
diverges).

In practice, Theorem~\ref{th:inv} suggests the following procedure to
find the super-set $B'$.  At each iteration of $\mathsf{BReach}$, the
algorithm represents symbolically in the variable $B$ the
configurations which are backward reachable in $n$ steps; before
computing the next pre-image of $B$, non deterministically replace
some of the configurations in a basis of $B$ with some
sub-configurations and update $B$ by a symbolic representation of the
upset obtained in this way.  As a consequence, if an invariant exists,
we are guaranteed to find it; otherwise, the process may diverge.
Notice that (in the local finiteness hypothesis for $T_E$) the search
space of the configurations which are sub-reachable in $n$ steps is
finite, although this search space is infinite if no bound on $n$ is
fixed.  To illustrate,
%
(\ref{eq:inssort}) in
Example~\ref{ex:only} contains some sub-reachable only configurations.
This shows that sub-reachability is crucial for Theorem~\ref{th:inv}
to hold.

The algorithm sketched above can be refined further so as to obtain a
completely symbolic method working with formulae without resorting to
configurations.
The key idea to achieve this is to rephrase in a symbolic setting the
relevant notions concerning sub-reachability.  However, this goal is
best achieved incrementally as there are some subtle aspects to take
care of.  The starting point is the following observation.  It is not
possible to characterize the fact that a configuration $(s, \cM)$ is
part of a basis for an $\exists^I$-formula $\exists\ui\, \phi(\ui,
a[\ui])$ by using another $\exists^I$-formula (a universal quantifier
is needed to express the suitable minimality requirement).  Instead,
we shall characterize by an $\exists^I$-formula the fact that a tuple
satisfying $\phi(\ui, a[\ui])$ generates a submodel which is a
configuration belonging to a basis (see Lemma~\ref{lem:min} below).
Notice that the simple fact that the tuple satisfies $\phi$ is not
sufficient alone: for instance, only pairs formed by \emph{identical}
elements satisfying $a[i_1]=a[i_2]$ generate a configuration in a
basis (tuples formed by pairs of different elements are not minimal).
To generalize this, we introduce the following abbreviation:
\begin{eqnarray}
  \label{eq:mindef}
  Min(\phi, a,\ui) & := &
  \phi(\ui, a[\ui])\wedge
  \bigwedge_{\sigma} \left(\phi(\ui\sigma, a[\ui\sigma])\to
                     \bigwedge_{i\in \ui}\bigvee_t (t\sigma=i)\right)
\end{eqnarray}
where $\phi(\ui,a[\ui])$ is a quantifier-free formula, $t$ ranges over
representative $\Sigma_I(\ui)$-terms, and $\sigma$ ranges over the
substitutions with domain $\ui$ and co-domain included in the set of
representative $\Sigma_I(\ui)$-terms.  The following lemma gives a
semantic characterization of $Min(\phi, a,\ui)$.
\begin{lemma}
  \label{lem:min}
  Consider an $\exists^I$-formula $K\equiv \exists\ui\, \phi(\ui,
  a[\ui])$, an $A^E_I$-model $\cM$, and a variable assignment
  $\mathtt{a}$ in $\cM$ such that $(\cM, \mathtt{a})\models \phi(\ui,
  a[\ui])$.  We have that $(\cM, \mathtt{a})\models Min(\phi, a, \ui)$
  iff the configuration $s$ obtained by restricting $\mathtt{a}(a)$ to
  the $\Sigma_I$-substructure generated by the $\mathtt{a}(\ui)$'s
  belongs to a basis of $K$.\footnote{To make the statement of the
    lemma precise, one should define not just $s$ but also the finite
    index model where $s$ is taken from.  In detail, we take the
    $A^E_I$-model $\cN$ whose $\Sigma_I$-reduct is the restriction of
    $\cM_I$ to the $\Sigma_I$-substructure generated by the
    $\mathtt{a}(\ui)$'s and whose $\Sigma_E$-reduct is equal to
    $\cM_E$. In this model, we can define the array $s$ to be the
    restriction of $\mathtt{a}(a)$ to ${\tt INDEX}^{\cN}\subseteq {\tt
      INDEX}^{\cM}$.  The pair $(s,\cN)$ is now a configuration in the
    sense defined in Section~\ref{subsec:dual}.  }
\end{lemma}
\begin{proof}
  Suppose that $(\cM, \mathtt{a})\models Min(\phi, a, \ui)$ (for
  simplicity, we shall directly call $\ui,a$ the elements assigned by
  $\mathtt{a}$ to $\ui,a$, respectively).  By
  Proposition~\ref{prop:equiv} and Lemma~\ref{lem:basis}, it is
  sufficient to show the following.  Consider $s'\leq s$ such that
  $s'\in \mywidehat{K}$: we show that the embeddings $\mu, \nu$
  witnessing the relation $s'\leq s$ and making the
  diagram
  \begin{center}
    \resetparms
    \setsqparms[+2`+1`+1`+2;500`500]
    \square[s'_I` s_I`s'_E`s_E;\mu`s'`s`\nu]
  \end{center}
  to commute are isomorphisms (in fact, it is sufficient to show only
  that $\mu$ is bijective, because the images of $s'$ and $s$ are
  $\Sigma_E$-generators and the square commutes).  Without loss of
  generality, we can assume that $\mu$ is an inclusion; the domain of
  $s'$ is then formed by elements of the form $t^{(\cM, \mathtt{a})}$
  for suitable (representative) $\Sigma_I(\ui)$-terms $t$ and the fact
  that $s'\in \mywidehat{K}$ means then that $(\cM, \mathtt{a})\models
  \phi(\ui\sigma, a[\ui\sigma])$ holds for a substitution $\sigma$
  whose domain is $\ui$ and whose range is contained into the set of
  those representative $\Sigma_I(\ui)$-terms $u$ such that $u^{(\cM,
    \mathtt{a})}$ is in the support of $s'_I$. Since $(\cM,
  \mathtt{a})\models Min(\phi, a, \ui)$ holds, for every $i\in \ui$
  there is a representative $\Sigma_I(\ui)$-term $t$ such that $(\cM,
  \mathtt{a})\models t\sigma=i$ holds. The latter means that $i$ is in
  the support of $s'_I$, hence the inclusion $\mu$ is onto.

  Conversely, if $s$ belongs to a basis of of $K$, then there is no
  $s'\leq s$ is in $\mywidehat{K}$, unless $s'$ is equivalent to $s$,
  by Lemma~\ref{lem:basis}. Suppose that $(\cM, \mathtt{a})\models
  \phi(\ui\sigma,a[\ui\sigma])$ holds for a substitution $\sigma$
  whose domain is $\ui$ and whose range is included into the set of
  representative $\Sigma_I(\ui)$-terms. For reductio, suppose that
  $(\cM, \mathtt{a})\models t\sigma=i$ does not hold for some $i\in
  \ui$ and all representative $\Sigma_I(\ui)$-terms $t$; we can
  restrict the array $a$ to the $\Sigma_I$-substructure given by the
  elements of the kind $t\sigma^{(\cM, \mathtt{a})}$, thus getting a
  configuration $s'\leq s$ such that $s'\in\mywidehat{K}$. Since the
  finite support of $s'_I$ has smaller cardinality than the support of
  $s_I$ (because $\mathtt{a}(i)$ does not belong to it), we cannot
  have $s'\approx s$, a contradiction!
\end{proof}
\begin{rem}
  \label{rem}
  We identify conditions under which it is trivial to compute
  $Min(\phi, a, \ui)$.  Besides being an interesting observation
  \emph{per se}, it will be used later in this section to illustrate
  simple and useful examples of the key notion of cover (see
  Example~\ref{ex:cover} below).  If (as it often happens in
  applications) the signature $\Sigma_I$ is relational and the formula
  $\phi(\ui, a[\ui])$ is differentiated, $Min(\phi, a, \ui)$ is
  $A^E_I$-equivalent to $\phi(\ui, a[\ui])$: this is because only
  variable permutations can be consistently taken into consideration
  as the $\sigma$'s in formula~\eqref{eq:mindef}, so that the
  $t\sigma$'s are precisely the $\ui$'s.
\end{rem}
\begin{cor}
  \label{coro:min}
  Consider an $\exists^I$-formula $K:= \exists\ui\, \phi(\ui, a[\ui])$
  and a configuration $(s,\cM)$; if $s$ belongs to a basis for $K$,
  then $(\cM, \mathtt{a})\models \phi(\ui, a[\ui])\to Min(\phi, a,
  \ui)$ holds for all $\mathtt{a}$ such that $\mathtt{a}(a)=s$.
\end{cor}
\begin{proof}
  If $(\cM, \mathtt{a})\models \phi(\ui, a[\ui])$, then the
  configuration $s'$ obtained by restricting $\mathtt{a}(a)=s$ to the
  $\Sigma_I$-substructure generated by the $\mathtt{a}(\ui)$ is
  equivalent to $s$ by Lemma~\ref{lem:basis} and hence belongs to a
  basis of $K$.  Thus Lemma~\ref{lem:min} applies and gives $(\cM,
  \mathtt{a})\models Min(\phi, a, \ui)$.
\end{proof}
The next step towards the goal of obtaining a completely symbolic
method for mechanizing the result stated in Theorem~\ref{th:inv}
consists of finding a purely symbolic substitute
of the function $\gamma$ used in the proof of Theorem~\ref{th:inv}.
The following result is the key to achieve this.
\begin{proposition}
  \label{prop:sintsem}
  Let $T_E$ be locally finite, $K:= \exists \ui.\phi(\ui, a[\ui])$ be
  an $\exists^I$-formula, and $L$ be an $\exists^I$-formula. The
  following two conditions are equivalent:
  \begin{enumerate}[{\rm (i)}]
  \item for every $s$ in
    a basis for $K$, there exists a configuration $s'$ in a basis for
    $L$ such that $s\leq s'$;
  \item $L$ is (up to $A^E_I$-equivalence) of the form $\exists\ui,
    \uj.\psi(\ui, \uj, a[\ui], a[\uj])$ for a quantifier-free formula
    $\psi$ and
    \begin{eqnarray*}
      \mbox{if }
      A^E_I\models Min(\psi,a, \ui\,\uj)\to \theta(\ut, a[\ut])
      & \mbox{ then } &
      A^E_I\models Min(\phi,a,\ui)\to \theta(\ut, a[\ut]) ,
    \end{eqnarray*}
    for all quantifier free ($\Sigma_E\cup \Sigma_I$)-formula
    $\theta$ and for all tuple of terms $\ut$ taken from the set
    of the representative $\Sigma_I(\ui)$-terms.
 \end{enumerate}
\end{proposition}
\begin{proof}
  Assume (i). We first apply a syntactic transformation to $L$ as
  follows.  Let $B, B'$ be bases for $K, L$, respectively; we know
  that for every $(s, \cM_s)\in B$ there is $(s^L, \cM_s^L)\in B'$
  such that $s\leq s^L$: the relationship $s\leq s^L$ is due to the
  existence of a pair of embeddings $(\mu_s, \nu_s)$ as required by
  the configuration ordering definition. For every $s\in B$ and for
  every assignment $\mathtt{a}$ such that $\mathtt{a}(a)=s$ and
  $(\cM_s, \mathtt{a})\models \phi(\ui, a[\ui])$, we build the diagram
  formula $K_{\mathtt{a}}$ for $s^L$ given by
  \begin{equation}
    \label{eq:dd}
    \exists\ui\,\exists\uk\, (\delta_{s^L_I}(\ui, \uk)
      \wedge \delta_{s^L_E}(a[\ui], a[\uk]))
    \end{equation}
    where the variables $\uk$ are names for the elements in the
    complement subset $supp(s^L_I)\backslash\mu_s(\mathtt{a}(\ui))$
    (here $supp(s^L_I)$ is the support of the $\Sigma_I$-structure
    $s^L_I$).  Notice that the formula \eqref{eq:dd} is nothing but
    formula \eqref{eq:diagformula} used in the proof of
    Proposition~\ref{prop:conf1}(i).\footnote{It might happen here
      that duplicate variables are used because the $\mathtt{a}(\ui)$
      need not be distinct. This is not a problem: if different index
      variables (say $i_1, i_2$) naming the same element are employed,
      the diagram formula will contain a conjunct like $i_1=i_2$.  The
      embedding property of Robinson Diagram Lemma is not affected by
      these duplications.  } Since, for a configuration $t$, the fact
    that $t\in\mywidehat{K_{\mathtt{a}}}$ means that there are
    suitable embeddings witnessing that $s^L\leq t$, we have that
    $\mywidehat{L}=\mywidehat{L\vee\bigvee_{\mathtt{a}}
      K_{\mathtt{a}}}$, hence by Proposition~\ref{prop:conf} the
    formula $L$ is $A^E_I$-equivalent to $L\vee\bigvee_{\mathtt{a}}
    K_{\mathtt{a}}$.\footnote{The assignments are infinite, but only
      finitely many variables are mentioned in them, so that only
      finitely many formulae $K_{\mathtt{a}}$ can be produced.}  Up to
    logical equivalence, we can move the existentially quantified
    variables outside the disjunctions so that $L$ is equivalent to a
    prenex existential formula of the kind $\exists\ui\exists\uj
    \psi$. With this new syntactic form, the following property holds:
    for every $s\in B$ and for every assignment $\mathtt{a}$ such that
    $\mathtt{a}(a)=s$ and $(\cM_s, \mathtt{a})\models \phi(\ui,
    a[\ui])$, there is an assignment $\mathtt{a}^L$ such that (i)
    $(\cM_s^L, \mathtt{a}^L)\models \psi(\ui, \uj, a[\ui], a[\uj])$,
    (ii) $\mathtt{a}^L(\ui)=\mu_s(\mathtt{a}(\ui))$, and (iii)
    $\mathtt{a}^L(a)= s^L$. Since $s^L$ is in a basis of $L$, from
    Corollary~\ref{coro:min}, it follows also that $(\cM_s^L,
    \mathtt{a}^L)\models Min(\psi, a, \ui\,\uj)$.

    Suppose now that $A^E_I\not\models Min(\phi,a,\ui)\to
    \theta(\ut(\ui), a[\ut(\ui)])$; by Lemma~\ref{lem:min} (and by the
    fact that $\phi, \theta$ are quantifier-free) this means that
    there are a configuration $(s, \cM_s)\in B$ and an assignment
    $\mathtt{a}$ such that $(\cM_s, \mathtt{a})\models \phi(\ui,
    a[\ui])$ and $(\cM_s, \mathtt{a})\not \models \theta(\ut,
    a[\ut])$.  Since $\theta$ is quantifier-free, taking the
    assignment $\mathtt{a}^L$ satisfying (i)-(ii)-(iii) above, we get
    that $(\cM_s^L, \mathtt{ a}^L)\not \models \theta(\ut, a[\ut])$,
    thus also $(\cM_s^L, \mathtt{ a}^L)\not \models Min(\psi, a,
    \ui\,\uj) \to\theta(\ut, a[\ut])$.

    Conversely, assume (ii). Fix $(s, \cM_s)$ in a basis $B$ for $K$
    and an assignment $\mathtt{a}$ such that $(\cM_s,
    \mathtt{a})\models \phi(\ui, a[\ui])$; by
    Corollary~\ref{coro:min}, we have that $(\cM_s, \mathtt{a})\models
    Min(\phi,a,\ui)$.  Let $\ut$ be the representative
    $\Sigma_I(\ui)$-terms and let $\theta(\ut(\ui), a[\ut(\ui)])$ be
    the negation of the formula $\delta_{s_I}(\ut(\ui))\wedge
    \delta_{s_E}(a[\ut(\ui)])$.
    We have $(\cM_s, \mathtt{a})\not\models Min(\phi,a,\ui) \to
    \theta(\ut, a[\ut])$, hence there are $\cN$ and $\mathtt{b}$ such
    that $(\cN, \mathtt{b})\not\models Min(\psi,a,\ui\,\uj) \to
    \theta(\ut, a[\ut])$. By restricting the support of $\cN_I$ if
    needed, we can suppose that $\cN$ is a finite index model and that
    $\cN_I$ is generated by the elements assigned by $\mathtt{b}$ to
    the $\ui, \uj$.  Let $s'$ be $\mathtt{b}(a)$: from
    Lemma~\ref{lem:min} it follows that $s'$ is in a basis for $L$;
    also, from the fact that $(\cN, \mathtt{b})\not\models\theta(\ut,
    a[\ut])$, we can conclude that $s\leq s'$, as desired.
\end{proof}
In the following, we will write $K\leq L$ whenever one of the
(equivalent) conditions in Proposition~\ref{prop:sintsem} holds. We
show that, under the working assumption that $T_E$ is locally finite,
it is possible to compute all the finitely many (up to
$A^E_I$-equivalence) $\exists^I$-formulae $K$ such that $K\leq L$.
\begin{proposition}
  \label{prop:finitelymany}
  Let $T_E$ be locally finite.  Given an $\exists^I$-formula $L$,
  there are only finitely many (up to $A^E_I$-equivalence)
  $\exists^I$-formulae $K$ such that $K\leq L$ and all such $K$ can be
  effectively computed.
\end{proposition}
\begin{proof}
  Suppose that $L$ is of the form $\exists \uk \gamma$.  To use the
  criterion of Proposition~\ref{prop:sintsem}(ii) in an effective way,
  we only need to find a bound for the length of the tuples $\ui$ and
  $\uj$. In fact, once the bound is known the search space for
  formulae of the forms $\exists\ui \,\exists\uj \, \psi$ and $\exists
  \ui \, \phi$ satisfying the conditions (which can be effectively
  checked by using Theorem~\ref{th:decidability})
  \begin{eqnarray*}
    & A^E_I\models \exists \uk \gamma\leftrightarrow \exists\ui
    \,\exists\uj \, \psi, \qquad {\rm and~for~all}~\theta(\ut(\ui),
    a[\ut(\ui)])~~~~~~~~~~~~~~~~~~~~~~~~~ \\ &
    A^E_I\models Min(\psi,a, \ui\,\uj)\to \theta(\ut, a[\ut])\quad
    \Rightarrow \quad A^E_I\models Min(\phi,a,\ui)\to \theta(\ut,
    a[\ut])
  \end{eqnarray*}
  is finite.  This is because $T_I$ and $T_E$ are both locally finite
  and hence, there are only finitely many quantifier-free formulae of
  the required type involving a fixed number of index variables which
  are not $A^E_I$-equivalent. The proof of
  Proposition~\ref{prop:sintsem} shows that the lengths of $\ui$ and
  $\uj$ are both bounded by the maximum cardinality $N$ of the support
  of $s_I$, where $s_I$ is a configuration that belongs to a basis for
  $L\equiv\exists\uk \, \gamma$.  For $\uj$, this is clear from the
  proof itself while for $\ui$, it is a consequence of the following
  considerations.  First, we can restrict the search to formulae $K$
  of the form $\exists \ui \, \phi$, where the length of $\ui$ is
  minimal, i.e.\ $K$ is not be equivalent to a formula with a shorter
  existential prefix.  Furthermore, by Proposition~\ref{prop:conf1},
  $K$ is equivalent to $K_{s_1}\vee \cdots \vee K_{s_n}$, where $\{
  s_1, \dots, s_n\}$ is a basis for $K$.  In turn,
  by~\eqref{eq:diagformula}, this means that there must exist a
  configuration $t$ in a basis for ${K}$ such that the cardinality of
  $t_I$ is bigger than or equal to the length of $\ui$; since $t\leq
  s$ for some $s$ in a basis for $L$ (see
  Proposition~\ref{prop:sintsem}(ii)), we have that the length of
  $\ui$ cannot exceed $N$.  To conclude, it is sufficient to observe
  that $N$ cannot be bigger than the number of the representative
  $\Sigma_I(\uk)$-terms.
\end{proof}
\begin{defi}
  We say that \emph{$K$ covers $L$} iff both $K\leq L$ and
  $A^E_I\models L\to K$.
\end{defi}
The following example illustrates the notions just introduced and will
be useful also when discussing the implementation of our invariant
synthesis technique (see Section~\ref{sec:back+inv-heu} below).
\begin{example}\em
  \label{ex:cover}
  Let $\Sigma_I$ be relational and $T_E $ be a locally finite theory
  admitting elimination of quantifiers.  Let
  \begin{eqnarray}
    \label{eq:inv_candidate}
    L & := &
    \exists \ui\,\uj.(\psi_E(a[\ui], a[\uj])  \wedge
                    \psi_I(\ui, \uj)        \wedge
                    \delta_I(\ui) )          
  \end{eqnarray}
  be a primitive differentiated and $A^E_I$-satisfiable
  $\exists^I$-formula such that (i) $\ui\cap \uj=\emptyset$, (ii)
  $\psi_E(\ue, \ud)$ is a conjunction of $\Sigma_E$-literals; (iii)
  $\psi_I(\ui, \uj)$ is a conjunction of $\Sigma_I$-literals; (iv)
  $\delta_I(\ui)$ is a maximal conjunction of $\Sigma_I(\ui)$-literals
  (i.e. for every $\Sigma(\ui)$-atom $A(\ui)$, $\delta_I$ contains
  either $A(\ui)$ or its negation).  If
  \begin{eqnarray}
    \label{eq:res}
    K &:=& \exists \ui \,(\delta_I(\ui)\wedge \phi_E(a[\ui])),
  \end{eqnarray}
  where $\phi_E(\ue)$ is $T_E$-equivalent to $\exists \ud\,
  \psi_E(\ue, \ud)$,\footnote{$\phi_E$ is guaranteed to exist as $T_E$
    admits elimination of quantifiers.} then $K$ covers $L$ and in
  particular $K\leq L$.  We prove this fact in the following.
  \begin{proof}
    We use Proposition~\ref{prop:sintsem}(ii): as shown in
    Remark~\ref{rem}, since $L$ and $K$ are differentiated, we can
    avoid mentioning the corresponding formulae $Min$ in the condition
    of Proposition~\ref{prop:sintsem}(ii) and just prove that
    \begin{eqnarray*}
      & A^E_I\not \models
      \delta_I(\ui)\wedge \phi_E(a[\ui])\to \theta(\ui, a[\ui]) ~~
      \Rightarrow~  \\
      &~\hskip 3.5cm A^E_I\not \models  \delta_I(\ui)\wedge \psi_I(\ui, \uj)\wedge \psi_E(a[\ui], a[\uj]) \to
      \theta(\ui, a[\ui])
    \end{eqnarray*}
    for every $\theta$ (notice that, since $\Sigma_I$ is relational,
    the only $\Sigma_I(\ui)$-terms are the $\ui$).  Pick a model $\cM$
    and an assignment $\mathtt{a}$ such that $(\cM,
    \mathtt{a})\models\delta_I(\ui)\wedge \phi_E(a[\ui])$ and $(\cM,
    \mathtt{a})\not\models\theta(\ui, a[\ui])$. We can freely assume
    that that the support of $\cM_I$ is a $\Sigma_I$-structure
    generated by the $\mathtt{a}(\ui)$; by modifying the value of
    $\mathtt{a}$ on the element variables $\ud$, if needed, we can
    also assume that $(\cM, \mathtt{a})\models\psi_E(a[\ui], \ud)$
    (this is because $\phi_E(\ue)$ is $T_E$-equivalent to $\exists
    \ud\, \psi_E(\ue, \ud)$). Since $L$ is consistent, there are also
    a model $\cN$ and an assignment $\mathtt{b}$ such that $(\cN,
    \mathtt{b})\models \delta_I(\ui)\wedge \psi_I(\ui, \uj)\wedge
    \psi_E(a[\ui], a[\uj])$. Again, we can assume that the support of
    $\cN_I$ is a $\Sigma_I$-structure generated by the
    $\mathtt{a}(\ui, \uj)$; since $\delta_I(\ui)$ is maximal, it is a
    diagram formula, hence (up to an isomorphism) $\cM_I$ is a
    substructure of $\cN_I$. Let us now take the model $\cN'$, whose
    $\Sigma_I$-reduct is $\cN_I$ and whose $\Sigma_E$-reduct is
    $\cM_E$. Let $\mathtt{b}'$ be the assignment which is like
    $\mathtt{b}$ as far as the index variables $\ui, \uj$ are
    concerned and which associates with the variable $a$ the array
    whose $\mathtt{b}'(\ui)$-values are the
    $\mathtt{b}'(\ui)=\mathtt{a}(\ui)$-values of $\mathtt{a}(a)$ and
    whose $\mathtt{b}'(\uj)$-values are the $\ud$ (notice that this is
    correct because by differentiatedness of $L$ the
    $\mathtt{b}'(\ui\uj)$ are all distinct). It turns out that $(\cN',
    \mathtt{b}')\not \models \delta_I(\ui)\wedge \psi_I(\ui,
    \uj)\wedge \psi_E(a[\ui], a[\uj]) \to \theta(\ui, a[\ui])$, as
    desired.
\end{proof}
\end{example}
We are now in the position to take the final step towards the goal of
obtaining a completely symbolic method for
restating
the results from Theorem~\ref{th:inv}.  Let $\mathsf{ChooseCover}(L)$
be a procedure that returns non-deterministically one of the
$\exists^I$-formulae $K$ such that $K$ covers $L$ (this procedure is
playing the role of the function $\gamma$ from the proof of
Theorem~\ref{th:inv}).  We consider the procedure $\mathsf{SInv}$ in
Figure~\ref{fig:reach-algo} (b) for the computation of safety
invariants and prove its correctness.
\begin{thm}
  \label{th:invsymb}
  Let $T_E$ be locally finite. Then, there exists a safety invariant
  for $U$ iff the procedure $\mathsf{SInv}$ in
  Figure~\ref{fig:reach-algo} (b) returns a safety invariant for $U,$
  for a suitable $\mathsf{ChooseCover}$ function.
\end{thm}
\begin{proof}
  Suppose that $\mathsf{SInv}$ returns $B$ after $k+1$ iterations of
  the loop: we show that $\neg B$ is a safety invariant.  Notice that
  $B$ is a disjunction $P_0\vee\cdots\vee P_k$ of $\exists^I$-formulae
  such that for all $i=0, \dots, k$,
  \begin{description}
  \item[{\rm (I)}] the formula $I\wedge P_i$ is not
    $A^E_I$-satisfiable;
  \end{description}
  also $P_{i}$ covers $Pre(\tau, P_{i-1})$ and $P_0$ covers $U$, which
  means in particular that
  \begin{description}
  \item[{\rm (II)}] $A^E_I\models \forall a\,(Pre(\tau,P_{i-1})(a)\to
    P_{i}(a))$ and $A^E_I\models \forall a\,(U(a)\to P_{0}(a))$.
  \end{description}
  Finally, $\mathsf{SInv}$ could exit the loop because for some
  $P_{k+1}$ covering $Pre(\tau, P_k)$, it happened that $P_{k+1}\wedge
  \neg B$ was not $A^E_I$-satisfiable: these two conditions entail
  that
  \begin{description}
  \item[{\rm (III)}] $A^E_I\models \forall a\,(Pre(\tau, P_k)(a)\to
    B(a))$.
  \end{description}
  Conditions (i) and (iii) of Definition~\ref{def:inv} now easily
  follows from (I) and (II); we only need to check condition (ii) of
  Definition~\ref{def:inv}, namely (up to logical equivalence) that
  $A^E_I\models \forall a \,(Pre(\tau, B)(a)\to B(a))$: since
  $Pre(\tau, B)$ is logically equivalent to the disjunction
  $\bigvee_{i=0}^n Pre(\tau, P_i)$, the claim follows immediately from
  (II)-(III).

  Let us now prove the converse, i.e.\ that in case a safety
  invariants exists, $\mathsf{SInv}$ is able to compute one.  Recall
  the proof of Theorem~\ref{th:inv}: given the negation $H$ of a
  safety invariant for $U$, another negation $L$ of a safety invariant
  for $U$ is produced in the following way. Define the sequence of
  $\exists^I$-formulae $L_i$ as follows: (i) $L_0:=\gamma(U)$ and (ii)
  $L_{i+1}:= L_i\vee \gamma(Pre(\tau, L_i))$. Our $L$ is the $L_i$
  with the smallest $i$ such that $L_{i+1}$ is $A^E_I$-equivalent to
  $L_i$ (the proof of Theorem~\ref{th:inv} guarantees that such an $i$
  exists).

  The above recursive definition for $L_i$ is based on the function
  $\gamma$, which is defined (non symbolically) by making use of
  configurations.  Actually, for an $\exists^I$-formula $S$ such that
  $\mywidehat{S}\subseteq \mywidehat{H}$, the function $\gamma(S)$
  returns an $\exists^I$-formula $K_{a_1}\vee\cdots\vee K_{a_n}$,
  where $\{a_1, \dots, a_n\}\subseteq \mywidehat{H}$ is a minimal set
  of configurations taken from a basis of $H$ such that
  $\mywidehat{S}\subseteq\,\uparrow\! a_1\cup \cdots \cup \uparrow\!
  a_n$. Using Proposition~\ref{prop:sintsem}, it is not difficult to
  see that minimality implies $\gamma(S)\leq S$: in fact, condition
  $\mywidehat{S}\subseteq\,\uparrow\! a_1\cup \cdots \cup \uparrow\!
  a_n$ says that for every $s$ in a basis for $S$ there is $a_i$ in
  the basis $\{a_1, \dots, a_n\}$ for $\gamma(S)$ such that $a_i\leq
  s$, but the converse (which is what really matters for us in view of
  Proposition~\ref{prop:sintsem}(i)) must hold too, by
  minimality. This can be shown as follows: if any $a_i$ is
  eliminated, the relation $$\mywidehat{S}\subseteq\,\uparrow\!
  a_1\cup \cdots \cup \uparrow\! a_{i-1}\cup \uparrow\! a_{i+1}\cup
  \cdots\cup\uparrow\!  a_n$$ does no longer hold, hence there is an
  $s$ from a basis of $S$ such that $a_j\not\leq s$ for all $j=1,
  \dots, i-1, i+1, \dots n$. Since, on the contrary,
  $\mywidehat{S}\subseteq\,\uparrow\! a_1\cup \cdots \cup \uparrow\!
  a_n$ holds, we must conclude that $a_i\leq s$. Hence for every $a_i$
  there is an $s$ in a basis of $S$ such that $a_i\leq s$.

  Thus $\gamma(S)$ is such that $\gamma(S)\leq S$ and $A^E_I\models
  S\to \gamma(S)$, i.e. $\gamma(S)$ covers $S$.
  It is then clear that an appropriate choice of the function
  $\mathsf{ChooseCover}$ in $\mathsf{SInv}$ can return precisely the
  formulae $L_i$ so that they are assigned to the variable $B$ at the
  $i$th-loop of the procedure, thus justifying the claim of the
  Theorem.
\end{proof}
When $\mathsf{ChooseCover}(L)=L$, i.e.\ $\mathsf{ChooseCover}$ is the
identity (indeed, $L$ covers $L$), the procedure $\mathsf{SInv}$ is
the (exact) dual of $\mathsf{BReach}$ in Figure~\ref{fig:reach-algo}
(a) and, hence it can only return (the negation of) a symbolic
representation of all backward reachable states as a safety invariant.

\section{Pragmatics of Invariant Synthesis and Experiments}
\label{sec:implementation}

The main drawback of algorithm $\mathsf{SInv}$ (in
Figure~\ref{fig:reach-algo} (b), explained in the last section) is the
non determinism of the function $\mathsf{ChooseCover}$.  Although
finite, the number of formulae covering a certain $\exists^I$-formula
is so large to make any concrete implementation of $\mathsf{SInv}$
impractical.  Instead, we prefer to study how to integrate the
synthesis of invariants into the backward reachability algorithm of
Figure~\ref{fig:reach-algo} (a). Given that finding a \emph{safety
  invariant} could be infeasible through an exhaustive search, we
content ourselves to find invariants \emph{tout court} and use them to
prune the search space of the backward reachability algorithm
$\mathsf{BReach}$ (in Figure~\ref{fig:reach-algo} (a)).

\subsection{Integrating Invariant Synthesis within Backward
  Reachability}
\label{sec:back+inv}

In our symbolic framework, at the $n$-th iteration of the loop of the
procedure $\mathsf{BReach}$, the set of backward reachable states is
represented by the formula stored in the variable $B$ (which is
equivalent to $BR^n(\tau,U)$).  So, `pruning the search space of the
backward reachability algorithm' amounts to disjoining the negation of
the available invariants to $B$.  In this way, the extra information
encoded in the invariants makes the satisfiability test at line 2 (for
fix-point checking) more likely to be successful and possibly
decreases the number of iterations of the loop.  Indeed, the problem
is to synthesize such invariants. Let us consider
this problem at a very abstract level.

Suppose the availability of a function $\mathsf{Choose}$ that takes an
$\exists^I$-formula $P$ and returns a (possibly empty) finite set $S$
of $\exists^I$-formulae representing `useful (with respect to $P$)
candidate invariants.'  We can integrate the synthesis of invariants
within the backward reachability algorithm by adding between lines 4
and 5 in Figure~\ref{fig:reach-algo} (a) the following instructions:
\begin{center}
\begin{minipage}{.7\textwidth}
\begin{tabbing}
  foo \= foo \= \kill
  4$'$\> \> \textbf{foreach} $\mathit{CINV} \in \mathsf{Choose}(P)$
            \textbf{do} \\
  \> \> \hspace{.25cm} \textbf{if} $\mathsf{BReach}(\mathit{CINV})=
                                     (\mathsf{safe},B_{\mathit{CINV}})$
           \textbf{then} $B\longleftarrow B\vee \neg B_{\mathit{CINV}}$;
\end{tabbing}
\end{minipage}
\end{center}
where $\mathit{CINV}$ stands for `candidate invariant.'  The resulting
procedure will be indicated with $\mathsf{BReach\mbox{+}Inv}$ in the
following.  Notice that $\mathsf{BReach}$ is used here as a
sub-procedure of $\mathsf{BReach\mbox{+}Inv}$.
\begin{proposition}
  \label{prop:soundness}
  If the procedure
  $\mathsf{BReach\mbox{+}Inv}$ terminates by returning $\mathsf{safe}$
  ($\mathsf{unsafe}$), then $\cSi$ is safe (unsafe) with respect to
  $U$.
\end{proposition}
\begin{proof}
  The claim is trivial when $\mathsf{BReach\mbox{+}Inv}$ returns
  $\mathsf{unsafe}$.  Let us consider the situation when the procedure
  terminates by returning $\mathsf{safe}$ at the $(k+1)$-th iteration
  of the main loop.  Observe that the content of the variable $B$ is
  \begin{equation}
    \label{eq:B_content}
    Pre^0(\tau,U)\vee Pre^1(\tau, U)\vee\cdots \vee Pre^k(\tau, U)\vee
    H_1\vee\cdots \vee H_m
  \end{equation}
  at the $(k+1)$-th iteration of the loop, where $H_1, \dots, H_m$ are
  negations of invariants (see Property~\ref{ex:basic}). For reductio,
  suppose that the system is unsafe, i.e.\ for some $n\geq 0$, the
  formula~\eqref{eq:unsafe} (shown here for the sake of readability)
  \begin{eqnarray*}
  I(a_n)\wedge \tau(a_{n}, a_{n-1})\wedge \cdots \wedge \tau(a_{1},
  a_0)\wedge U(a_0)
  \end{eqnarray*}
  is $A_I^E$-satisfiable.  Assume that the formula is true in a model
  of $A^E_I$ with the array assignments $s_n, \dots, s_0$; in the
  following, we say that $s_n, \dots, s_0$ is a \emph{bad trace}.  We
  also assume that $s_n, \dots, s_0$ is a bad trace of shortest
  length.  Since the formulae $I\wedge Pre^0(\tau,U)$, $I\wedge
  Pre^1(\tau, U)$, $\dots$, and $I\wedge Pre^k(\tau, U)$ are all
  $A^E_I$-unsatisfiable (see line 3 of Figure~\ref{fig:reach-algo}
  (a), which is also part of $\mathsf{BReach\mbox{+}Inv}$), it must be
  $n> k$.  Let us now focus on $s_{k+1}$; since
  $\mathsf{BReach\mbox{+}Inv}$ returned $\mathsf{safe}$ at iteration
  $k+1$, it must have exited the loop because the formula currently
  stored in $P$ (which is $Pre^{k+1}(\tau, U)$) is not
  $A^E_I$-satisfiable with the negation of the formula currently in
  $B$ (which is~\eqref{eq:B_content}). Hence, $s_{k+1}$ (which
  satisfies $Pre^{k+1}(\tau, U)$) must satisfy either some $Pre^{l}$
  (for $l<k+1$) or some $H_i$, but both alternatives are
  impossible. In fact, the former would yield a shorter bad trace,
  whereas the latter is in contrast to the fact that $s_{k+1}$ is
  forward reachable from a state satisfying $I$ and, as such, it
  should satisfy the invariant $\neg H_i$.
\end{proof}
The procedure $\mathsf{BReach\mbox{+}Inv}$ is
\begin{enumerate}[$\bullet$]
\item incomplete, in the sense that it is not guaranteed to terminate
  even when a safety invariant exists,
\item deterministic, since no backtracking is required,
\item highly parallelizable: it is possible to run in parallel as many
  instances of $\mathsf{BReach}$ as formulae in the set returned by
  $\mathsf{Choose}$, and
\item it performs well (for \emph{appropriate} $\mathsf{Choose}$
  functions, see below for a discussion of the meaning of
  ``appropriate'' in this context) as witnessed by the experiments in
  the next section.
\end{enumerate}
As a result, invariant synthesis becomes a powerful \emph{heuristic}
within a refined version of the basic backward reachability algorithm.
Furthermore, its integration in the tableaux calculus of
Section~\ref{subsec:tab} is particularly easy: \emph{just use the
  calculus itself with some bounds on the resources} (such as a limit
on the depth of the tree) to check if a candidate invariant is a
``real'' invariant.  Indeed, the crucial point is how to design an
\emph{appropriate} function $\mathsf{Choose}$.  There are several
possible criteria leading to a variety of implementations for
$\mathsf{Choose}$.  The usefulness of the resulting functions is
likely to depend on the application.  Despite the complexity of the
design space, it is possible to identify a \emph{minimal requirement}
on $\mathsf{Choose}$ by taking into account the tableaux calculus
introduced in Section~\ref{subsec:tab}. To this end, recall that
backward reachable sets of states are described by primitive
differentiated formulae and that a formula $P$ representing a
pre-image is eagerly expanded to disjunctions of primitive
differentiated formulae by using the $\mathsf{Beta}$ rule.  Thus, a
reasonable implementation of $\mathsf{Choose}$ should be such that
$\mathsf{Choose}(P)=S$ where \emph{$S$ is a set of primitive
  differentiated formulae such that each $Q'\in S$ is implied by a
  disjunct $Q$ occurring in the disjunction of primitive
  differentiated formulae obtained as expansion of $P$}.  In this way,
each $Q'\in S$ can be seen as a \emph{tentative over-approximation} of
$Q$.  (Notice that guessing a candidate invariant can be seen as a
form of abstraction.)  All the implementations of the function
$\mathsf{Choose}$ in \textsc{mcmt} satisfy the minimal requirement
above and can be selected by appropriate command line options and
directives to be included in the input file (the interested reader is
pointed to the user manual available in the distribution for details).
We now describe two types of abstractions that lead to different
implementations of the function $\mathsf{Choose}$ that are available
in the current release of \textsc{mcmt}.

\subsection{Index  Abstraction}
\label{sec:back+inv-heu}
Index abstraction amounts to eliminating some index variables; if done
in the appropriate way, this is equivalent to replacing configurations
with sub-con\-fi\-gu\-ra\-tions (as discussed in
Section~\ref{sec:inv}).  Thus, it is possible to design approximations
(quite loose, but suitable for implementation) of the procedure
suggested in the proof of Theorem~\ref{th:invsymb}.  An idea (close to
what is implemented in the current release of \textsc{mcmt}) is to
follow the suggestions in Example~\ref{ex:cover} so as to satisfy the
minimal requirement discussed above on $\mathsf{Choose}$.
More precisely, given $Q:=\exists \uk.\theta(\uk,a[\uk])$, we first
try to transform it into the form of~\eqref{eq:inv_candidate}, i.e.
\begin{eqnarray*}
  \exists \ui\,\uj.(\psi_E(a[\ui], a[\uj])  \wedge
                    \psi_I(\ui, \uj)        \wedge
                    \delta_I(\ui) ).
\end{eqnarray*}
To do this, we decompose $\uk$ into two disjoint sub-sequences $\ui$
and $\uj$ such that $\uk=\ui\cup \uj$ according to some criteria: if
the conjunction of $\Sigma_I(\ui)$ literals occurring in $\theta$ is
maximal, we get a candidate invariant by returning the corresponding
$\exists^I$-formula~\eqref{eq:res}, i.e.
\begin{eqnarray*}
  \exists \ui \,(\delta_I(\ui)\wedge \phi_E(a[\ui])) .
\end{eqnarray*}
This is computationally feasible in many situations.  For example,
quantifier elimination reduces to a trivial substitution if $T_E$ is
an enumerated data-type theory and the $\Sigma_E$-literals in $\theta$
(i.e.\ those in $\psi_E$) are all positive.  The maximality of
$\theta$ is guaranteed (by being differentiated) if $T_I$ is the
theory of finite sets.  Another case in which maximality of $\theta$
is guaranteed is when
 $T_I$ is the theory of linear orders and $\ui=i_1$ or
($\ui=i_1, i_2$ and $\theta$ contains the atom $i_1 < i_2$). In more
complex cases, it is possible to obtain
a useful formula (similar to~\eqref{eq:res}) in a purely syntactic and
computationally cheap way.  There is no risk in using methods giving
very coarse approximations since
a candidate invariant is used for pruning the search space of the
backward reachability procedure only if it has been proved to be a
``real'' invariant (see also Remark~\ref{remr} below).

\subsection{Signature Abstraction}
\label{subsec:abs-inv}
Index abstraction can be useless or computationally too expensive (if
done precisely) in several applications.  Even worse, when $T_E$ is
not locally finite, the related notion of sub-configuration loses most
of its relevance.  In these cases, other forms of abstraction inspired
to predicate abstraction~\cite{seminal} may be of great help.
Although predicate abstraction with refinement (as in the CEGAR loop)
is not yet implemented in \textsc{mcmt}, it features a technique for
invariant synthesis that we have called \emph{signature abstraction},
which can be seen as a simplified version of predicate abstraction.
This technique uses quantifier elimination (whenever possible) to
eliminate the literals containing a selected sub-set $X$ of the set of
array variables.  The subset $X$ can either be suggested by the user
or dynamically built by the tool from the shape of the disjunct
belonging to the pre-image being currently computed.  Again, the
elimination is applied to each of the primitive differentiated
disjuncts of the currently computed pre-image $P$ to obtain the
differentiated formulae to form the set of formulae returned by
$\mathsf{Choose}$.  It is easy to see that this way of implementing
the function $\mathsf{Choose}$ satisfies the minimal requirement
discussed above.
\begin{rem}
  \label{remr}
  The reader may wonder whether the use of abstraction techniques can
  have a negative impact on the correctness of \textsc{mcmt} outcome.
  We emphasize that this is not the case because of the way the
  candidate invariants are used to prune the search space during
  backward reachability.  In fact, abstraction is just to generate the
  candidate invariants which are then tested to be ``real'' invariants
  by a resource bounded version of backward reachability.  Only if
  candidate invariants pass this test, they are used to prune the
  search of backward reachable states.  In other words, the answer
  supplied by \textsc{mcmt} to a safety problem is always correct: as
  it is clear from the proof of Proposition~\ref{prop:soundness}, the
  set of backward reachable states can be augmented if invariants are
  used during backward search, but it is augmented by adding it only
  states satisfying the negation of an invariant (these states are not
  forward reachable, hence they cannot alter safety checks). As a
  consequence, safety tests remain exhaustive,
%
  although it may happen that resources (such as computation
  time) are wasted in checking candidate invariants that turn out not
  to be ``real'' invariants
  or not to be useful to significantly  prune the search space.
\end{rem}

\subsection{Experiments}
\label{sec:exp}
To show the flexibility and the performances of \textsc{mcmt}, we have
built a library of benchmarks in the format accepted by our tool by
translating from a variety of sources safety problems.  More
precisely, our sources were the following:
\begin{enumerate}[$\bullet$]
\item parametrised systems from the distribution of the infinite model
  checkers \textsc{pfs}
  (\texttt{\url{http://www.it.uu.se/research/docs/fm/apv/tools/pfs}})
  and Undip
  (\texttt{\url{http://www.it.uu.se/research/docs/fm/apv/tools/undip}}),
\item parametrised and distributed systems from the invisible
  invariant methods (see, e.g.,~\cite{BFPZ05}),
\item imperative programs manipulating arrays (such as sorting or
  string manipulation) taken from standard books about algorithms,
\item imperative programs manipulating numeric variables from the
  distribution of the model checker ARMC
  (\texttt{\url{http://www7.in.tum.de/~rybal/armc}}),
\item protocols from the distribution of Mur$\phi$ extended with
  predicate abstraction
  (\texttt{\url{http://verify.stanford.edu/satyaki/research/PredicateAbstractionExamples.}} \linebreak \texttt{\url{html}}).
\end{enumerate}
We did not try to be exhaustive in the selection of problems but
rather to pick problems from the wider possible range of different
classes of infinite state systems so as to substantiate the claim
about the flexibility of our tool.  All the files in \textsc{mcmt}
format are contained in the \textsc{mcmt} distribution which is
available at the tool web page
(\texttt{\url{http://homes.dsi.unimi.it/~ghilardi/mcmt}}).  Each file
comes with the indication of source from which it has been adapted and
a brief informal explanation about its content.

We divided the problems into four categories: mutual exclusion and
cache coherence protocols taken mainly from the distributions of
\textsc{pfs} and Undip (see Tables~\ref{subapp:mep}
and~\ref{subapp:ccp}), imperative programs manipulating arrays (see
Table~\ref{subapp:ip}), and heterogeneous problems (see
Table~\ref{subapp:misc}) taken from the remaining sources listed
above.
For the first two categories, the benchmark set is sufficiently
representative, whereas for the last two categories just some
interesting examples have been submitted to the tool.
For each category, we tried the tool in two configurations: one,
called ``Default Setting,'' is the standard setting used when
\textsc{mcmt} is invoked without any option and the other, called
``Best Setting,'' is the result of some experimentation with various
heuristics for invariant synthesis, signature abstraction, and
acceleration.  It is possible that for some problems, the ``real''
best setting is still to be identified and the results reported here
can be further improved.

\begin{table}[t]
\small
\caption{\label{subapp:mep}Mutual exclusion protocols}
\begin{tabular}{||l|l|l|l|l|l||l|l|l|l|l|l||}
\hline\hline
& \multicolumn{5}{|c||}{Default Setting} & 
  \multicolumn{6}{|c||}{Best Setting} \\ \hline
Problem  & d & \#n & \#d & \#SMT & time 
         & d & \#n & \#d & \#SMT & \#inv. & time \\ \hline\hline
Bakery & 	2 & 	1 & 	0 & 	6 & 	0.00
       & 	2 & 	1 & 	0 & 	6 & 	0 & 	0.00 \\ \hline
Bakery\_bogus & 	8 & 	90 & 	14 & 	1413 & 	0.81
              & 	8 & 	53 & 	4 & 	1400 & 	7 & 	0.68 \\ \hline
Bakery\_e & 	12 & 	48 & 	17 & 	439 & 	0.20 
          & 	7 & 	8 & 	1 & 	213 & 	16 & 	0.10 \\ \hline
Bakery\_Lamport & 	12 & 	56 & 	15 & 	595 & 	0.27 
                & 	4 & 	7 & 	1 & 	209 & 	7 & 	0.08 \\ \hline
Bakery\_t & 	9 & 	28 & 	5 & 	251 & 	0.11 
          & 	7 & 	8 & 	1 & 	134 & 	5 & 	0.06 \\ \hline
Burns & 	14 & 	56 & 	7 & 	373 & 	0.14 
      & 	2 & 	2 & 	1 & 	53 & 	3 & 	0.02 \\ \hline
Dijkstra & 	14 & 	122 & 	37 & 	2920 & 	2.11 
         & 	2 & 	1 & 	1 & 	215 & 	12 & 	0.08 \\ \hline
Dijkstra1 & 	13 & 	38 & 	11 & 	222 & 	0.10 
          & 	2 & 	1 & 	1 & 	35 & 	2 & 	0.02 \\ \hline
Distrib\_Lamport & 	23 & 	913 & 	242 & 	47574 & 	120.62 
                 & 	23 & 	248 & 	42 & 	19254 & 	7 & 	32.84 \\ \hline
Java  M-lock & 	9 & 	23 & 	2 & 	289 & 	0.10 
             & 	9 & 	23 & 	2 & 	289 & 	0 & 	0.10 \\ \hline
Mux\_Sem & 	7 & 	8 & 	2 & 	57 & 	0.02
         & 	2 & 	1 & 	1 & 	65 & 	6 & 	0.02 \\ \hline
Rickart\_Agrawala & 	13 & 	458 & 	119 & 	35355 & 	187.04 
                  & 	13 & 	458 & 	119 & 	35355 & 	0 & 	187.04 \\ \hline
Sz\_fp & 	22 & 	277 & 	3 & 	7703 & 	5.12 
       & 	22 & 	277 & 	3 & 	7703 & 	0 & 	5.12 \\ \hline
Sz\_fp\_ver & 	30 & 	284 & 	38 & 	10611 & 	6.66 
            & 	30 & 	284 & 	38 & 	10611 & 	0 & 	6.66 \\ \hline
Szymanski & 	17 & 	136 & 	10 & 	2529 & 	1.60 
          & 	9 & 	14 & 	5 & 	882 & 	12 & 	0.30 \\ \hline
Szymanski\_at & 	23 & 	1745 & 	311 & 	424630 & 	540.19 
              & 	9 & 	22 & 	10 & 	2987 & 	42 & 	1.25 \\ \hline
Ticket & 	9 & 	18 & 	0 & 	284 & 	0.17 
       & 	9 & 	18 & 	0 & 	284 & 	0 & 	0.17 \\ \hline\hline
\end{tabular}
\end{table}
In Tables~\ref{subapp:mep},~\ref{subapp:ccp},~\ref{subapp:ip},
and~\ref{subapp:misc}, the column `d' is the depth of the tableaux
obtained by applying the rules listed in Section~\ref{subsec:tab},
`\#n' is the number of nodes in the tableaux, `\#d' is the number of
nodes which are deleted because they are subsumed by the information
contained in the others (see~\cite{afm09} for details about this
point), `\#SMT' is the number of invocations to Yices during backward
reachability to solve fix-point and safety checks, `\#inv.'  is the
number of invariants found by the available invariant synthesis
techniques (see also~\cite{tableaux09} for a more in-depth discussion
on some of these issues),\footnote{In the table for the ``Default
  Setting,'' the column labelled with `\#inv.' is not present because
  \textsc{mcmt}'s default is to turn off invariant synthesis.} and
`time' is the total amount of time (in seconds) taken by the tool to
solve the safety problem.  Timings were obtained on a Intel Centrino
1.729 GHz with 1 Gbyte of RAM running Linux Gentoo.  In some cases,
the system seemed to diverge as it clearly entered in a loop: it kept
applying the same sequence of transitions.  In these cases, we stopped
the system, left the corresponding line of the table empty, and put
`timeout' in the last column.

\begin{table}[t]
\small
\caption{\label{subapp:ccp}Cache coherence protocols}
\begin{tabular}{||l|l|l|l|l|l||l|l|l|l|l|l||}
\hline\hline
& \multicolumn{5}{|c||}{Default Setting} & 
  \multicolumn{6}{|c||}{Best Setting} \\ \hline
Problem  & d & \#n & \#d & \#SMT & time 
         & d & \#n & \#d & \#SMT & \#inv. & time \\ \hline\hline
Berkeley & 	2 & 	1 & 	0 & 	16 & 	0.00 
         & 	2 & 	1 & 	0 & 	16 & 	0 & 	0.00 \\ \hline
Futurebus & 	8 & 	37 & 	3 & 	998 & 	0.96 
          & 	8 & 	37 & 	3 & 	998 & 	0 & 	0.96 \\ \hline
German07 & 	26 & 	2442 & 	576 & 	121388 & 	145.68 
         & 	26 & 	2442 & 	576 & 	121388 & 	0 & 	145.68 \\ \hline
German\_buggy & 	16 & 	1631 & 	203 & 	41497 & 	49.70 
              & 	16 & 	1631 & 	203 & 	41497 & 	0 & 	49.70 \\ \hline
German\_ca & 	9 & 	13 & 	0 & 	62 & 	0.03 
           & 	9 & 	13 & 	0 & 	62 & 	0 & 	0.03 \\ \hline
German\_pfs & 	33 & 	11605 & 	2755 & 	858184 & 	31m 01s 
            & 	33 & 	11141 & 	2673 & 	784168 & 	149 & 	30m 27s \\ \hline
Illinois & 	4 & 	8 & 	0 & 	144 & 	0.08 
         & 	4 & 	8 & 	0 & 	144 & 	0 & 	0.08 \\ \hline
Illinois\_ca & 	3 & 	3 & 	1 & 	48 & 	0.02 
             & 	3 & 	3 & 	1 & 	48 & 	0 & 	0.02 \\ \hline
Mesi & 	3 & 	2 & 	0 & 	9 & 	0.00 
     & 	3 & 	2 & 	0 & 	9 & 	0 & 	0.00 \\ \hline
Mesi\_ca & 	3 & 	2 & 	0 & 	13 & 	0.00 
         & 	3 & 	2 & 	0 & 	13 & 	0 & 	0.00 \\ \hline
Moesi & 	3 & 	2 & 	0 & 	10 & 	0.01 
      & 	3 & 	2 & 	0 & 	10 & 	0 & 	0.01 \\ \hline
Moesi\_ca & 	3 & 	2 & 	0 & 	13 & 	0.00 
          & 	3 & 	2 & 	0 & 	13 & 	0 & 	0.00 \\ \hline
Synapse & 	2 & 	1 & 	0 & 	16 & 	0.01 
        & 	2 & 	1 & 	0 & 	16 & 	0 & 	0.01 \\ \hline
Xerox P.D. & 	7 & 	13 & 	0 & 	388 & 	0.23 
           & 	7 & 	13 & 	0 & 	388 & 	0 & 	0.23 \\ \hline\hline
\end{tabular}
\end{table}

\begin{table}[b]
\small
\caption{\label{subapp:ip}Imperative Programs}
\begin{tabular}{||l|l|l|l|l|l||l|l|l|l|l|l||}
\hline\hline
& \multicolumn{5}{|c||}{Default Setting} & 
  \multicolumn{6}{|c||}{Best Setting} \\ \hline
Problem  & d & \#n & \#d & \#SMT & time 
         & d & \#n & \#d & \#SMT & \#inv. & time \\ \hline\hline
Find & 	4 & 	27 & 	7 & 	691 & 	0.90 
     & 	4 & 	27 & 	7 & 	691 & 	0 & 	0.90 \\ \hline
Max\_in\_Array & 	- & 	- & 	- & 	- & 	timeout 
               & 	2 & 	1 & 	1 & 	46 & 	5 & 	0.03 \\ \hline
Selection\_Sort & 	- & 	- & 	- & 	- & 	timeout 
                & 	5 & 	13 & 	2 & 	1141 & 	11 & 	0.62 \\ \hline
Strcat & 	- & 	- & 	- & 	- & 	timeout
       & 	2 & 	2 & 	2 & 	80 & 	2 & 	0.07 \\ \hline
Strcmp & 	- & 	- & 	- & 	- & 	timeout 
       & 	2 & 	1 & 	1 & 	21 & 	3 & 	0.01 \\ \hline
Strcopy & 	3 & 	3 & 	1 & 	694 & 	1.22 
        & 	3 & 	3 & 	2 & 	564 & 	4 & 	0.38 \\ \hline\hline
\end{tabular}
\end{table}

\begin{table}[t]
\small
\caption{\label{subapp:misc}Miscellanea}
\begin{tabular}{||l|l|l|l|l|l||l|l|l|l|l|l||}
\hline\hline
& \multicolumn{5}{|c||}{Default Setting} & 
  \multicolumn{6}{|c||}{Best Setting} \\ \hline
Problem  & d & \#n & \#d & \#SMT & time 
         & d & \#n & \#d & \#SMT & \#inv. & time \\ \hline\hline
Alternating\_bit & -	 &- 	 &- 	 &- 	 & timeout 	  
                 & 	21 & 	1008 & 	156 & 	41894 & 	1 & 	44.48 \\ \hline
Bakery & 	6 & 	12 & 	0 & 	86 & 	0.04 
       & 	6 & 	12 & 	0 & 	86 & 	0 & 	0.04 \\ \hline
Bakery2 & 	6 & 	22 & 	1 & 	247 & 	0.07 
        & 	6 & 	22 & 	1 & 	247 & 	0 & 	0.07 \\ \hline
Controller & 	6 & 	8 & 	0 & 	95 & 	0.03 
           & 	6 & 	8 & 	0 & 	95 & 	0 & 	0.03 \\ \hline
Csm & 	- & 	- & 	- & 	- & 	timeout 
    & 	2 & 	2 & 	2 & 	76 & 	1 & 	0.02 \\ \hline
Filter\_simple & 	- & 	- & 	- & 	- & 	timeout 
               & 	2 & 	4 & 	4 & 	1013 & 	132 & 	3.94 \\ \hline
Fischer & 	10 & 	16 & 	2 & 	336 & 	0.16 
        & 	10 & 	16 & 	2 & 	336 & 	0 & 	0.16 \\ \hline
Fischer\_U & 	8 & 	13 & 	3 & 	198 & 	0.08
           & 	8 & 	13 & 	3 & 	198 & 	0 & 	0.08 \\ \hline
German & 	26 & 	2642 & 	678 & 	157870 & 	191.39 
       & 	26 & 	2642 & 	678 & 	157870 & 	0 & 	191.39 \\ \hline
Ins\_sort & 	- & 	- & 	- & 	- & 	timeout 
          & 	2 & 	2 & 	1 & 	40 & 	1 & 	0.04 \\ \hline
MIS & -	 & -	 &- 	 &- 	 & timeout 
    & 	1 & 	0 & 	0 & 	1261 & 	95 & 	0.85 \\ \hline
Mux\_Sem & 	7 & 	15 & 	0 & 	174 & 	0.04 
         & 	7 & 	15 & 	0 & 	174 & 	0 & 	0.04 \\ \hline
Mux\_Sem\_param & 	4 & 	5 & 	0 & 	85 & 	0.04 
                & 	2 & 	3 & 	1 & 	57 & 	4 & 	0.02 \\ \hline
Order & 	3 & 	3 & 	0 & 	18 & 	0.01 
      & 	2 & 	2 & 	2 & 	16 & 	2 & 	0.01 \\ \hline
Simple & 	2 & 	1 & 	0 & 	10 & 	0.00 
       & 	2 & 	1 & 	0 & 	10 & 	0 & 	0.00 \\ \hline
Swimming\_Pool & 	3 & 	81 & 	0 & 	1300 & 	0.67 
               & 	3 & 	62 & 	3 & 	927 & 	0 & 	0.73 \\ \hline
Szymanski+ & 	21 & 	685 & 	102 & 	43236 & 	47.00
           & 	2 & 	1 & 	1 & 	90 & 	2 & 	0.04 \\ \hline
Ticket\_o & 	- & 	- & 	- & 	- & 	timeout 
          & 	3 & 	4 & 	2 & 	201 & 	10 & 	0.06 \\ \hline
Token\_Ring & 	3 & 	2 & 	0 & 	30 & 	0.02 
            & 	3 & 	2 & 	0 & 	30 & 	0 & 	0.02 \\ \hline
Tricky & 	8 & 	7 & 	0 & 	22 & 	0.02 
       & 	2 & 	1 & 	1 & 	13 & 	1 & 	0.00 \\ \hline
Two\_Semaphores & 	4 & 	5 & 	1 & 	48 & 	0.02 
                & 	4 & 	5 & 	1 & 	48 & 	0 & 	0.02 \\ \hline\hline
\end{tabular}
\end{table}

As it is apparent by taking a look at the Tables, gaining some
expertise in using the available options of the tool may give dramatic
improvements in performances, either in terms of reduced timings or in
getting the system to terminate.  For the category ``Mutual exclusion
protocols,'' invariant synthesis is helpful to reduce the solving time
for the larger examples.  For the category ``Cache coherence
protocols,'' the effect of invariant synthesis as well as other
techniques is negligible.  For the category ``Imperative programs,''
invariant synthesis techniques are the key to make the tool terminate
on almost all problems.  In particular, signature abstraction,
introduced in this last version of the tool, is a crucial ingredient.

A comparative analysis is somewhat difficult in lack of a standard for
the specifications of safety problems.  This situation is similar to
the experimental evaluation of SMT solvers before the introduction of
the SMT-LIB standard~\cite{smtlib}.  It would be interesting to
investigate if the proposed format can become the new interlingua for
infinite state model checkers so that exchange of problems becomes
possible as well as the fair comparison of performances.  Just to give
an idea of the relative performance of our tool, we only mention that
\textsc{mcmt} performs better or outperforms (on the largest
benchmarks) the model checkers PFS and Undip on the problems taken
from their distributions.  In addition, these two systems are not
capable of handling many of the problems considered here such as those
listed in the category ``Imperative Programs'' (their input syntax and
the theoretical framework they are based on are too restrictive to
accept them).

\section{Discussion}
\label{sec:discussion}
We have given a comprehensive account of our approach to the model
checking of safety properties of infinite state systems manipulating
array variables by SMT solving.  The idea of using arrays to represent
system states is not new in model-checking (see in
particular~\cite{L2,L1}); what seems to be new in our approach is the
fully declarative characterization of \emph{both} the topology and the
(local) data structures of systems by using theories.  This has two
advantages.  First, implementations of our approach can handle a wide
range of topologies without modifying the underlying data structures
representing sets of states.  This is in contrast with recently
developed techniques~\cite{tacas06,cav06} for the uniform verification
of parametrized systems, which consist in exploring the state space of
a system by using a finitary representation of (infinite) sets of
states and require substantial modifications in the computation of the
pre-image to adapt to different topologies.  Second, since SMT solvers
are capable of handling several theories in combinations, we can avoid
encoding everything in one theory, which has already been proved
detrimental to performances in~\cite{bultan-gerber-pugh,composite-mc}.
SMT techniques were already employed in
model-checking~\cite{demoura-bounded,armando-boundedmc-and-smt}, but
only in the bounded case (whose aim is mostly limited at finding bugs,
not at full verification).

In more details, our contributions are the following.  First, we have
explained how to use certain classes of first-order formulae to
represent sets of states and identified the requirements to mechanize
a fully symbolic and declarative version of backward reachability.
Second, we have discussed sufficient conditions for the termination of
the procedure on the theories used to specify the topology (indexes)
and the data (elements) manipulated by the array-based system.  Third,
we have argued that the classes of formulae allow us to specify a
variety of parametrized and distributed systems, and imperative
algorithms manipulating arrays.  Finally, we have studied invariant
synthesis techniques and their integration in the backward
reachability procedure.  Theoretically, we have given sufficient
conditions for the completeness on the theories of indexes and
elements of the array-based system.  Pragmatically, we have described
how to interleave invariant guessing and backward reachability so as
to ameliorate the termination of the latter.  We have implemented the
proposed techniques in \textsc{mcmt} and evaluated their viability on
several benchmark problems extracted from a variety of sources.  The
experimental results have confirmed the efficiency and flexibility of
our approach.

\subsection{Related work}

We now discuss the main differences and similarities with existing
approaches to the verification of safety properties of infinite state
systems.  We believe it is convenient to recall two distinct and
complementary approaches among the many possible alternatives
available in the literature.  In examining related works, we do not
attempt to be exhaustive (we consider this an almost desperate task
given the huge amount of work in this area) but rather to position our
approach with respect to
some of the main lines of research in the field.

The first approach is pioneered in~\cite{lics} and its main notion is
that of well-structured system.  For example, it was implemented in
two systems (see, e.g.,~\cite{tacas06,cav06}),
which were able to automatically verify several protocols for mutual
exclusion and cache coherence.  One of the key ingredients to the
success of these tools is their capability to perform accurate
fix-point checks so as to reduce the number of iterations of the
backward search procedure.  A fix-point check is implemented by
`embedding' an old configuration (i.e.\ a finite representation of a
potentially infinite set of states) into a newly computed pre-image;
if this is the case, then the new pre-image is considered
``redundant'' (i.e., not contributing new information about the set of
backward reachable states) and thus can be discarded without loss of
precision.  Indeed, the exhaustive enumeration of embeddings has a
high computational cost.  An additional problem is that constraints
are only used to represent the data manipulated by the system while
its topology is encoded by \emph{ad hoc} data structures.  This
requires to implement from scratch algorithms both to compute
pre-images and embeddings, each time the topology of the systems to
verify is modified.  On the contrary, \textsc{mcmt} uses particular
classes of \emph{first-order formulae} to represent configurations
parametrised with respect to a theory of the data and a theory of the
topology of the system so that pre-image computation reduces to a
fixed set of logical manipulations and fix-point checking to solve SMT
problems containing universally quantified variables.  To mechanize
these tests, a quantifier-instantiation procedure is used, which is
the logical counterpart of the enumeration of ``embeddings.''
Interestingly, this notion of ``embedding'' can be recaptured via
classical model theory (see~\cite{ijcar08}  or Section~\ref{sec:term} above) in the logical framework
underlying \textsc{mcmt}, a fact that allows us to import into our
setting the decidability results of~\cite{lics} for backward
reachability.  Another important advantage of the approach underlying
\textsc{mcmt} over that proposed in~\cite{lics} is its broader scope
of applications with respect to the implementations
in~\cite{tacas06,cav06,vmcai08}.  The use of theories for specifying
the data and the topology allows one to model disparate classes of
systems in a natural way.  Furthermore, even if the quantifier
instantiation procedure becomes incomplete with rich theories, it can
soundly be used and may still permit to prove the safety of a system.
In fact, \textsc{mcmt} has been successfully employed to verify
sequential programs (such as sorting
algorithms) 
that are far beyond the reach of the systems described
in~\cite{tacas06,cav06}.

The second and complementary approach to model checking infinite state
system relies on \emph{predicate abstraction} techniques,
initially proposed in~\cite{seminal}.  The idea is to abstract the
system to one with finite states, to perform finite-state model
checking, and to refine spurious traces (if any) by using decision
procedures or SMT solvers.  This technique has been implemented in
several tools 
and is often combined with interpolation algorithms for the refinement
phase.  As pointed out in~\cite{qaaderflanagan,indexedabs}, predicate
abstraction must be carefully adapted when (universal) quantification
is used to specify the transitions of the system or its properties, as
it is the case for the problems tackled by \textsc{mcmt}.  The are two
crucial problems to be solved.  The first is to find an appropriate
set of predicates to compute the abstraction of the system.  In fact,
besides system variables, universally quantified variables may also
occur in the system.  The second problem is that the computation of
the abstraction as well as its refinement require to solve proof
obligations containing universal quantifiers.  Hence, we need to
perform suitable quantifier instantiation in order to enable the use
of decision procedures or SMT solving techniques for quantifier-free
formulae.  The first problem is solved by
Skolemization~\cite{qaaderflanagan} or fixing the number of variables
in the system~\cite{indexedabs} so that standard predicate abstraction
techniques can still be used.  The second problem is solved by
adopting very straightforward (sometimes naive) and incomplete
quantifier instantiation procedures.  While being computationally
cheap and easy to implement, the heuristics used for quantifier
instantiation are largely imprecise and does not permit the detection
of redundancies due to variable permutations, internal symmetries, and
so on.  Experiments performed with \textsc{mcmt}, tuned to mimic these
simple instantiation strategies, show much poorer performances.
We believe that the reasons of success of the predicate abstraction
techniques
in~\cite{qaaderflanagan,indexedabs} 
lie in the clever heuristics used to find and refine the set of
predicates for the abstraction.  The current implementation of
\textsc{mcmt} is orthogonal to the predicate abstraction approach; it
features an extensive quantifier instantiation (which is complete for
the theories over the indexes satisfying the Hypothesis (I) from
Theorem~\ref{th:decidability} and is enhanced with completeness
preserving heuristics to avoid useless instances), but it performs
only a primitive form of predicate abstraction, called signature
abstraction (see Section~\ref{subsec:abs-inv}).  Another big
difference is how abstraction is used in \textsc{mcmt}: the set of
backward reachable states is always computed precisely while
abstraction is only exploited for guessing candidate invariants which
are then used to prune the set of backward reachable states.  Since we
represent sets of states by 
formulae, guessing and then using the synthesized invariants turns out
to be extremely easy, thereby helping to solve the tension between
model checking and deductive techniques that has been discussed a lot
in the literature and is still problematic in the tools described
in~\cite{tacas06,cav06}
where sets of states are represented by ad hoc data structures.

Besides the two main approaches mentioned above, there is a third line
of research in the area that applied constraint solving techniques to
the model-checking of infinite state systems.  One of the first
attempts was described in~\cite{bultan-gerber-pugh} and then furtherly
studied in~\cite{composite-mc}.  The idea was to use composite
constraint domains (such as integers and Booleans) to encode the data
and the control flow of, for example, instances of parametrised
systems.  Compared to our framework, the verification methods
in~\cite{bultan-gerber-pugh,composite-mc} are not capable of checking
safety regardless of the number of process in a system but only
supports the verification of its instances.  Indeed, increasing the
number of processes quickly degrades performances.  Babylon is a tool
for the verification of counting abstractions of parametrized systems
(e.g., multithreaded Java programs~\cite{babylon}). It uses a
graph-based data structure to encode disjunctive normal forms of
integer arithmetic constraints.  Computing pre-images requires
computationally expensive normalization, which is not needed for us as
SMT solvers efficiently handle arbitrary integer constraints.  Brain
is a model-checker for transition systems with finitely many integer
variables which uses an incremental version of Hilbert's bases to
efficiently perform entailment or satisfiability checking of integer
constraints (the results reported in~\cite{voronkov-cav02} shows that
it scales very well).  Taking $T_I$ to be an enumerated data-type
theory, the notion of array-based systems considered in this paper
reduce to those used by Brain.  However, many of the systems that can
be modelled as array-based systems cannot be handled by Brain.
Another interesting proposal to uniform verification of parametrized
systems using constraint solving techniques is~\cite{bouajjani-rew},
where a decidability result for $\Sigma^0_2$-formulae is derived
(these are $\exists\forall$-formulae roughly corresponding to those
covered by Theorem~\ref{th:decidability} above, for the special case
in which the models of the theory $T_I$ are all the finite linear
orders).  While the representation of states in~\cite{bouajjani-rew}
is (fully) declarative, transitions are not, as a rewriting semantics
(with constraints) is employed.  Since transitions are not
declaratively handled, the task of proving pre-image closure becomes
non trivial;
in~\cite{bouajjani-rew}, pre-image closure of $\Sigma^0_2$-formulae
under transitions encoded by $\Sigma^0_2$-formulae ensures the
effectiveness of the tests for inductive invariant and bounded
reachability analysis, but not for fix-point checks.  In our approach,
an easy (but orthogonal) pre-image closure result for existential
state descriptions (under certain $\Sigma^0_2$-formulae representing
transitions) gives the effectiveness of fix-point checks, thus
allowing implementation of backward search.

\subsection{Future work}

We envisage to develop the work described here in three directions.
First, we plan to enhance the implementation of the signature
abstraction technique in future releases of \textsc{mcmt}.  The idea
is to find the best trade-off between the advantages of predicate
abstraction and extensive quantifier instantiation.  Another aspect is
the design of methods for the dynamic refinement of the abstraction
along the lines of the counter-example-guided-refinement (CEGAR)
loop~\cite{seminal}.
A complementary approach could be to use techniques for the automatic
discovery of relationships among values of array elements developed in
abstract interpretation (see, e.g.,~\cite{popl05}).  Second, we want
to perform more extensive experiments for different classes of
systems.  For example, we have already started to investigate
parametrised timed automata (introduced in~\cite{AbdullaTCS}) with
\textsc{mcmt} and found encouraging preliminary results~\cite{verify}.
Another class of problems in which successful experiments have been
performed with \textsc{mcmt} concerns the verification of
fault-tolerant distributed
algorithms~\cite{Francesco1,Francesco2}.
The third line of future research consists of in exploring further and
then implementing the verification method for a sub-class of liveness
properties of array-based systems sketched in~\cite{ijcar08}.

\subsection*{Acknowledgments} The authors would like to thanks the
anonymous referees for the useful remarks that helped to improve the
clarity of the paper.  

The second author was partially supported by the FP7-ICT-2007-1
Project no.\ 216471, ``AVANTSSAR: Automated Validation of Trust and
Security of Service-oriented Architectures''
(\texttt{\url{www.avantssar.eu}}) and by the Incoming Team
Project ``SIAM: Automated Security Analysis of Identity and Access
Management Systems,'' funded by Provincia Autonoma di Trento in the
context of the COFUND action of the European Commission (FP7).

\bibliographystyle{plain} 
\bibliography{smtmc} 

\newpage

\appendix


\section{Omitted Proofs}
\label{app:proofs}

\newcounter{eqsave}

\subsection*{Decidability of restricted satisfiability checking}

The following result is a simple generalization of
Theorem~\ref{th:decidability} (of Section~\ref{subsec:brs}).

\begin{thm}
  \label{thm:pi01}
  The $A^E_I$-satisfiability of a sentence of the kind
  \begin{equation}\label{eq:atest1}
    \exists a_1\cdots\exists a_n \; \exists \ui\; \exists \ue\;\forall \uj \; \psi(\ui, \uj, \ue, a_1[\ui], \dots, a_n[\ui], a_1[\uj], \dots, a_n[\uj])
  \end{equation}
  is decidable. Moreover, the following conditions are equivalent:
  \begin{enumerate}[{\rm (i)}]
  \item\label{ite:pi01_1} the sentence \eqref{eq:atest1} is $A^E_I$-satisfiable;
  \item\label{ite:pi01_2} the sentence \eqref{eq:atest1} is satisfiable in a finite index model of $A^E_I$;
  \item\label{ite:pi01_3} the sentence 
    \begin{equation}
      \exists a_1\cdots\exists a_n \; \exists \ui\; \exists \ue \bigwedge_{\sigma}\; \psi(\ui, \uj\sigma, \ue, a_1[\ui], \dots, a_n[\ui], a_1[\uj\sigma], \dots, a_n[\uj\sigma])
    \end{equation} 
    is $A^E_I$-satisfiable (here $\sigma$ ranges onto the
    substitutions mapping the variables $\uj$ into the set of
    representative $\Sigma_I(\ui)$-terms).
\end{enumerate}
\end{thm}

\begin{proof}
  In order to avoid difficulties with the notation, we consider the
  case where $n=1$ only (the reader may check that there is no loss of
  generality in that).\footnote{Since existentially quantifying over
    variables that do not occur in the formula does not affect
    satisfiability, we can also assume that the tuple $\ui$ is not
    empty (this observation is needed if we want to prevent the
    structure $\cN$ defined below from having empty index domain).} We
  first show that the $A^E_I$-satisfiability of
  \begin{equation}\label{eq:stest1}
    \exists a\; \exists \ui\; \exists \ue\;\forall \uj \; \psi(\ui, \uj,\ue, a[\ui],a[\uj])
  \end{equation}
  is equivalent to the $A^E_I$-satisfiability of
  \begin{equation}\label{eq:stest2}
    \exists a\; \exists \ui\;\exists \ue\; \bigwedge_{\sigma}\; \psi(\ui, \uj\sigma, \ue, a[\ui],a[\uj\sigma])
  \end{equation}
  where $\sigma$ ranges onto the substitutions mapping the variables
  $\uj$ into the set of representative $\Sigma_I(\ui)$-terms.

  That $A^E_I$-satisfiability of \eqref{eq:stest1} implies
  $A^E_I$-satisfiability of \eqref{eq:stest2} follows from trivial
  logical manipulations, so let's assume $A^E_I$-satisfiability of
  \eqref{eq:stest2} and show $A^E_I$-satisfiability of
  \eqref{eq:stest1}. Let $\cM$ be a model of \eqref{eq:stest2}; we can
  assign elements in this model to the variables $a, \ui, \ue$ in such
  a way that (under such an assignment ${\tt a}$) we have $\cM, {\tt
    a}\models \bigwedge_{\sigma}\; \psi(\ui, \uj\sigma, \ue,
  a[\ui],a[\uj\sigma])$. Consider the model $\cN$ which is obtained
  from $\cM$ by restricting the interpretation of the sort {\tt INDEX}
  (and of all function and relation symbols for indexes) to the
  $\Sigma_I$-substructure generated by the elements assigned by ${\tt
    a}$ to the $\ui$: since models of $T_I$ are closed under
  substructures, this substructure is a model of $T_I$ and
  consequently $\cN$ is still a model of $A^E_I$. Now let $s$ be the
  restriction of ${\tt a}(a)$ to the new smaller index domain and let
  $\tilde {\tt a}$ be the assignment differing from ${\tt a}$ only for
  assigning $s$ to $a$ (instead of ${\tt a}(a)$); since $\psi$ is
  quantifier free and since, varying $\sigma$, the elements assigned
  to the terms $\uj\sigma$ covers all possible $\uj$-tuples of
  elements in the interpretation of the sort {\tt INDEX} in $\cN$, we
  have $\cN, \tilde {\tt a}\models \forall \uj \psi(\ui, \uj, a[\ui],
  a[\uj])$. This shows that $\cN \models \exists a\; \exists \ui\;
  \forall \uj \; \psi(\ui, \uj, a[\ui],a[\uj])$, i.e that
  \eqref{eq:stest1} holds. Notice that $\cN$ is a finite index
  model,\footnote{This is because $T_I$ is locally finite and the
    $\Sigma_I$ reduct of $\cN$ is a structure which is generated by
    finitely many elements.} hence we proved also the equivalence
  between \eqref{ite:pi01_1} and \eqref{ite:pi01_2}.
 
  We now need to decide $A^E_I$-satisfiability of sentences
  \eqref{eq:stest2}. Let $\ut$ be the representative
  $\Sigma(\ui)$-terms and let us put them in bijective correspondence
  with fresh variables $\ul$ of sort {\tt INDEX}; let
  $\psi_{\sigma}(\ui, \ul, \ue, a[\ui], a[\ul])$ be the formula
  obtained by replacing in $\psi(\ui, \uj\sigma, \ue,
  a[\ui],a[\uj\sigma])$ the $\Sigma(\ui)$-terms $\uj\sigma$ by the
  corresponding $\ul$. We first rewrite \eqref{eq:stest2} as
  \begin{equation}\label{eq:stest3}
    \exists a\; \exists \ui\;\exists \ue\; \exists \ul\;(\ul = \ut \wedge\bigwedge_{\sigma} \psi_{\sigma}(\ui, \ul, \ue, a[\ui],a[\ul]))
  \end{equation}
  (here $\ul = \ut$ means component-wise equality, expressed as a
  conjunction).

  Notice that $T_I$ and $T_E$ are disjoint (they do not have even any
  sort in common), which means that $ \ul = \ut
  \wedge\bigwedge_{\sigma} \psi_{\sigma}(\ui, \ul, \ue,
  a[\ui],a[\ul])$ is a Boolean combination of $\Sigma_I$-atoms and of
  $\Sigma_E$-atoms (in the latter kind of atoms, the variables for
  elements are replaced by the terms $a[\ui], a[\ul]$). This means
  that our decision problem can be further rephrased in terms of the
  problem of deciding for $A^E_I$-satisfiability formulae like
  \begin{equation}\label{eq:conn}
    \psi_I(\uj) \wedge a[\uj]=\ud \wedge \psi_E(\ud,\ue)
  \end{equation}
  where $\psi_I(\uj)$ is a conjunction of $\Sigma_I$-literals and
  $\psi_E(\ud,\ue)$ is a conjunction of $\Sigma_E$-literals.
 
  Since we are looking for a model of $T_I$, a model of $T_E$ and for
  a function connecting their domains (the function interpreting the
  variable $a$), this is a satisfiability problem for a theory
  connection (in the sense of \cite{BaGh}):\footnote{Strictly
    speaking, one cannot directly apply the results from~\cite{BaGh},
    because in this paper we have adopted a `semantic' notion of
    theory characterized by a class $\mathcal{C}$ of models.  In other
    words, the class $\mathcal{C}$ of models is not required to be
    elementary.} since the signatures of $T_I, T_E$ are disjoint, the
  problem is decided by propagating equalities.\footnote{This is
    different from the standard Nelson-Oppen combination, where also
    inequalities must be propagated.}  Hence, to
  decide~\eqref{eq:conn}, it is sufficient to apply the following
  steps:
  \begin{enumerate}[$-$]
  \item guess an equivalence relation $\Pi$ on the index variables
    $\uj$ (let's assume $\uj=j_1, \dots. j_n$);
  \item check $\psi_I(\uj) \cup \lbrace j_k=j_l \mid (j_k, j_l)\in
    \Pi\rbrace\cup \lbrace j_k\not= j_l \mid (j_k, j_l)\not \in
    \Pi\rbrace$ for $T_I$-satisfiability;
  \item check $\psi_E(\ud,\ue) \cup \lbrace d_k=d_l \mid (j_k,
    j_l)\in \Pi\rbrace$ for $T_E$-satisfiability;
  \item return `unsatisfiable' iff failure is reported in the
    previous two steps for all possible guesses.
  \end{enumerate}
  Soundness and completeness of the above procedure are easy.
\end{proof}

\subsection*{Undecidability of backward reachability}

Here, we give the \emph{proof of Theorem~\ref{th:undecidability}} (of
Section~\ref{subsec:undec}).
\begin{proof}
 A \emph{two registers Minsky machine} is a finite set $\bf P$ of
  instructions (also called a program) for manipulating configurations
  seen as triples $(q,m,n)$ where $m,n$ are natural numbers
  representing the registers content and $q$ represents the machine
  location state ($q$ varies on a fixed finite set $Q$).  There are
  four possible kinds of instructions, inducing transformations on the
  configurations as explained in Table~\ref{tab:Minsky}.
\begin{table}[ht]
  \begin{center}
  \begin{tabular}{|c|c|c|}
    \hline
    N.  & Instruction         & Transformation\\ \hline
    I   & $q\to(r,1,0)$       & $(q, m,n) \to (r, m+1, n)$ \\ \hline
    II  & $q\to(r,0,1)$       & $(q, m,n) \to (r, m, n+1)$ \\ \hline
    III & $q\to(r,-1,0)[r']$  & {\tt if} $m\neq 0$ {\tt then} $(q, m,n) \to (r, m-1, n)$ \\
        &                     & {\tt else} $(q, m,n) \to (r', m, n)$ \hspace{21.5mm} \\ \hline
    IV  & $q\to(r,0,-1)[r']$  & {\tt if} $n\neq 0$ {\tt then} $(q, m,n) \to (r, m, n-1)$ \\
        &                     & {\tt else} $(q, m,n) \to (r', m, n)$ \hspace{20.5mm} \\
        \hline
   \end{tabular}
   \end{center}
   \caption{\label{tab:Minsky} Instructions and related transformations for (two-registers) Minsky Machines}
\end{table}
A $\bf P$-transformation is a transformation induced by an instruction
of $\bf P$ on a certain configuration.  For a Minsky machine $\bf P$,
we write $(q, m, n) \rightarrow^{\star}_{\bf P} (q', m', n')$ to say
that it is possible to reach configuration $(q', m', n')$ from $(q, m,
n)$ by applying finitely many $\bf P$-transformations.  Given a Minsky
machine $\bf P$ and an initial configuration $(q_0, m_0, n_0)$, the
problem of checking whether a configuration $(q', m', n')$ is
reachable from $(q_0, m_0, n_0)$ (i.e., if $(q_0, m_0,
n_0)\rightarrow^{\star}_{\bf P} (q', m', n')$ holds or not) is called
\emph{the (second) reachability (configuration) problem}.  It is
well-known\footnote{For details and further references, see for
  instance~\cite{CZ}. 
} that there exists a (two-register)
Minsky machine $\bf P$ and a configuration $(q_0, m_0, n_0)$ such that
the second reachability configuration problem is undecidable. To
simplify the matter, we assume that $m_0=0$ and $n_0=0$: there is no
loss of generality in that, because one can add to the program $\bf P$
more states and instructions (precisely $m_0+n_0$ further states and
instructions of type I-II) for the initialization to $m_0, n_0$.

We build a locally finite array-based system $\cSi_{\bf P}=(a, I_{\bf
  P}, \tau_{\bf P})$ and an $\exists^I$-formula $U_{q,m,n}$ such that
$\cSi$ is unsafe w.r.t. $U_{q,m,n}$ iff the machine $\bf P$ reaches
the configuration $(q, m, n)$.  We take as $\Sigma_I$ the signature
having two constants $o, o'$ and a binary relation $S$; models of
$T_I$ are the $\Sigma_I$-structures satisfying the axioms
\begin{eqnarray*}
  &\forall i \,\neg S(i,o),~~~~~~~~~~~~~ &
  \forall i\,\forall j_1\,\forall j_2\,(S(i,j_1)\wedge S(i,j_2)\to j_1=j_2),
  \\
  &S(o,o'),~~~~~~~~~~~~~~~~~&
  \forall i_1\,\forall i_2\,\forall j\,(S(i_1,j)\wedge S(i_2,j)\to i_1=i_2),
\end{eqnarray*}
saying that $S$ is a an injective partial function having $o$ in the
domain but not in the range.  As $\Sigma_E$ we take the enumerated
datatype theory relative to the finite set $Q\times\{0,1\}\times\{0,
1\}$.  Notice that $T_I, T_E$ are both locally finite; in addition,
$T_I$ is closed under substructures and $T_E$ has quantifier
elimination.

The idea is that of encoding a configuration $(q,m,n)$ as any
configuration $s$ (in the formal sense of Subsection~\ref{subsec:configurations}) satisfying the following
conditions:
\begin{enumerate}[(i)]
 \item the support of $s_I$ contains a substructure of the kind
   $$
   o=i_0\to_S o'=i_1\to_S i_2\cdots \to_S i_K
   $$
   for some $K> m,n$ (we write $i \to_S j$ to means that $(i,j)$ is in
   the interpretation of the relational symbol $S$ in $s_I$).
 \item for all $i$ in the support of $s_I$, if
   $s(i)=\langle r, u,v\rangle$ then (a) $r=q$; (b) $u=1$ iff $i=i_k$
   for  a $k\leq m$; (c) $v=1$ iff $i=i_k$ for a $k\leq n$.
\end{enumerate}
In case the above conditions (i)-(ii) hold, we say that \emph{$s$
  bi-simulates $(q, m,n)$}.

The initial formula $I$ is $$\forall i ((i\not = o\wedge a[i]=\langle
q_0, 0, 0\rangle) \vee (i=o \wedge a[i]=\langle q_0, 1, 1\rangle)).$$
Clearly for every model $\cM$ and for every $s\in
\mathtt{ARRAY}^{\cM}$, the following happens:
\begin{description}
\item[{\rm ($\alpha$)}] $\cM\models I(s)$ iff $s$ bi-simulates the
  initial machine configuration $(q_0,0,0)$.
\end{description}

We write the transition $\tau$ in such a way that for every model
$\cM$ and for every $s,s'\in \mathtt{ARRAY}^{\cM}$, the following
happens:
\begin{description}
\item[{\rm ($\beta$)}] if $s$ bi-simulates $(q,m,n)$, then $\cM\models
  \tau(s,s')$ iff there is $(q', m', n')$ such that $s'$ bi-simulates
  $(q', m', n')$ and $(q,m,n)\to_{\bf P} (q', m', n')$.
\end{description}
This goal is obtained by taking $\tau$ to be a disjunction of
$T$-formulae corresponding to the instructions for $\bf P$. The
$T$-formula corresponding to the first kind of instructions
$q\to(r,1,0)$ is the following:\footnote{For simplicity,
we assume that the
  signature $\Sigma_E$ is 4-sorted and endowed with the three
  projection functions $pr_1, pr_2, pr_3$ mapping a data $\langle r, u,v\rangle$ to $r, u,v$, respectively:  
there is no need of this assumption,
  but without it specifications become cumbersome.}
\begin{eqnarray*}
  \exists i_1\,\exists i_2\, \exists i_3\,
  & (S(i_1, i_2) \wedge S(i_2,i_3)\wedge pr_1(a[i_1])= q\wedge  \, pr_2( a[i_1])=1 \wedge
   \\ &
  \wedge~ pr_2(a[i_2])= 0 \wedge
  pr_2(a[i_3])=0\wedge a'=\lambda j\, F)
\end{eqnarray*}
where
\begin{eqnarray*}
  F & := & \mathtt{if} ~(j=i_2)~ \mathtt{then}~ \langle r, 1, pr_3(a[j])\rangle
  \\ 
  && \mathtt{else} ~\langle r, pr_2(a[j]), pr_3(a[j])\rangle
\end{eqnarray*}
Instructions $q\to(r,-1,0)[r']$ of the kind (III) are simulated by the
following $T$-formula
\begin{eqnarray*}
  \exists i_1\,\exists i_2\, (S(i_1, i_2) \wedge pr_1(a[i_1])= q\wedge
  pr_2( a[i_1])=1 \wedge
  pr_2(a[i_2])= 0)
\end{eqnarray*}
where
\begin{eqnarray*}
  F & := & \mathtt{if}~ (i_1\not=o \wedge j=i_1) ~\mathtt{then}~ \langle r, 0, pr_3(a[j])\rangle \\
  &&\mathtt{else~ if}~ (i_1\not=o \wedge j\not=i_1) ~\mathtt{then}~\langle r, pr_2(a[j]), pr_3(a[j])\rangle \\
  &&\mathtt{else}~ \langle r', pr_2(a[j]), pr_3(a[j])\rangle
\end{eqnarray*}
$T$-formulae for instructions of kind (II) and (IV) are defined
accordingly.

We write the unsafe states formula $U_{q,m,n}$ in such a way that for
every model $\cM$ and for every $s\in \mathtt{ARRAY}^{\cM}$, the
following happens:
\begin{description}
\item[{\rm ($\gamma$)}] if $\cM\models U_{q,m,n}(s)$ and $s$ bi-simulates some machine configuration,
  then it bi-simulates $(q,m,n)$.
\end{description}
This goal is achieved by taking $U_{q,m,n}$ to be the following
formula (suppose $m\geq n$, the case $n\leq m$ is symmetric):
\begin{eqnarray*}
  \exists i_0\cdots\exists i_{m+1}\,
  & (i_0=o \wedge \bigwedge_{0\leq k\leq m} S(i_k, i_{k+1}) \wedge
  \bigwedge_{0\leq k\leq n} a[i_k]=\langle q, 1,1\rangle \wedge \\
  &\wedge
  \bigwedge_{n< k\leq m} a[i_k]=\langle q, 1,0\rangle \wedge
  a[i_{m+1}]=\langle q, 0,0\rangle).~~~~~~~~~~~
\end{eqnarray*}
From ($\alpha$)-($\beta$)-($\gamma$) above it is clear that $\bf P$
reaches the configuration $(q,m,n)$ iff $\cSi$ is unsafe
w.r.t. $U_{q,m,n}$, so
that the latter is not decidable (for the left to right implication, take a
  run in a model with a large enough $S$-chain starting with $o$).
\end{proof}

\subsection*{Undecidability of unrestricted satisfiability checking}

We show that \emph{Hypothesis (I) cannot be removed from the statement
  of Theorem~\ref{th:decidability}} (and of Theorem~\ref{thm:pi01}).
We use 
a reduction to the reachability problem for Minsky machines as we have
done for the proof of the undecidability of the safety problem
(Theorem~\ref{th:undecidability}); the argument is similar to one used
in~\cite{arrays}.

Let $T_I$ be the theory 
having as a class of models the
natural numbers in the signature with just
zero, the successor function, and $\leq$.  Notice that this is not
locally finite.  Let $T_E$ be the theory having $Q\times \mathbb
N\times \mathbb N$ as a unique structure. Here $Q$ is like in the
previous subsection of this Appendix.  In the following, we freely use
projections, sums, numerals, subtraction, constants for elements of
$Q$, etc.  Formally, all these 
operations 
can be defined in many ways and the precise way is not
 relevant for the argument below.  In other words, we can avoid
to define precisely the signature $\Sigma_E$.  
This sloppiness is justified because we
must use a 
$T_I$ not satisfying the local finiteness requirement from Hypothesis (I) of
Theorem~\ref{th:decidability}, whereas we can use an arbitrary $T_E$.  
Let $\tau(a[j_1], a[j_2])$ abbreviate
the disjunction of the following formulae describing the
transformations from Table~\ref{tab:Minsky}:
\begin{eqnarray*}
&
pr_2(a[j_1])=q \wedge a[j_2] = \langle r, pr_2(a[j_1])+1, pr_3(a[j_1])\rangle
\\
&
pr_2(a[j_1])=q \wedge a[j_2] = \langle r, pr_2(a[j_1]), pr_3(a[j_1])+1\rangle
\\
&
pr_2(a[j_1])=q \wedge pr_2(a[j_1])>0 \wedge a[j_2] = \langle r, pr_2(a[j_1])-1, pr_3(a[j_1])\rangle
\\
&
pr_2(a[j_1])=q \wedge pr_2(a[j_1])=0 \wedge a[j_2] = \langle r', pr_2(a[j_1]), pr_3(a[j_1])\rangle
\\
&
pr_2(a[j_1])=q \wedge pr_3(a[j_1])>0 \wedge a[j_2] = \langle r, pr_2(a[j_1]), pr_3(a[j_1])-1\rangle
\\
&
pr_2(a[j_1])=q \wedge pr_3(a[j_1])=0 \wedge a[j_2] = \langle r', pr_2(a[j_1]), pr_3(a[j_1])\rangle
\end{eqnarray*}
Now consider the satisfiability of the following
$\exists^{A,I}\forall^I$-formula:
\begin{eqnarray*}
&
\exists a\;\exists i\; \exists j\; 
(i=0 \wedge a[i]=\langle q_0, 0, 0\rangle \wedge a[j]=\langle q, m, n\rangle \wedge 
\\
& \hskip 1.5cm \wedge
\forall j_1\; \forall j_2\; (j_1<j \wedge j_2=j_1+1 \to \tau(a[j_1],a[j_2])))
\end{eqnarray*}
Clearly, this is satisfiable iff the configuration $\langle q, m,
n\rangle$ is reachable: the array $a$ in fact stores the whole
computation leading to $\langle q, m, n\rangle$. 
Thus satisfiability of $\exists^{A,I}\forall^I$-formulae can be
undecidable if $T_I$ is not locally finite, even if the $SMT(T_I)$ and
the $SMT(T_E)$ problems are decidable (and even if $T_I$ is closed
under substructures).

A final observation is crucial.  If we keep local finiteness and drop
closure under substructures in the statement of Hypothesis (I) from
Theorem~\ref{th:decidability}, then the above counterexample still
applies!  In fact, the successor function for indexes is used only in
$j_2=j_1+1$ occurring in the formula above: we can replace the
application of the successor function with a binary relation $S(j_1,
j_2)$ so as to recover local finiteness.  However, closure under
substructures is dropped as the structure of natural numbers
has proper substructures if successor is a relation and not a function
and these substructures must be excluded for the above argument to
work (from satisfiability in such substructures a full computation
cannot be recovered). Thus, we can conclude that \emph{the two
  conditions of Hypothesis (I) are strictly connected and both
  needed}.

\end{document}